\documentclass[12pt,reqno]{amsart}
\RequirePackage{multirow}
\usepackage[showonlyrefs]{mathtools}
\usepackage{enumitem}
\usepackage{xifthen,setspace,multirow}
\usepackage{newunicodechar}
\newunicodechar{ℓ}{\ensurematnatbibh{\ell}}
\usepackage{diagbox,graphicx}
\usepackage{caption,subcaption}
\usepackage{float}
\usepackage[margin=3cm]{geometry}
\usepackage{bbm}
\usepackage{xcolor}
\usepackage{enumitem}
\usepackage{bm}
\usepackage[linesnumbered,ruled,vlined]{algorithm2e}
\usepackage{hyperref}
\usepackage{esvect}
\usepackage{amssymb,mathrsfs} 
\usepackage{lipsum}  
\usepackage{tabularx}
\usepackage{placeins}
\usepackage{comment}
\usepackage{multirow}
\newtheorem{thm}{Theorem}

\newtheorem{example}{Example}
\newtheorem{defi}{Definition}

\newtheorem{lem}{Lemma}
\newtheorem{assumption}{Assumption}
\newtheorem{cond}{Condition}

\newcommand{\vertiii}[1]{{\left\vert\kern-0.25ex\left\vert\kern-0.25ex\left\vert #1 
    \right\vert\kern-0.25ex\right\vert\kern-0.25ex\right\vert}}

\newcommand{\hp}{{\hat{\bm p}}}
\newcommand{\pb}{{\bm p(\beta)}}

\renewcommand{\leq}{\leqslant}
\renewcommand{\geq}{\geqslant}

\newcommand{\sech}{\mathrm{sech}}

\newcommand{\eps}{{\varepsilon}}

\newcommand{\e}{{\mathbb{E}}}

\newcommand{\E}{\mathbb{E}}

\newcommand{\R}{{\mathbb{R}}}

\newcommand{\p}{{\mathbb{P}}}

\allowdisplaybreaks

\title[Robust Estimation for Dependent Binary Network Data]{Robust Estimation for Dependent Binary Network Data}

\author[Liu]{Tianyu Liu} 
\address{Department of Statistics and Data Science, National Institute of Singapore, Singapore, {\tt tianyu.liu@u.nus.edu}}

\author[Mukherjee]{Somabha Mukherjee}
\address{Department of Statistics and Data Science, National Institute of Singapore, Singapore, {\tt  somabha@nus.edu.sg}}

\author[Ghosh]{Abhik Ghosh}
\address{Interdisciplinary Statistical Research Unit, Indian Statistical Institute, Kolkata, India, {\tt abhik.ghosh@isical.ac.in}}

\begin{document}
\begin{abstract}
    We consider the problem of learning the interaction strength between the nodes of a network based on dependent binary observations residing on these nodes, generated from a Markov Random Field (MRF). Since these observations can possibly be corrupted/noisy in larger networks in practice, it is important to robustly estimate the parameters of the underlying true MRF to account for such inherent contamination in observed data. However, it is well-known that classical likelihood and pseudolikelihood based approaches are highly sensitive to even a small amount of data contamination. So, in this paper, we propose a density power divergence (DPD) based robust generalization of the computationally efficient maximum pseudolikelihood (MPL) estimator of the interaction strength parameter, and derive its rate of consistency under the pure model. Along the way, we establish consistency and asymptotics for a class of
    general $Z$-estimators, covering our proposed DPD based estimators, 
    under flexible assumptions that hold for a substantial class of standard models. 
    To the best of our knowledge, these are the first central limit theorems for the class of 
    general $Z$-estimators in such settings.  
    Moreover, we show that the gross error sensitivities of the proposed DPD based estimators are significantly smaller than that of the MPL estimator, thereby theoretically justifying the greater (local) robustness of the former under contaminated settings. Finally, we demonstrate the superior (finite sample) performance of the DPD based variants over the traditional MPL estimator in a number of synthetically generated contaminated network datasets, and apply them to learn the network interaction strength in several real datasets from diverse domains of social science, neurobiology and genomics.  
\end{abstract}


\maketitle

\section{Introduction}

The rapid emergence of complex network data over the past few decades across a wide range of disciplines encompassing machine learning, biological and medical sciences, economics, marketing and finance industries, and social sciences, has created a pressing need for the development of computationally tractable algorithms for learning network and graphical models. Peer effects have long been recognized as influential in shaping decision-making across a variety of domains, such as voting behavior, consumer choices, recommender systems, online platforms, financial decision-making, and adolescent risk-taking. In spite of the presence of a rich body of literature in social sciences on peer influences, rigorously estimating the underlying mechanisms by which peer effects operate, remains a challenging and compelling problem in settings involving dependent network data. In many practical scenarios, one does not even have access to multiple samples over the same network under consideration and has to rely on just a single observation (on each node) for learning the underlying model. For instance, in an epidemic network at a specific time, it's not possible to have multiple independent outcomes from the same outbreak. Similarly, in an election, only one set of voting results is possible across a population network. In a criminal network, authorities usually have to rely on a single realization (the activities of criminals in a specific region) to predict future crimes.

In most practical settings, the outcomes associated with the nodes of a network are binary or discrete in nature and also dependent across the edges of the network. Classical statistical methods, such as logistic or multinomial regression, prove inadequate for modeling such scenarios, as they generally rely on the assumption of independent responses. A natural and widely used class of network models for such settings is that of Markov Random Fields (MRFs), which provide a flexible framework for modeling network dependencies. This class includes the Ising model (for binary outcomes) and the Potts model (for discrete observations) as two popular variants.
In particular, the Ising model is a discrete exponential family for modeling network-dependent binary data which was originally coined by physicists as a model for ferromagnetism \cite{ising}, and has found immense applications since then in a variety of areas, such as image processing, neural networks, spatial statistics, computational biology, epidemiology, social sciences and political sciences \cite{geman_graffinge, neural, spatial, ising_dna_correlations, disease,  innovations, election_spin}. For a network with $N$ nodes, it is a probability distribution on the hypercube $\{-1,1\}^N$, given by:
\begin{equation}\label{modeldef}
    \p_{\beta}(\bm x) := \frac{1}{2^N Z(\beta)}e^{\beta H(\bm x)}\quad, 
    ~~~\bm x = (x_1, \cdots, x_N)^\top \in \{-1,1\}^N,
\end{equation}
where $x_i$ is the observed outcome at the $i$-th node of the network, 
$\beta > 0$ is a model parameter, 
$H$ is the sufficient statistic having the form: 
$$H(\bm x) := \sum_{1\le i<j\le N} J_{i,j}x_i x_j,$$
for some symmetric interaction matrix $\bm J = \bm J_N := ((J_{i,j}))_{1\le i,j\le N}$ with zeros on the diagonal, and $Z(\beta)$ is the normalizing constant needed to ensure that the probabilities in \eqref{modeldef} add up to $1$. 
In physics, the quantities $\beta$, $H$ and $Z$ are often referred to as  the \textit{inverse temperature}, the \textit{Hamiltonian} and the \textit{partition function}, respectively. 

As noted earlier, one often does not have the luxury of accessing multiple samples from network models for learning, and must instead rely on a single observed sample ($\bm x$).
In the context of the Ising model \eqref{modeldef}, learning the parameter $\beta$ from only one sample $\bm x \sim \p_{\beta}$ by a simple maximum likelihood approach is not a viable option, due to the presence of the computationally expensive normalizing term $Z(\beta)$. An alternative, computationally efficient approach suggested in \cite{chatterjee}, was to use Julian Besag's \textit{maximum pseudolikelihood (MPL)} estimator \cite{besag_lattice, besag_nl}:
\begin{equation}\label{MPL}
  \hat{\beta}_{MPL} := \arg \max_{\beta} \mathcal{L}(\beta) 
\end{equation}
where the \textit{pseudolikelihood function} $\mathcal{L}(\beta)$ is defined as:
$$\mathcal{L}(\beta) := \prod_{i=1}^N \p_{\beta}\left(X_i|\bm X_{-i}\right),$$
with $\bm X_{-i} := (X_{j})_{j\ne i}$. 
The presence of the conditional probabilities makes the computation of the pseudolikelihood function easier, as they are free of any term involving the partition function.  In fact, the pseudolikelihood function has an explicit form, and one can use methods as simple as Newton-Raphson or even grid search to maximize it. Consistency and rate of convergence of the MPL estimator were established in \cite{chatterjee} under certain assumptions; this result was subsequently extended by \cite{BM16} and \cite{joint} to obtain different rates of convergence of the MPL estimator for Ising models on weighted graphs and in the context of joint estimation of parameters in Ising models. Rates of convergence of the MPL estimator in Ising-type generalizations of the logistic regression model with network-dependent responses were established in \cite{cd_ising_II, cd_ising_I,mukherjee_jmlr}. Asymptotic distributions of the MPL estimators in the fully-connected Ising model (i.e. when all the off-diagonal entries of the matrix $\bm J$ are $1/N$) were established in \cite{comets}, which can be used to construct confidence intervals for the model parameters. Generalizations of some of these results for the tensor version of the Ising model \eqref{modeldef} can be found in \cite{psm1, psm2, psm3}.

Although throughout this paper, we are going to assume that the matrix $\bm J$ is known and the task is to estimate the parameter $\beta$ only, we would like to mention that a closely related area of interest, known as \textit{structure learning}, assumes that the matrix $\bm J$ is unknown, and the task is to recover it (see \cite{bresler, bms,highdim_ising,vml20,vmlc16}). Structure learning involves estimating $\binom{N}{2}$ free parameters (upper-diagonal entries of $\bm J$), which inherently demands access to multiple i.i.d. samples from the model. However, such data is typically unavailable or unrealistic in many practical scenarios as discussed earlier. 

As noted in \cite{chen2011learning}, classical algorithms for parameter estimation and structure learning in network data do not account for the challenges of missing or contaminated data. For instance, in social networks, it is often the case that only a partial set of factors is observed, with certain behaviors or interactions left unrecorded. Moreover, as mentioned in \cite{katiyar2020robust}, malfunctioning sensors in networks can lead to the reporting of corrupted/flipped outcomes. As a result, developing methods for learning graphical models from incomplete or contaminated data is of paramount interest. In spite of the existence of a considerable amount of literature on parameter estimation and structure learning in Ising models based on full and uncontaminated data, there is a relative scarcity of analogous robust learning procedures when the underlying data is corrupted or missing. And, most, if not all, of the available body of work in the latter direction focuses on robust structure learning using multiple samples (see, for example, \cite{hubers_contamination_ising, dobrushin_condition_robust_ising, ising_models_independent_failures, robust_learning_ising, semi_supervised_ising, ising_models_independent_failures, katiyar2020robust} and the references therein), but the problem of robustly estimating the parameter $\beta$ using only one sample (and known $\bm J$) has been largely unexplored. 

In this paper, we partly fill this gap in the literature of network data modeling by developing a robust estimator of the parameter $\beta$ in \eqref{modeldef} based only on one sample at risk of possible contamination. 
For this purpose, we adopt the approach of minimum divergence inference, which estimates model parameters by minimizing a suitable statistical divergence between the data and the model. The choice of a particular divergence governs the properties of the resulting estimator, such as its asymptotic efficiency and robustness. The use of appropriate density based divergence measures leads to estimators which are both highly efficient under pure data and also highly robust under possible data contamination \cite{basu2011statistical}, making this approach popular for generating robust estimators under different parametric set-ups. Here we use one such popular family of divergences, called the \textit{density power divergence (DPD)}, originally introduced in \cite{basu1998robust} for independent and identically distributed data. 
This divergence family is defined in terms of a positive-valued tuning parameters that controls the trade-offs between the efficiency and robustness of the resulting minimum divergence estimator. Due to its nice interpretations as a generalization of the maximum likelihood estimator and a model-dependent $M$-estimator, relatively easier computational burden, and high efficiency and robustness properties, this minimum DPD (MDPD) estimator has become extremely popular in the literature of robust parametric inference in recent times, and has been successfully applied to a wide variety of data types and parametric models. 
These include regression models and general non-homogeneous set-ups  \cite{ghosh2013robust,ghosh2015influence,ghosh2016robust,ghosh2017robust} 
as well as high dimensional regression models \cite{ghosh2020ultrahigh, ghosh2023robustness,ghosh2024robust}.

\subsection{Our Contributions}
\textcolor{black}{ 
In this paper, we extend the framework of MDPD estimation to the Ising model \eqref{modeldef} based on only one sample observed from this model, and study the theoretical properties of the MDPD estimators. In particular, we show that the resulting MDPD estimator of the parameter $\beta$ in the model \eqref{modeldef} is asymptotically $\sqrt{N}$-consistent (as $N \rightarrow\infty$) under standard assumptions (Section \ref{sec:sccons}), and also locally B-robust (Section \ref{sec:locrob8}) with significantly smaller gross error sensitivities.
The underlying conditions are quite general and are verified for a wide variety of network structures, including common graph ensembles such as regular graphs, Erd\H{o}s-R\'enyi random graphs, general dense graphs, as well as for other types of interaction structures, such as spin glass models in statistical physics, encompassing the celebrated Sherrington-Kirkpatrick model and the Hopfield model.} 

We also establish asymptotics of a class of general $Z$-estimators 
(which encompasses our proposed MDPD estimators) for pure data under a set of flexible assumptions 
that can be verified for Ising models on a wide variety of approximately regular graphs that are not too sparse (Section \ref{sec:cltasp8}). To the best of our knowledge, 
this is the very first attempt towards proving consistency and central limit theorems for $Z$-estimators in such general settings, without assuming any specific network structure such as the complete graph or the Erd\H{o}s-R\'enyi graph. 
For such Ising models, particularly those on slightly dense Erd\H{o}s-R\'enyi and stochastic block models, 
the MDPD estimators are seen to exhibit improved robustness under contamination 
in spite of having full asymptotic relative efficiency with respect to the MPL estimator (see Section \ref{sec:tradeoffn8}). 
In the literature of robust statistics, this is also a remarkable discovery 
of statistical model set-ups where MDPD estimators show such a favorable behavior, 
making them at par with the classical minimum disparity estimators 
(such as the minimum Hellinger distance estimator).

Further, we examine finite sample performances of the MDPD estimators of $\beta$ under both pure and contaminated synthetic data generated from different commonly used network structures through extensive simulation studies (Section \ref{sec:simulatons}). Finally, we demonstrate the applicability and advantages of our proposed methodology in real-world datasets from the domains of social networks, neurobiology, and genomics (Section \ref{sec:real data}). Our findings are consistent with domain knowledge: for datasets that are pre-cleaned or have little risk of contamination, the MDPD estimators perform on par with the MPL estimator, whereas in datasets with plausible contamination, MDPD estimators exhibit remarkably better performances across domains. 
Finally, we conclude with a discussion on potential extensions of our work to more general MRFs, such as the Potts model and tensor Ising models, and further to logistic regression models with network-dependent responses.

\section{Proposed Methodology and Theoretical Guarantees}

In this section, we introduce the minimum density power divergence (MDPD) estimator and describe its theoretical guarantees.

\subsection{The Minimum Density Power Divergence Estimator}

We start by describing the minimum DPD inference framework in the context of the Ising model \eqref{modeldef}. Following the construction in \cite{basu2017wald}, let us consider the probability vectors:
$$\bm p(\beta) := \left(\frac{\p_\beta(X_1=1|\bm x_{-1})}{N}, \frac{\p_\beta(X_1=-1|\bm x_{-1})}{N},\cdots,\frac{\p_\beta(X_N=1|\bm x_{-N})}{N}, \frac{\p_\beta(X_N=-1|\bm x_{-N})}{N}\right)^\top$$
and
$$\hat{\bm p} := \left(\frac{y_1}{N},\frac{1-y_1}{N},...,\frac{y_N}{N},\frac{1-y_N}{N}\right)^\top,$$
where $\bm x_{-i} := (x_j)_{j\ne i}$ and $y_i := (1+x_i)/2$. Then, the Kullback-Leibler (KL) divergence between these probability vectors $\hat{\bm p}$ and $\bm p(\beta)$ turns out to be
\begin{eqnarray*}
    d_{KL}(\hp, \pb) &=& \frac{1}{N}\sum_{i=1}^N \left[y_i \log \frac{y_i}{\p_{\beta}(X_i=1|\bm x_{-i})} + (1-y_i) \log \frac{1-y_i}{\p_{\beta}(X_i=-1|\bm x_{-i})}\right]\\&=& -\frac{1}{N} \sum_{i=1}^N \left[y_i \log \p_{\beta}(X_i=1|\bm x_{-i}) + (1-y_i) \log \p_{\beta}(X_i=-1|\bm x_{-i})\right] + I(\bm y)\\&=& -\frac{1}{N} \sum_{i=1}^N \log \p_\beta(x_i|\bm x_{-i}) + I(\bm y) = -\frac{1}{N} \log \mathcal{L}(\beta) + I(\bm y),
\end{eqnarray*}
where $I(\bm y) := N^{-1}\sum_{i=1}^N [y_i \log y_i + (1-y_i)\log (1-y_i)]$. Consequently, the MPL estimator of $\beta$, previously defined in \eqref{MPL}, can equivalently be defined as :
$$\hat{\beta}_{MPL} = \arg \min_{\beta} d_{KL}(\hp, \pb).$$

This gives us a natural way to define suitable generalizations of the MPL estimator of $\beta$ by minimizing any valid divergence between $\hp$ and $\pb$, instead of the KL divergence. In particular, to achieve the desired robustness under data contamination, we use the density power divergence (DPD) of \cite{basu1998robust} due to its excellent robustness properties. The DPD between the vectors $\hp$ and $\pb$ is given by:
\begin{equation}\label{dpd}
    d_{\lambda}(\hp, \pb)=\frac{1}{N^{1+\lambda}}\left\{ \sum_{i=1}^N\left(\sum_{j=1}^2 \pi_{ij}^{1+\lambda} -\left(1+\frac{1}{\lambda}\right)\sum_{j=1}^2 y_{ij}\pi_{ij}^{\lambda}\right) + \frac{N}{\lambda}\right\}, ~~~~~~~~~\lambda > 0,~~~~
\end{equation}
where $y_{i1} = y_i, y_{i2}=1-y_i$, $\pi_{i1} = \p_\beta(X_i=1|\bm X_{-i})$ and $\pi_{i2} = \p_\beta(X_i=-1|\bm X_{-i})$.
The corresponding \textit{minimum density power divergence (MDPD)} estimator of $\beta$ is then defined as: 
$$ \hat{\beta}_{\lambda}=\arg\min_{\beta} d_{\lambda}(\hp, \pb).$$
Here the tuning parameter $\lambda$ controls the trade-off between robustness and efficiency. 
In particular, it is easy to see that 
$\lim\limits_{\lambda\rightarrow0}~d_{\lambda}(\hp, \pb) = - N^{-1}\log \mathcal{L}(\beta)$, and hence the MDPD estimator coincides (in a limiting sense) with the efficient but highly non-robust MPL estimator as $\lambda\rightarrow0$.
The MDPD estimators achieve greater robustness, at a price of loss in pure data efficiency, as $\lambda$ increases; 
however, such a loss in efficiency is often not quite significant at smaller $\lambda>0$ as illustrated in Section \ref{sec:simulatons}.

In order to compute the MDPD estimator $\hat{\beta}_\lambda$
at any given $\lambda$, we can either numerically minimize the associated objective function $d_{\lambda}(\hp, \pb)$ 
given in \eqref{dpd} or alternatively solve the corresponding estimating equations derived below. For this purpose, it is more convenient to rewrite \eqref{dpd} as follows:
$$
d_{\lambda}(\hp, \pb)=\frac{1}{N^{1+\lambda}}\bigg\{ \sum_{i=1}^N\bigg(\pi_{i}^{1+\lambda} + (1-\pi_{i})^{1+\lambda} \\
-\left(1+\frac{1}{\lambda}\right)\big(y_i \pi_i^{\lambda}+(1-y_i)(1-\pi_i)^{\lambda}\big)\bigg) + \frac{N}{\lambda}\bigg\}$$
where 
$$
\pi_i := \p_\beta(X_i=1|\bm x_{-i}) = \frac{e^{2\beta m_i(\bm x)}}{1+ e^{2\beta m_i(\bm x)}}~,\quad\text{with} ~ m_i(\bm x) :=\sum_{j=1}^N J_{i,j} x_j.
$$
Note that $\frac{\partial \pi_i}{\partial \beta} = 2\pi_i(1-\pi_i) m_i(\bm x).$
Using this, we have:

\begin{eqnarray*}
    &&\frac{\partial d_{\lambda}(\hp,\pb)}{\partial \beta}\\ &=& 
     \frac{2(1+\lambda)}{N^{\lambda+1}}\sum_{i=1}^N m_i(\bm x)\bigg(\pi_i^{\lambda+1}(1-\pi_i)-\pi_i(1-\pi_i)^{\lambda+1}
     \\
     &&\quad\quad\quad\quad\quad\quad\quad\quad\quad\quad\quad
     -\pi_i^{\lambda}(1-\pi_i)y_i+\pi_i(1-\pi_i)^{\lambda}(1-y_i)\bigg) \\
    &= &\frac{2(1+\lambda)}{N^{\lambda+1}}\sum_{i=1}^N m_i(\bm x)(e^{2\lambda \beta m_i(\bm x)}+ e^{2\beta m_i(\bm x)})\cdot \frac{e^{2\beta m_i(\bm x)} -y_i(1+e^{2\beta m_i(\bm x)})}{(1+e^{2\beta m_i(\bm x)})^{\lambda+2}}\\
    &= &  -\frac{1+\lambda}{\textcolor{black}{2^\lambda} N^{\lambda+1}}\sum_{i=1}^N \frac{\cosh((\lambda-1)\beta m_i(\bm x))}{\cosh^{\lambda+1}(\beta m_i(\bm x))}m_i(\bm x)(x_i - \tanh(\beta m_i(\bm x))).
\end{eqnarray*}

Therefore, the estimating equation for $\hat{\beta}_\lambda$ is given by:
\begin{equation}\label{esteql}
    \sum_{i=1}^N \frac{\cosh((\lambda-1)\beta m_i(\bm x))}{\cosh^{\lambda+1}(\beta m_i(\bm x))}m_i(\bm x)(x_i - \tanh(\beta m_i(\bm x))) = 0.
\end{equation}
Clearly, the above estimating equation for the MDPD estimator  is also valid at $\lambda=0$ (which is nothing but the MPL estimator). Furthermore, at any $\lambda>0$, it is easy to see that these are indeed unbiased estimating equations for the assumed Ising model \eqref{modeldef}.   
However, the summands in \eqref{esteql} are not independent, and hence, simple M-estimation techniques cannot be applied to study the asymptotic properties of $\hat{\beta}_\lambda$.

\subsection{Consistency of the MDPD estimators under Ising models}\label{sec:sccons}

We justify the asymptotic validity of the proposed MDPD estimators $\hat{\beta}_\lambda$ under the Ising model in \eqref{modeldef} by proving its consistency for large networks (i.e., $N\rightarrow \infty$). 
The following theorem presents the specific consistency result under rather general conditions on the interaction matrix $\bm J$ and the rate of growth of the partition function $Z(\beta)$.

\begin{thm}\label{cons}
    Let $\hat{\beta}_{\lambda}$ be the MDPD estimator of $\beta$ for some fixed $\lambda>0$, based on a single observation $\bm X$ from the  Ising model \eqref{modeldef}. Suppose that the following conditions hold:
    \begin{enumerate}
        \item $\sup_{N\geq 1} \|\bm J\|<\infty$,
        \item $\liminf_{N\rightarrow \infty}\frac{1}{N} \log Z(\beta)>0$,
    \end{enumerate}
    where $\|\cdot\|$ denotes the spectral norm. Then $\hat{\beta}_{\lambda}$ is a $\sqrt{N}$-consistent sequence of estimators for $\beta$, i.e. $\sqrt{N}(\hat{\beta}_{\lambda}-\beta) = O_P(1)$.
\end{thm}


The proof of Theorem \ref{cons} is based on some non-trivial extensions of the arguments used in proving the consistency of the MPL estimators under similar model assumptions (see, e.g., \cite{chatterjee}). 
Before presenting the detailed proof, which is deferred to Section \ref{sec:proof} for ease of presentation, let us now discuss the required assumptions and their validity for commonly used network models. 

In fact, both assumptions of Theorem \ref{cons} are satisfied by a wide variety of interaction structures. The most commonly used interaction structure is probably the adjacency matrix ($\bm A$) of a graph scaled by its average degree, i.e.
$$\bm J_N = \frac{1}{d_{avg}(G_N)}\bm A (G_N),$$
where $G_N$ is a graph on $N$ vertices, and $d_{avg}(G_N) := \frac{1}{N}\sum_{i=1}^N d_i(G_N)$ is the average degree of $G_N$. In this case, $$\|\bm J_N\| \le \frac{\max_{1\le i\le N} d_i(G_N)}{d_{avg}(G_N)}$$ and hence, Condition (1) is satisfied as long as the maximum degree of the graph is of the same order as its average degree. This encompasses a wide variety of graph ensembles such as bounded degree graphs, dense graphs, regular graphs, and even irregular graphs which do not have any particular node with atypical high degree. 

Verifying Condition (2) is slightly more tricky, and two kinds of phenomena arise. It holds for all values of $\beta>0$ for bounded degree graphs, but for dense graph ensembles, there typically exists a threshold $\beta^*$, such that it is violated for all $\beta <\beta^*$ and holds true for all $\beta > \beta^*$ (often referred to as the \textit{high-temperature} and \textit{low-temperature} regimes, respectively,  in statistical physics). Therefore, the MDPD estimators are consistent at all $\beta$ for Ising models on bounded degree graphs, and in the low-temperature regime for Ising models on dense graphs. Although Theorem \ref{cons} does not say anything about what happens in the high-temperature regime in the latter case, one can show by a simple contiguity argument that no consistent estimator for $\beta$ exists in this regime (see \cite{BM16}). Below, we give some common examples of interaction structures that satisfy these two required conditions.

\begin{example}[Non mean-field interactions/bounded-degree graphs]
   If the interaction matrix $\bm J_N$ satisfies 
   \begin{equation}\label{example:bdded7}
       \sup_{N\ge 1} \|\bm J_N\| <\infty\quad \text{and}\quad \sum_{1\le i<j\le N} J_{i,j}^2 \gg N,
   \end{equation}
   then it follows from the proof of Corollary 2.6 in \cite{psm2}  that Conditions (1) an (2) of Theorem \ref{cons} are satisfied, which makes the MDPD estimators $\sqrt{N}$-consistent for all $\beta >0$. If $\bm J$ is the adjacency matrix of a graph (scaled by its average degree), then \eqref{example:bdded7} is equivalent to the maximum degree of the graph being bounded.
\end{example}

\begin{example}[Regular graphs] If $G_N$ is a sequence of $d_N$-regular graphs and $\bm J_N := \frac{1}{d_N} \bm A(G_N)$, then Condition (1) is satisfied as discussed above.
Condition (2) is also satisfied for $\beta > 1$ (see Corollary 3.1 in \cite{BM16}) . In this case, the MDPD estimators are $\sqrt{N}$-consistent for all $\beta>1$ .  
\end{example}

\begin{example}[Sequence of dense graphs converging to a Graphon] Suppose that $G_N$ is a sequence of simple graphs converging in cut-metric to a graphon $W$ (see \cite{Borgs2008, Borgs2012}), such that $\int_{[0,1]^2}W(x,y)\text{d}x\text{d}y>0$. 
Consider the family of probability distributions \eqref{modeldef} with $\bm J_N := \frac{1}{N}\bm A(G_N)$, where $\bm A(G_N)$ is the adjacency matrix of graph $G_N$.  Then Condition (1) is satisfied by the discussion above. Also, Condition (2) is satisfied for $\beta>\frac{1}{\|W\|}$, where $\|\cdot\|$ is the operator norm (see Theorem 3.3 in \cite{BM16}).
This example encompasses Ising models on many well-known random graph ensembles, such as Erd\H{o}s-Rényi graphs $G(N,p)$ with fixed $p >0$ and stochastic block models with fixed block connection probabilities.
For example, if $G_N \sim G(N,p)$ with $0< p \leq 1$ fixed, then $G_N$ converges in cut-metric to the constant graphon $p$, and hence, the MDPD estimators are $\sqrt{N}$-consistent for $\beta > \frac{1}{p}$.
\end{example}

\begin{example}[Spin glass models]\label{skmodel_ex} 
Probably the most celebrated Ising spin glass model is the Sherrington–Kirkpatrick model. Here, we have a family of probability distributions of the form \eqref{modeldef} with $J_{ij} := N^{-\frac{1}{2}}g_{ij}$, where $(g_{ij})_{1\leq i<j<\infty}$ is a fixed realization of a sequence of independent standard normal random variables (and $g_{ij}=g_{ji}$). It is shown in Section 1.4 of \cite{chatterjee} that both conditions (1) and (2) are satisfied for every $\beta>0$. Hence, the MDPD estimators are $\sqrt{N}$-consistent for all $\beta>0$ in the Sherrington–Kirkpatrick model.

Another popular spin glass model is the Hopfield model of neural networks, with $N$ particles affected by $M$ attractors. Here, the underlying probability distributions are again of the form \eqref{modeldef} with $J_{ij}= \frac{1}{N}\sum_{k=1}^M w_{ik}w_{jk}$, where $(w_{ik})_{i\leq N,k \leq M}$ is a fixed realization of a collection of independent random variables with $P(w_{ik} = \pm 1) = 1/2$. Condition (1) can be shown to be satisfied almost surely when $M/N$ remains bounded as $N\rightarrow \infty$ and Condition (2) is verified in \cite{chatterjee} for $\beta>0$. Hence, the MDPD estimators are $\sqrt{N}$-consistent in this regime.
\end{example}

\subsection{Asymptotics for a general class of $Z$-estimators}\label{sec:cltasp8}
In this section, we derive the asymptotic distribution of the MDPD estimators under the pure model.
When the interaction structure is given by the adjacency matrix of an underlying graph, the standard practice is to scale this matrix by the average degree of the graph. Under this convention, $\beta^*=1$ serves as a natural estimation threshold for Ising models on a variety of dense graph ensembles, in the sense that  no estimator of $\beta$ is consistent in the high-temperature ($\beta <1$) regime, while consistent estimation becomes feasible for $\beta > 1$.
 Consequently, we will focus only on the low-temperature $(\beta >1)$ regime throughout the rest of this section.
Under this regime of Ising models, we will derive a central limit theorem (CLT) 
for a class of general $Z$-estimators $\hat{\beta}_g$, defined as solutions of the following estimating equation:
\begin{equation}
        \sum_{i=1}^N g(m_i(\bm X))(X_i - \tanh(\beta m_i(\bm X))) = 0,
        \label{z-est-ee}
    \end{equation}
    where $g = g_\beta: \mathbb{R} \mapsto \mathbb{R}$ is any  differentiable odd function satisfying:
    \begin{enumerate}
        \item For every $\beta\ge 1$, $g_\beta'$ is bounded on bounded subsets of $\mathbb{R}$.\
        \item The function $\beta \mapsto g_\beta(m_i(\bm X))$ is twice-differentiable, and both derivatives are uniformly bounded on some bounded neighborhood of the true parameter $\beta$.
    \end{enumerate}

Note that, our proposed MDPD estimators clearly belong to the above-mentioned class of $Z$-estimators
for a particular choice of the function $g$ having the form:
$$
g_\beta(x) := \frac{x \cosh((\lambda-1)\beta x)}{\cosh^{\lambda+1} (\beta x)},
$$ 
which is an odd function satisfying all necessary assumptions on boundedness of derivatives.
So, the asymptotic distribution of these MDPD estimators will follow directly from the general CLT 
for the $Z$-estimators $\hat{\beta}_g$ proved below.  

A crucial first step in establishing central limit theorems for $Z$-estimators is to verify their consistency under the pure model. The following assumption plays a key role in this direction.

\begin{assumption}\label{asd0}
    The random variable $M(\bm X) := m_+\mathbbm{1}_{\bar{X} \ge 0} - m_+ \mathbbm{1}_{\bar{X} < 0}$, where $m_+$ is the unique positive solution of the equation (in $m$) $\tanh(\beta m) =m$, satisfies:
    $$\sum_{i=1}^N (m_i(\bm X)-M(\bm X))^2 = o_P(N)\quad\text{and}\quad \bar{X}-M(\bm X) = o_P(1).$$ 
\end{assumption}

Below, we state a result on the consistency of $Z$-estimators in the model \eqref{modeldef};
the proof can be found in Section \ref{sec:cltproof921}.

\begin{thm}\label{thm:consiZ}
   Under Assumption \ref{asd0}, $\hat{\beta}_g \xrightarrow{P} \beta$ under the pure model \eqref{modeldef}.
\end{thm}

Note that Assumption \ref{asd0} is more stringent than Conditions (1) and (2) of Theorem \ref{cons} needed to ensure $\sqrt{N}$-consistency for the MDPD estimators. For instance, Ising models on dense, irregular graphs with maximum degree of the same order as the average degree (for example dense, unbalanced stochastic block models with blocks of fixed, unequal proportional sizes) are covered by Conditions (1) and (2) of Theorem \ref{cons}, but do not satisfy Assumption \ref{asd0}. We will later verify Assumption \ref{asd0} for Ising models on certain approximately regular graphs with diverging average degree, while Conditions (1) and (2) of Theorem \ref{cons} apply to much more general settings. This distinction is natural, 
because the former is used to establish consistency of general class of $Z$-estimators, whereas the latter pertain only to the special class of MDPD estimators. Another important point is that Theorem \ref{cons}, unlike Theorem \ref{thm:consiZ} for general $Z$-estimators, establishes a $\sqrt{N}$-rate of consistency of the MDPD estimators, 
so the techniques required in proving the former are much more involved.

Next, in order to establish the CLT for $\hat{\beta}_g$, 
we need the following additional assumptions:

\begin{assumption}\label{asd1}
   For all functions $c_i: \R\mapsto \R$ satisfying $|c_i(x)| \le C$ for all $x\in \R$ and some universal constant $C>0$, we have the following:
   \begin{enumerate}
       \item For all $\alpha_1,\ldots,\alpha_N \in [0,1]$, $$\sum_{i=1}^N c_i(\alpha_i m_i(\bm X) + (1-\alpha_i)M(\bm X))\cdot (m_i(\bm X)-M(\bm X))(X_i -\tanh(\beta m_i(\bm X))) = o_P(\sqrt{N}),$$

       \item $\sum_{i=1}^N c_i(m_i(\bm X))(X_i -\tanh(\beta m_i(\bm X))) = o_P(N).$
   \end{enumerate}
    
\end{assumption}

\begin{assumption}\label{asd2}
There exists $\sigma_\beta >0$, such that:
$$\frac{1}{\sqrt{N}}\sum_{i=1}^N (X_i - \tanh(\beta m_i(\bm X))) \xrightarrow{D} \mathcal{N}\left(0, \sigma_{\beta}^2\right)$$  
\end{assumption}

\begin{assumption}\label{asd3}
The following CLT holds:
$$\frac{g(M(\bm X))}{\sqrt{N}}\sum_{i=1}^N (X_i - \tanh(\beta m_i(\bm X))) \xrightarrow{D} \mathcal{N}\left(0, g(m_+)^2 \sigma_{\beta}^2\right)$$  
\end{assumption}


The problem of proving CLT for the Hamiltonian and parameter estimators in general Ising models is inherently difficult, unless one imposes additional simplifying structural conditions on the interaction matrix satisfying the above assumptions. One can, for example, follow the framework developed in \cite{pivotal_clt,meanfield}, and assume the following structural properties of the interaction matrix:

\begin{cond}\label{bounded_sum}
    The interaction matrix $\bm J_N$ satisfies: 
    $$\limsup_{N\rightarrow \infty} \max_{1\leq i \leq N} \sum_{j=1}^N J_{i,j} <\infty.$$
\end{cond}

\begin{cond}\label{regular}
    The interaction matrix is approximately regular, i.e., $\bm J_N$ has nonnegative entries, is symmetric and satisfies:
    $$
        \lambda_1(\bm J_N)\xrightarrow{N\rightarrow\infty} 1, \frac{1}{N}\sum_{i=1}^N \delta_{R_i}\rightarrow 1,\text{ where } R_i:=\sum_{j=1}^N J_{i,j},
    $$
    where $\lambda_1(\bm J_N)\geq \lambda_2(\bm J_N)\geq \cdots \geq \lambda_N(\bm J_N)$ are the eigenvalues of $\bm J_N$.
\end{cond}

\begin{cond}\label{mean_field}
    The Frobenius norm of the interaction matrix $\bm J_N$ satisfies:
    $$
        \|\bm J_N\|^2_F := \sum_{1\leq i,j \leq N} J_{i,j}^2 = o(N).
    $$
\end{cond}

\begin{cond}\label{spectral_gap}
    The interaction matrix $\bm J_N$ satisfies the spectral gap condition, i.e.,
    $$\limsup_{N\rightarrow \infty} \frac{\lambda_2(\bm J_N)}{\lambda_1(\bm J_N)} <1.$$
\end{cond}

We will show in Section \ref{sec:struct} of the Appendix, that Assumptions \ref{asd0}, \ref{asd1} and \ref{asd2} are satisfied with $\sigma_\beta^2 := (1-m_+^2)(1-\beta(1-m_+^2))$ for interaction structures satisfying Conditions \ref{bounded_sum}, \ref{regular}, \ref{mean_field} and \ref{spectral_gap}. This only leaves Assumption \ref{asd3} open for verification in common examples. It is known to hold in Curie-Weiss models (by Proposition 2.1 in \cite{comets} and a straightforward delta-method argument), and is very likely to hold for Ising models on interaction structures satisfying Conditions \ref{bounded_sum}, \ref{regular}, \ref{mean_field} and \ref{spectral_gap} too, as is evident from the framework of \cite{pivotal_clt} (one can, for example, prove this rigorously, if a conditional version of Theorem 5.3 in \cite{pivotal_clt} is proved given $\bar{X} \ge 0$ and $\bar{X}\le 0$).

Let us now say a few words about Conditions \ref{bounded_sum}, \ref{regular}, \ref{mean_field} and \ref{spectral_gap}. Of course, the most popular choice of the interaction matrix $\bm J_N$ is the adjacency matrix of a graph $G_N$ scaled by its average degree. In this case, Condition \ref{bounded_sum} does \textit{not} impose any sparsity restriction on $G_N$, but simply says that the maximum degree of $G_N$ is of the same order as its average degree, i.e. there is no vertex with atypically high degree. Conditions \ref{bounded_sum}, \ref{regular} and \ref{mean_field} are satisfied by the scaled adjacencies of deterministic/random regular graphs with degree going to $\infty$, and also encompass Erd\H{o}s-R\'enyi graphs $G(N,p_N)$ with $p_N \gg \log N/N$, balanced stochastic block models with sum of between and within group connection probabilities $\gg \log N/N$, and even sparse inhomogeneous random graphs generated from regular graphons (see \cite{pivotal_clt} for a complete discussion on this).

We now  establish the asymptotics of general $Z$-estimators of $\beta$ in the low-temperature ($\beta>1$) regime
under the above-mentioned assumptions, which in particular, applies to the MDPD estimators.
The main result is given in the following theorem, 
while its proof is  given in Section \ref{sec:cltproof92}.
\begin{thm}\label{cltthm8}
    Suppose that Assumptions \ref{asd0}, \ref{asd1}, \ref{asd2} and \ref{asd3} hold. Then, under the pure Ising model \eqref{modeldef}, we have:
  
    \begin{equation}\label{z-est92}
        \sqrt{N}(\hat{\beta}_g -\beta)\xrightarrow{D} \mathcal{N}\left(0, \frac{\sigma_\beta^2}{m_+^2(1-m_+^2)^2}\right)
    \end{equation}
   where $m_+$ is the unique positive solution (in $m$) of equation $\tanh(\beta m)=m$. In particular, MDPD estimators also satisfy the CLT \eqref{z-est92} for all $\lambda > 0$.
\end{thm}

For Ising models on interaction structures satisfying Conditions \ref{bounded_sum}, \ref{regular}, \ref{mean_field} and \ref{spectral_gap}, we thus have:
\begin{equation}\label{eq:clt_un_asm}
    \sqrt{N}(\hat{\beta}_g -\beta)\xrightarrow{D} \mathcal{N}\left(0, \frac{1-\beta(1-m_+^2)}{m_+^2(1-m_+^2)}\right)
\end{equation}
which, in fact, matches exactly the asymptotic distribution of the maximum likelihood estimator $\hat{\beta}_{MLE}$ of $\beta$ in the Curie-Weiss model with $\beta>1$ as stated in Theorem 1.4 (c) of \cite{comets}.
Theorem \ref{cltthm8} additionally implies that this same asymptotic distribution also holds for a large class 
of $Z$-estimators under a wide range of Ising models.

As a special case, Theorem \ref{cltthm8} also states that all MDPD estimators share the same asymptotic variance under the pure model in our regime of interest (which is already shown to encompass a large class of interaction structures). In other words, none of the MDPD estimators loses asymptotic efficiency compared to the MPL estimator. 
Even more appealing is that MDPD estimators still retain their enhanced robustness property under contamination, as will be evident in the subsequent parts of the paper (see, e.g., Figures \ref{DERM242} and \ref{SBM142}).
To our knowledge, this is the first such example of a statistical set-up 
where MDPD estimators can produce robust inference in noisy samples 
without losing asymptotic pure data efficiency, making Theorem \ref{cltthm8} truly remarkable.

Finally, the CLT \eqref{eq:clt_un_asm} can be used to construct asymptotic confidence intervals 
for the parameter $\beta$ in general Ising models satisfying Assumptions \ref{asd0} to \ref{asd3}.
Towards this, note that by Assumption \ref{asd0}, we have $\bar{X} =(1+o_P(1))M(\bm X)$ and hence, $\bar{X}^2 = (1+o_P(1))m_+^2$. It now follows from \eqref{eq:clt_un_asm} and an application of Slutsky's theorem that:
$$\sqrt{\frac{N\bar{X}^2(1-\bar{X}^2)}{1-\hat{\beta}_g (1-\bar{X}^2)}} (\hat{\beta}_g-\beta)\xrightarrow{D} \mathcal{N}(0,1).$$ Thus, a $(1-\alpha)$ coverage confidence interval for $\beta$ can be given as:
$$\left[\hat{\beta}_g - |\bar{X}|^{-1} z_{\alpha/2} \sqrt{\frac{1-\hat{\beta}_g (1-\bar{X}^2)}{N (1-\bar{X}^2)}} ~,~\hat{\beta}_g + |\bar{X}|^{-1} z_{\alpha/2} \sqrt{\frac{1-\hat{\beta}_g (1-\bar{X}^2)}{N (1-\bar{X}^2)}}\right],$$
where $z_{\alpha/2}$ is defined by $\mathbb{P}(\mathcal{N}(0,1)>z_{\alpha/2}) = \alpha/2$.

\subsection{Local robustness of the MDPD estimators}\label{sec:locrob8}

In this section, we analyze the classical influence function (see \cite{hampel1,hampel2, basu2011statistical, ghosh2015influence, maronna2019robust}) to study the local (B-)robustness of the MDPD estimators under the Ising model. 
Toward this, we first need to define a suitable  statistical functional corresponding to the MDPD estimator $\hat{\beta}_\lambda$. 
We model the distribution of $\bm X := (X_1, X_2,...,X_N)^\top\in \{-1,+1\}^N$ by the family of probability distributions of the form $\p_{\beta}$ as defined in \eqref{modeldef}. 
If $G$ denotes the true distribution of $\bm X$, which may or may not be equal to $\p_{\beta}$, we define the corresponding MDPD functional for $\beta$ at $G$ as follows.  

\begin{defi}\label{functional}
The MDPD functional $T_{\lambda}(G)$, corresponding to the
MDPD estimator $\hat{\beta}_{\lambda}$, at the true  distribution $G$ is defined as the minimizer of 
$d_{\lambda}(\bm \pi, \pb)$ with respect to $\beta$,
where $\bm\pi_G = \frac{1}{N}(g_1, \ldots,g_N)^\top$,
with 
$$
g_{i} :=\big(\p_G(X_i=1|\bm x_{-i}), 1- \p_G(X_i=1|\bm x_{-i})\big)
$$
being the probability mass function associated with $G_i$, which is the (true) conditional distribution of $X_i$ given $\bm x_{-i}$ under the true joint distribution $G$ of $\bm X$, for each $i=1, \ldots, N$. 
It can equivalently be defined as the  solution of the system of equations
\begin{equation}\label{ee_func}
  \sum_{i=1}^N \E_{G_i} [\Psi_{\lambda}(\bm X, \beta)]=\bm 0,
\end{equation}
with respect to $\beta$, whenever the solution exists, where
$$
    \Psi_{\lambda}(\bm X, \beta) := \frac{\cosh((\lambda-1)\beta m_i(\bm X))}{\cosh^{\lambda+1}(\beta m_i(\bm X))}m_i(\bm X)(X_i - \tanh(\beta m_i(\bm X))).
$$
\end{defi}

Note that the solution of the estimating equation \eqref{ee_func} at $\lambda=0$ is nothing but the MPL functional. Further, the resulting functional $T_{\lambda}(\cdot)$ is Fisher consistent \cite{fisher1922mathematical,cox1974theoretical} for all $\lambda\geq 0$, since 
$$
    T_{\lambda}(\p_\beta) = \beta,\text{ for all } \beta >0.
$$

Now, in order to derive the influence function (IF)
for the MDPD and MPL estimators of $\beta$,
we consider the contaminated distribution 
$G_\eps=(1-\eps)G+ \eps \wedge_{\bm t}$ of $\bm X$,
where $\eps$ is the contamination proportion
and $\wedge_{\bm t}$ is the degenerate distribution at the contaminating point ${\bm t} \in \mathcal{X} :=\{-1, 1\}^N$.
Note that, not all these points ${\bm t}$ produce contamination
in the observed data; a given ${\bm t}$ gives rise to contamination in $i$-th observation $x_i$ if the $i$-th component of ${\bm t}$ takes the value opposite to the one expected under $G_i$ (for the given true distribution $G$). 
Then, the  IF of the functional $T_\lambda$ at 
the true distribution $G$ is defined as \cite{hampel2}
\begin{equation}
    \mathcal{IF}({\bm t};T_\lambda,G)
    =\lim_{\eps\downarrow 0}\frac{T_\lambda(G_\eps)-T_\lambda(G)}{\eps}=\frac{\partial}{\partial\eps}T_\lambda(G_{\eps})\bigg|_{\eps=0}.
\end{equation}
in those ${\bm t}\in \mathcal{X}$ where the limit exists. 

Now, based on \eqref{ee_func}, the functional $T_{\lambda}(G_\eps)$ must satisfy
\begin{equation}\label{tg}
    \sum_{i=1}^N \int \Psi_{\lambda}(\bm x, T_\lambda(G_\eps)) ~dG_{i,\eps}(x_i) = 0,
\end{equation}
where $G_i = (1-\eps) G_i + \eps \wedge_{t_i^*}$
is the contaminated distribution for the $i$-th component induced by $G$; $t_i^*$ may not necessarily be the same as $t_i$
but ${\bm t^*}=(t_1^*, \ldots, t_N^*)^T\in\mathcal{X}$. Differentiating this equation with respect to $\eps$ at $\eps=0$, we obtain:
\begin{equation*}
    0 = \sum_{i=1}^N \bigg[ \int \Psi_{\lambda}(\bm x, T_\lambda(G)) d(\wedge_{t_i^*} - G_i) + \int \frac{\partial}{\partial \beta}\Psi_{\lambda}(\bm x, \beta)\bigg|_{\beta=T_\lambda(G)} 
    \frac{\partial}{\partial\eps}T_\lambda(G_{\eps})\bigg|_{\eps=0} dG_i \bigg].
\end{equation*}
The IF of $T_\lambda$ at $G$ is thus given by: 
\begin{equation*}
\mathcal{IF}(\bm t; T_{\lambda}, G) = - \frac{\sum_{i=1}^N \Psi_{\lambda}(t_i^*, T_{\lambda}(G))}{\sum_{i=1}^N \int \frac{\partial}{\partial \beta}\Psi_\lambda(\bm x, \beta)\big|_{\beta=T_\lambda(G)} dG_i}.
\end{equation*}
This IF can be further simplified
if $G$ belongs to the assumed model family \eqref{modeldef},
as presented in the following theorem. 

\begin{thm}\label{IF}
The influence function of the minimum DPD functional $T_{\lambda}$, as
defined in Definition \ref{functional}, at the model distribution $G=\p_\beta$, for some $\beta>0$, is given by
\begin{equation*}
    \mathcal{IF}(\bm t; T_{\lambda}, \p_\beta) = -\frac{1}{J_\lambda(\beta)} \sum_{i=1}^N\Psi_{\lambda}(t_i^*, \beta), 
\end{equation*}
with 
$$
J_{\lambda}(\beta) = \sum_{i=1}^N\e_{G_i}\left[\frac{\partial \Psi_{\lambda}(\bm X, \beta)}{\partial \beta} \right] = N S_{\bm X,\lambda}^{\prime}(\beta)_2.
$$
where $S_{\bm X,\lambda}^{\prime}(\beta)_2 :=-\frac{1}{N}\sum_{i=1}^N \frac{\cosh((\lambda-1)\beta m_i(\bm X)) m_i^2(\bm X)}{\cosh^{\lambda+3}(\beta m_i(\bm X))}$ (as defined in \eqref{sp2}).

\end{thm}
Interestingly, here we only have finitely many possible values for the contamination point ${\bm t}\in\{-1, 1\}^N$.
Thus, the IF of the MDPD estimators remains bounded at all $\lambda\geq 0$, including the MPL estimator. 
So, in order to compare the extent of robustness of these estimators, we study their gross error sensitivity (GES) defined as the maximum of the absolute value of the IF at the model, i.e., 
\begin{equation*}
\color{black}
GES_\lambda(\beta) := \max\limits_{\bm t \in \{-1, 1\}^N}
\big| \mathcal{IF}(\bm t; T_{\lambda}, \p_\beta) \big| = \frac{1}{|J_{\lambda}(\beta)|}\cdot \max\limits_{\bm t^* \in \{-1, 1\}^N} \left|\sum_{i=1}^N\Psi_{\lambda}(t_i^*, \beta)\right|
\end{equation*}
For simplicity of analysis, we can rewrite the function $\Psi_{\lambda}(t_i^*, \beta)$ as a function of $\beta m_i(\bm t^*)$, and refactor the fraction by $\frac{1}{N}$ as follows:
$$
\color{black}
GES_\lambda(\beta) = \frac{1}{\beta|J_{\lambda}(\beta)/N|}\cdot \max\limits_{\bm t^* \in \{-1, 1\}^N} \left|\frac{1}{N}\sum_{i=1}^N\Tilde{\Psi}_{\lambda}(\beta m_i(\bm t^*),t_i^*)\right|,
$$
where $\Tilde{\Psi}_{\lambda}(m,t_i):=\frac{\cosh((\lambda-1)m)}{\cosh^{\lambda+1}( m)}(t_i - \tanh(m))m$. 
Note that, by Lemma \ref{mx} we know that $|J_{\lambda}(\beta)/N|=|S_{\bm X,\lambda}^{\prime}(\beta)_2|$ is bounded away from 0 with high probability, and from Lemma \ref{f_bound} we have $|J_{\lambda}(\beta)/N|\leq 1$.

We compute these GES numerically and plot them over $\lambda\geq 0$ in Figure \ref{GES}, 
for different values of the true model parameter $\beta>0$ and the following three interaction structures:
\begin{enumerate}
    \item $G(N,p)$ with $N=100$ and $p= 0.5\log(N)/N$,
    \vspace{0.05in}
    \item Stochastic Block model of equal communities with $N=100$, $p= 0.8 \log(N)/N$, and $q=0.2 \log(N)/N$.
    \vspace{0.05in}
    
    \item The Sherrigton-Kirkpatrick model on $100$ nodes (Example \ref{skmodel_ex}).
\end{enumerate}
It is evident that in these cases the GES of the MDPD estimators decreases as $\lambda$ increases, 
indicating their greater robustness under possible data contamination.

\begin{figure}[!h]
    \centering
    \begin{subfigure}[b]{0.49\textwidth}
        \centering
        \includegraphics[width=7cm]{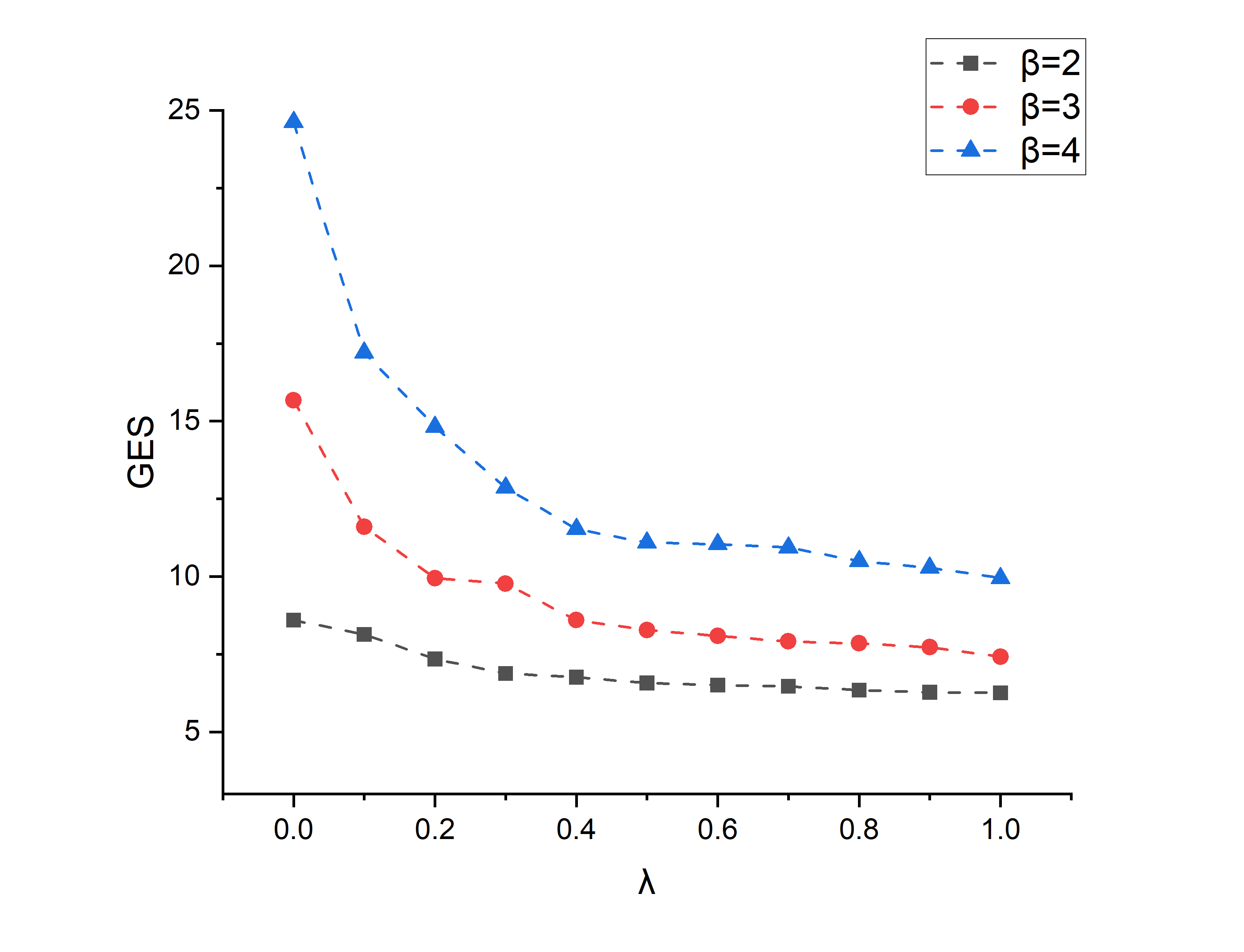}
        \caption[]{{Ising model on $G\left(N,\frac{0.5\log N}{N}\right)$}}  
    \end{subfigure}
    \hfill
    \begin{subfigure}[b]{0.49\textwidth}
        \centering
        \includegraphics[width=7cm]{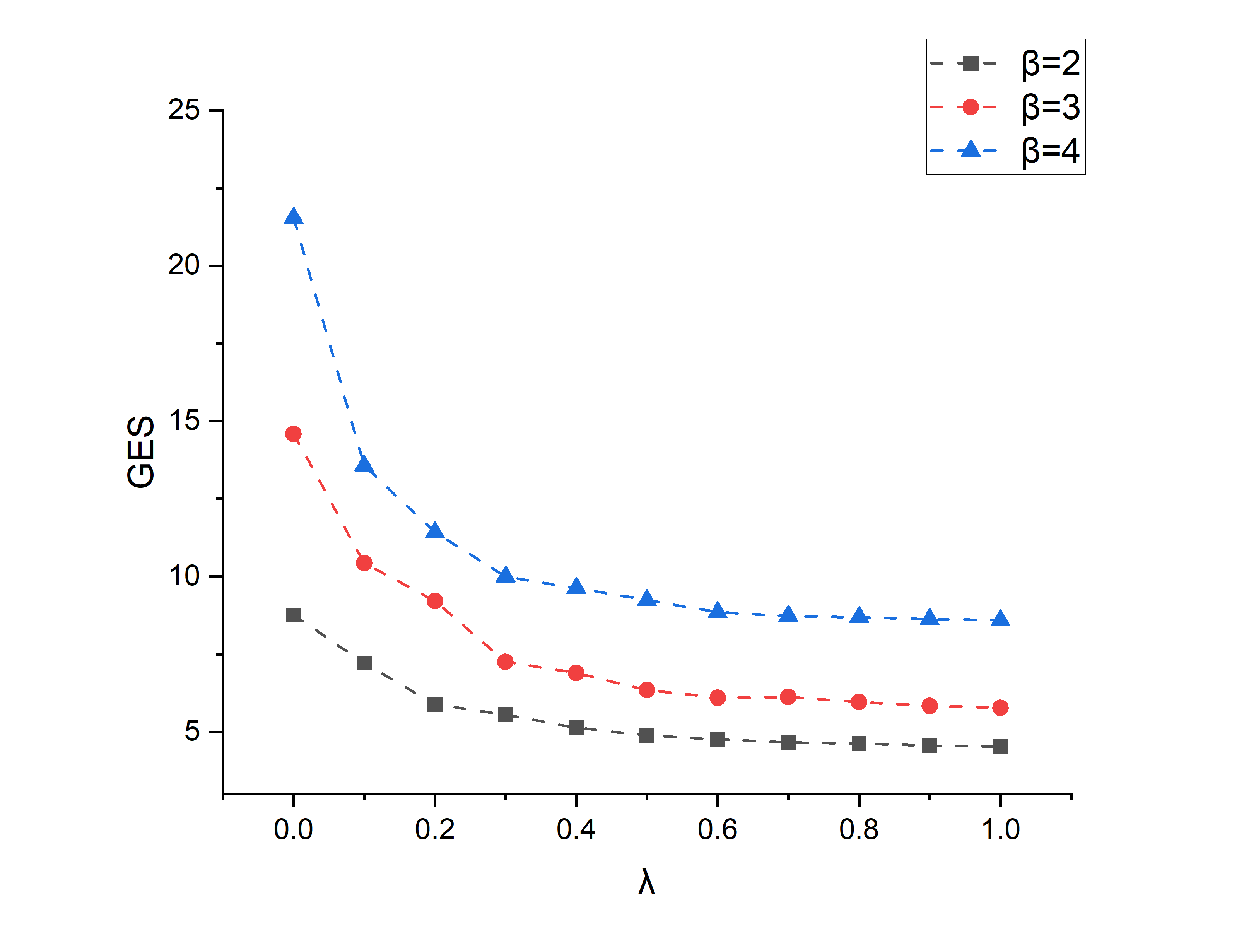}
        \caption[]{{Stochastic block model with $p= \frac{0.8\log N}{N}$ and $q= \frac{0.2\log N}{N}$}}  
    \end{subfigure}

    \begin{subfigure}[b]{0.49\textwidth}  
        \centering 
        \includegraphics[width=7cm]{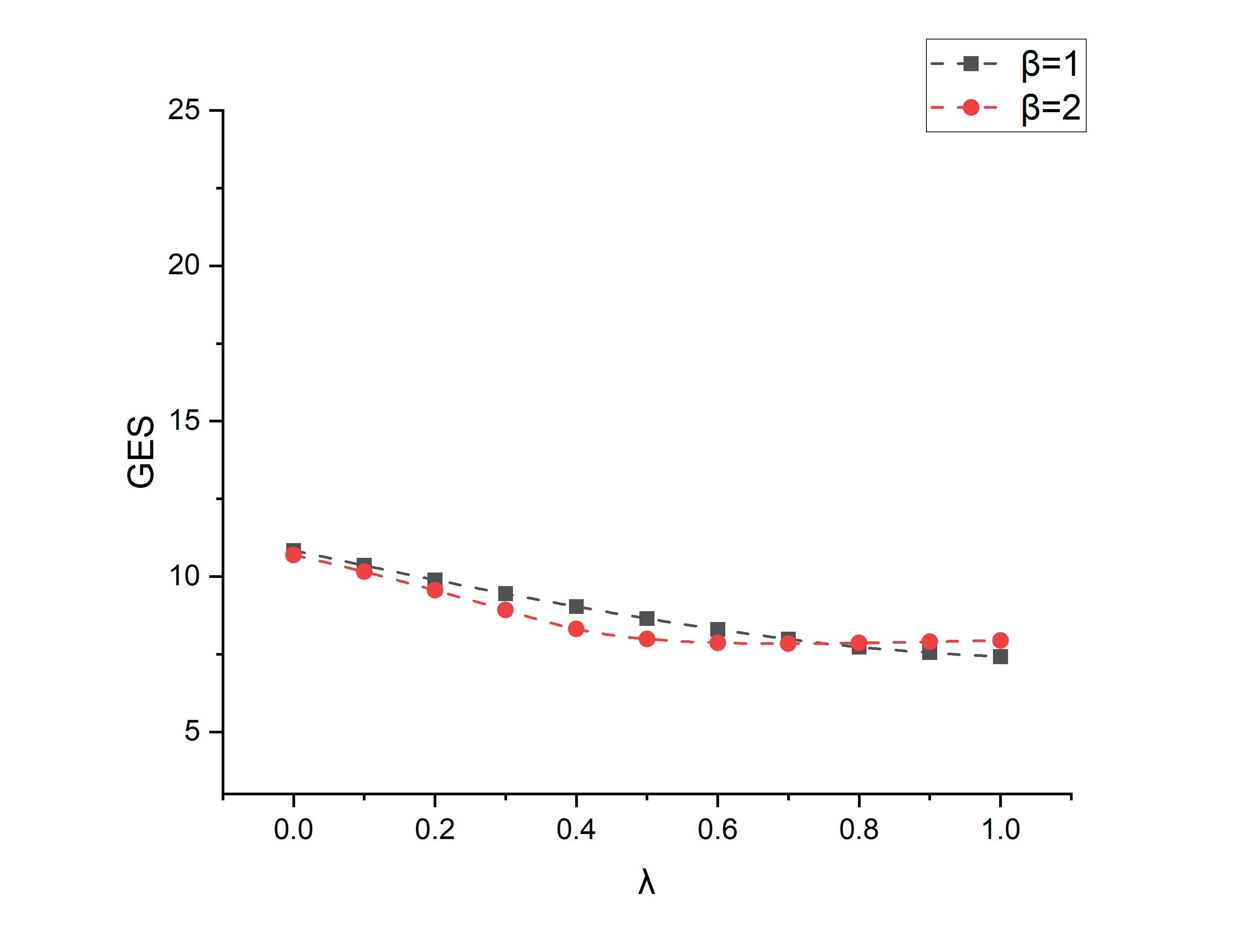}
        \caption[]{{Sherrington Kirkpatrick model}}  
    \end{subfigure}
    \caption{The GES for MDPD estimates from Ising models on different network structures,
    obtained numerically at $N=100$.}
    \label{GES}
\end{figure}

\newpage
\subsection{Proof of the Theorem \ref{cons}}
\label{sec:proof}

Let us now present the detailed proof of Theorem \ref{cons}.
To begin with, we define: 
\begin{equation}\label{Sxl}
\begin{aligned}
     S_{\bm X,\lambda}(\beta) & := -\frac{\textcolor{black}{(2N)}^{\lambda}}{1+\lambda}\frac{\partial d_{\lambda}(\hat{p},p(\beta))}{\partial \beta} \\
     &= \frac{1}{N}\sum_{i=1}^N \frac{\cosh((\lambda-1)\beta m_i(\bm X))}{\cosh^{\lambda+1}(\beta m_i(\bm X))}m_i(\bm X)(X_i - \tanh(\beta m_i(\bm X))),
\end{aligned}
\end{equation}
and
\begin{equation}\label{Sx_prime}
    \begin{aligned}
        S_{\bm X, \lambda}^{\prime}(\beta) &:= -\frac{(2N)^{\lambda}}{1+\lambda}\frac{\partial^2 d_{\lambda}(\hat{p},p(\beta))}{\partial \beta^2} \\
        & = \frac{1}{N}\sum_{i=1}^N \bigg[\frac{\partial f(\beta, m_i(\bm X))}{\partial \beta} m_i(\bm X)(X_i - \tanh(\beta m_i(\bm X)))-f(\beta, m_i(\bm X))\frac{m_i^2(\bm X)}{\cosh^2(\beta m_i(\bm X))}\bigg],
    \end{aligned}
\end{equation}
where
\begin{equation}\label{fx}
    f(\beta, t)=\frac{\cosh((\lambda-1)\beta t)}{\cosh^{\lambda+1}(\beta t)},
\end{equation}
and hence,
\begin{equation}\label{symmbt}
    \frac{\partial f(\beta, t)}{\partial \beta}= \frac{\lambda t\sinh((\lambda-2)\beta t)-t\sinh(\lambda\beta t)}{\cosh^{\lambda+2}(\beta t)} = \frac{t}{\beta}\frac{\partial f(\beta, t)}{\partial t}.
\end{equation}

Following \cite{chatterjee}, we need to accomplish the following two steps in order to show that $\hat{\beta}_{\lambda}$ converges to $\beta$:
\begin{itemize}
    \item \textbf{Step 1:} Show that $\e_\beta S_{\bm X,\lambda}(\beta)^2$ converges to $0$ at a certain rate.
    
    \item  \textbf{Step-2:} Show that $S_{\bm X,\lambda}^\prime$ is bounded away from $0$ in a neighborhood of the true $\beta$ with high probability.
\end{itemize}

The main logic behind the above two steps, is the following heuristic argument:
$$0 = S_{\bm X,\lambda}(\hat{\beta}_{\lambda}) = S_{\bm X,\lambda}(\beta) + (\hat{\beta}_{\lambda} - \beta) S_{\bm X,\lambda}^\prime(\xi)$$
for some (random) $\xi$ lying between $\beta$ and $\hat{\beta}_\lambda$. We thus have:
\begin{equation}\label{heurst}
\hat{\beta}_{\lambda} - \beta = -\frac{S_{\bm X,\lambda}(\beta)}{S_{\bm X,\lambda}^\prime(\xi)},
\end{equation}
and the RHS of \eqref{heurst} will give the rate of convergence of $\hat{\beta}_\lambda$ to $\beta$.
Step-1 will ensure that the numerator of the RHS in \eqref{heurst} is $o_P(1)$, and Step-2 will ensure that the denominator is not too small. 

The following lemma, proved in Section \ref{prooflem1} of the appendix, accomplishes Step-1.
 
\begin{lem}\label{sx}
    For any value of the true parameter $\beta \geq 0$ and any tuning parameter $\lambda \geq 0$, we have:
    $$
        \E S_{\bm X,\lambda}(\beta)^2 \le \frac{C_{\lambda,\beta,\|\bm J\|}}{N},
    $$
    where $C_{\lambda,\beta,\|\bm J\|}$ is a constant depending only on $\lambda$, $\beta$ and $\|\bm J\|$. Consequently, $$\p(|S_{\bm X,\lambda}(\beta)|>\delta)\leq \frac{C_{\lambda,\beta,\|\bm J\|}}{N\delta^2}$$ for any $\delta>0$.
\end{lem}

Next, in order to execute Step-2, we will show that $S_{\bm X,\lambda}^\prime$ is strictly negative and bounded away from 0 with high probability. Towards this, we define:
\begin{equation}\label{sp1}
    S_{\bm X,\lambda}^{\prime}(\beta)_1:=\frac{1}{N}\sum_{i=1}^N \frac{\partial f(\beta, m_i(\bm X))}{\partial \beta} m_i(\bm X)(X_i - \tanh(\beta m_i(\bm X))),
\end{equation}
and 
\begin{equation}\label{sp2}
    S_{\bm X,\lambda}^{\prime}(\beta)_2:=-\frac{1}{N}\sum_{i=1}^N f(\beta, m_i(\bm X))\frac{m_i^2(\bm X)}{\cosh^2(\beta m_i(\bm X))}.
\end{equation}
Then by \eqref{Sx_prime}, we have:
$$ S_{\bm X,\lambda}^{\prime}(\beta) = S_{\bm X,\lambda}^{\prime}(\beta)_1 + S_{\bm X,\lambda}^{\prime}(\beta)_2~.$$
We now show that the first term concentrates around 0 at the true parameter value $\beta>0$ in the following lemma;
its proof is given in Section \ref{prooflem2} of the appendix. 

\begin{lem}\label{sbeta1}
For any $\beta \geq 0$ and $\lambda \geq 0$, we have:
    $$
        \E[S^{\prime}_{\bm X,\lambda}(\beta)_1^2]\leq \frac{D_{\lambda, \beta, \|\bm J\|}}{N},
    $$
    where $D_{\lambda,\beta,\|\bm J\|}$ is a constant depending only on $\lambda$, $\beta$ and $\|\bm J\|$. Consequently, $$\p(|S^{\prime}_{\bm X,\lambda}(\beta)_1|>\delta)\leq \frac{D_{\lambda,\beta,\|\bm J\|}}{N\delta^2}$$ for any $\delta>0$.
\end{lem}

Next, note that for any \textcolor{black}{$M>0$}, we have:
\begin{equation*}
\begin{aligned}
    S_{\bm X,\lambda}^{\prime}(\beta)_2 &=-\frac{1}{N}\sum_{i=1}^N f(\beta, m_i(\bm X))\frac{m_i^2(\bm X)}{\cosh^2(\beta m_i(\bm X))}\\
    & \leq -\frac{1}{N}\textcolor{black}{\sech^{\lambda+3}(\beta M)} \sum_{i=1}^N m_i(\bm X)^2 \bm{1}\{|m_i(\bm X)|\leq M\}.
\end{aligned}
\end{equation*}
Thus, we can control the term $S_{\bm X,\lambda}^{\prime}(\beta)_2$ by using the following lemma. 

\begin{lem}\label{mx}
Fix $0 < \delta < 1$. Then, there exist $\varepsilon = \varepsilon(\delta, \beta)>0$ and $M= M(\delta, \beta)<\infty$ such that
$$
\p\bigg(\sum_{i=1}^N m_i(\bm X)^2 \bm{1}\{|m_i(\bm X)|\leq M\} \geq \varepsilon N\bigg)\geq 1-\delta,
$$
for all N large enough.
\end{lem}

In view of the above lemmas, we are now finally ready to prove Theorem \ref{cons}. 
First, by Lemma \ref{sx}, we have: 
\begin{equation*}
    \p \bigg(|S_{\bm X, \lambda}(\beta)|>\frac{M_1}{\sqrt{N}} \bigg)\leq \frac{C}{M_1^2}.
\end{equation*}
for any constant $M_1>0$, where $C$ is as defined in the statement of Lemma \ref{sx}.
Now, fix $\delta>0$, whence it is possible to choose $M_1 = M_1(\delta, \beta, \lambda, \|\bm J\|) > 0$ such that the RHS above is less
than $\delta/3$. Similarly, by Lemma \ref{sbeta1}, it is also possible to choose $M_2 = M_2(\delta, \beta, \lambda, \|\bm J\|) > 0$ such that
\begin{equation*}
    \p \bigg( |S^\prime_{\bm X, \lambda}(\beta)_1|>\frac{M_2}{\sqrt{N}} \bigg)<\delta/3.
\end{equation*}
Next, we notice that by Lemma \ref{mx} there exists $\varepsilon = \varepsilon(\delta, \beta)>0$ and $M_3= M_3(\delta,\beta)>0$ such that
\begin{equation*}
    \p\bigg(\sum_{i=1}^N m_i(\bm X)^2 \bm{1}\{|m_i(\bm X)|\leq M_3\} \geq \varepsilon N\bigg)\geq 1-\delta/3,
\end{equation*}
for $N$ large enough. Define
\begin{equation}
\begin{aligned}
    T_N:=\bigg\{\bm X \in \{-1,1\}^N: |S_{\bm X, \lambda}(\beta)|\leq\frac{M_1}{\sqrt{N}}&, |S^\prime_{\bm X, \lambda}(\beta)_1|\leq\frac{M_2}{\sqrt{N}},\\
    & \sum_{i=1}^N m_i(\bm X)^2 \bm{1}\{|m_i(\bm X)|\leq M_3\} \geq 2\varepsilon N \bigg\}.
\end{aligned}
\end{equation}
Then we have $\p(T_N)\geq 1-\delta$, for $N$ large enough. For $\bm X \in T_N$, we have:
\begin{equation*}
\begin{aligned}
    \frac{\partial }{\partial \beta}S_{\bm X, \lambda}(\beta) & = S^\prime_{\bm X, \lambda}(\beta)_1 + S^\prime_{\bm X, \lambda}(\beta)_2\\
    & \leq \frac{M_2}{\sqrt{N}} -2\varepsilon~ \textcolor{black}{\sech^{\lambda+3}(\beta M_3)}\\
    & \leq -\varepsilon ~\textcolor{black}{\sech^{\lambda+3}(\beta M_3)}
\end{aligned}
\end{equation*}
for $N$ large enough. Then for $\bm X \in T_N$,
\begin{equation*}
\begin{aligned}
    \frac{M_1}{\sqrt{N}} \geq |S_{\bm X, \lambda}(\beta)|& =|S_{\bm X, \lambda}(\beta)-S_{\bm X, \lambda}(\hat{\beta}_\lambda(\bm X))|\\
    & \geq \int_{\beta \wedge \hat{\beta}_\lambda(\bm X)}^{\beta \vee \hat{\beta}_\lambda(\bm X)}-\frac{\partial }{\partial \beta}S_{\bm X, \lambda}(\beta)d\beta\\
    &\geq \textcolor{black}{\frac{\varepsilon}{M_3} \int_{M_3(\beta \wedge \hat{\beta}_\lambda(\bm X))}^{M_3(\beta \vee \hat{\beta}_\lambda(\bm X))} \sech^{\lambda+3}(u)~du}\\ &\geq  \textcolor{black}{\frac{\varepsilon}{M_3} \int_{M_3(\beta \wedge \hat{\beta}_\lambda(\bm X))}^{M_3(\beta \vee \hat{\beta}_\lambda(\bm X))} e^{-(\lambda+3) u}~du}\\ &= \frac{\varepsilon}{(\lambda+3)M_3}\left(e^{-(\lambda+3)M_3(\beta \wedge \hat{\beta}_\lambda(\bm X))} - e^{-(\lambda+3)M_3(\beta \vee \hat{\beta}_\lambda(\bm X))}\right)\\ &= \frac{\varepsilon}{(\lambda+3)M_3} \left|e^{-(\lambda+3)M_3\hat{\beta}_\lambda(\bm X)} - e^{-(\lambda+3)M_3\beta}\right|
\end{aligned}
\end{equation*}

Now, defining $M=M(\delta,\beta,\lambda):={(\lambda+3)M_1M_3}/{ \varepsilon} $, we have
\begin{equation}\label{tightbd78}
    \p\left(\sqrt{N}\left|e^{-(\lambda+3)M_3\hat{\beta}_\lambda(\bm X)} - e^{-(\lambda+3)M_3\beta}\right|\leq M\right)\geq 1-\delta.
\end{equation}

It now follows from \eqref{tightbd78} and an application of Lemma \ref{prokh} of Appendix C, that for all $N$ large enough, with probability at least $1-\delta$, we have:
$$|\hat{\beta}_\lambda(\bm X) - \beta| \le \frac{Me^{2(\lambda+3) M_3\beta}}{(\lambda+3)M_3\sqrt{N}} = \frac{M_1 e^{2(\lambda+3) M_3\beta}}{\varepsilon \sqrt{N}}$$
which completes the proof of Theorem \ref{cons}. 
\qed

\subsection{Proof of Theorem \ref{thm:consiZ}}\label{sec:cltproof921}
To begin with, define:
$$h_N(\tilde{\beta}) := \frac{1}{N}\sum_{i=1}^N g_{\tilde{\beta}}(m_i(\bm X))(X_i - \tanh(\tilde{\beta} m_i(\bm X))),\quad \text{and}$$ $$h(\tilde{\beta}) := g_{\tilde{\beta}}(M(\bm X))(M(\bm X)-\tanh(\tilde{\beta} M(\bm X))).$$
It follows from Assumption \ref{asd0}, a first order Taylor expansion, and the boundedness of the derivative of $g_{\beta}$ on bounded sets, that for any fixed $\tilde{\beta}$,
\begin{equation*}
\begin{aligned}
    & \left|\frac{1}{N}\sum_{i=1}^Ng_{\tilde{\beta}}(m_i(\bm X))\tanh(\Tilde{\beta}m_i(\bm X)) -g_{\tilde{\beta}}(M(\bm X))\tanh(\Tilde{\beta}M(\bm X))\right|\\
    & ~~~\lesssim  \frac{1}{N}\sum_{i=1}^N  |m_i(\bm X)-M(\bm X)|\xrightarrow{\p_{\beta}}0.
\end{aligned}
\end{equation*}
Moreover, by Assumption \ref{asd0}, we have the following for any fixed $\tilde{\beta}$:
\begin{equation*}
\begin{aligned}
    &\left|\frac{1}{N}\sum_{i=1}^N g_{\tilde{\beta}}(m_i(\bm X)) X_i -g_{\tilde{\beta}}(M(\bm X)) M(\bm X)\right|\\
    = & \left|\frac{1}{N}\sum_{i=1}^N(g_{\tilde{\beta}}(m_i(\bm X)) X_i - g_{\tilde{\beta}}(M(\bm X)) X_i)\right|+\left|\frac{1}{N}\sum_{i=1}^N g_{\tilde{\beta}}(M(\bm X))(X_i - M(\bm X))\right|\\
    \leq & \frac{1}{N}\sum_{i=1}^N|g_{\tilde{\beta}}(m_i(\bm X))-g_{\tilde{\beta}}(M(\bm X))| + \left| g_{\tilde{\beta}}(M(\bm X))(\bar{X}-M(\bm X))\right|\\ \lesssim & \frac{1}{N}\sum_{i=1}^N  |m_i(\bm X) - M(\bm X)| + |\bar{X} - M(\bm X)| \xrightarrow{\p_\beta} 0.
\end{aligned}
\end{equation*}
This shows that $h_N(\tilde{\beta}) \xrightarrow{\p_\beta} h(\tilde{\beta})$ for all $\tilde{\beta}$ in some neighborhood of the true $\beta$. Since $g$ is an odd function, we actually have $h(\tilde{\beta}) = g_{\tilde{\beta}}(m_+)(m_+-\tanh(\tilde{\beta} m_+))$, where recall that $m_+ = m_+(\beta)$ is the solution of $\tanh(\beta m)=m$. For $\beta>1$, there exists a neighborhood of $\beta$ such that $h$ changes sign on either sides of $\beta$ within this neighborhood. In particular, for all $\varepsilon >0$ sufficiently small, $h(\beta-\varepsilon)$ and $h(\beta+\varepsilon)$ have different (non-zero) signs, which implies that as $N \rightarrow \infty$, $h_N(\beta -\varepsilon)$ and $h_N(\beta+\varepsilon)$ have different signs with probability $1-o(1)$. $\hat{\beta}_g$ being the root of $h_N$, must thus satisfy the following for all $\varepsilon>0$ sufficiently small:  $$\p_\beta(\beta-\varepsilon < \hat{\beta}_g < \beta+\varepsilon) \rightarrow 1$$ which implies that $\hat{\beta}_g \xrightarrow{\p_\beta} \beta$. This completes the proof of Theorem \ref{thm:consiZ}.

\subsection{Proof of Theorem \ref{cltthm8}}\label{sec:cltproof92}
We start by looking at a first order Taylor expansion of the left-hand side of 
the general $Z$-estimating equation in (\ref{z-est-ee}), which yields:
\begin{eqnarray*}
        &&\frac{1}{\sqrt{N}}\sum_{i=1}^N g(m_i(\bm X))(X_i - \tanh(\beta m_i(\bm X)))\\ &=& \frac{g(M(\bm X))}{\sqrt{N}}\sum_{i=1}^N (X_i - \tanh(\beta m_i(\bm X)))\\&+& \frac{1}{\sqrt{N}}\sum_{i=1}^N (m_i(\bm X)- M(\bm X))g'(\xi_i(\bm X)) (X_i - \tanh(\beta m_i(\bm X)))
\end{eqnarray*}
for random variables $\xi_i(\bm X)$ lying between $m_i(\bm X)$ and $M(\bm X)$. 
By the boundedness of $g'$ on bounded sets and Assumption \ref{asd1} (1), we thus have:
$$\frac{1}{\sqrt{N}}\sum_{i=1}^N g(m_i(\bm X))(X_i - \tanh(\beta m_i(\bm X))) = \frac{g(M(\bm X))}{\sqrt{N}}\sum_{i=1}^N (X_i - \tanh(\beta m_i(\bm X))) + o_P(1).$$
It now follows from Assumption \ref{asd3} that:
\begin{equation}\label{cltfr68}
    \frac{1}{\sqrt{N}}\sum_{i=1}^N g(m_i(\bm X))(X_i - \tanh(\beta m_i(\bm X))) \xrightarrow{D} \mathcal{N}\left(0,g(m_+)^2\sigma_\beta^2\right).
\end{equation}
Now, if we define $f_i(\beta) = f_{i,\bm X}(\beta) := g_\beta(m_i(\bm X))(X_i - \tanh(\beta m_i(\bm X)))$,
then 
\begin{equation}\label{mainjh34}
 \frac{1}{\sqrt{N}}\sum_{i=1}^N f_i(\beta) = \frac{\hat{\beta}_g - \beta}{\sqrt{N}}\sum_{i=1}^N f_i'(\beta) + \frac{(\hat{\beta}_g - \beta)^2}{2N}\sum_{i=1}^N f_i''(\xi_i(\beta))
\end{equation}
 for random variables $\xi_i(\beta)$ lying between $\beta$ and $\hat{\beta}_g$. Now, note that:
$$f_i'(\beta) = -m_i(\bm X) g_\beta(m_i(\bm X)) \sech^2 (\beta m_i(\bm X)) +  (X_i - \tanh(\beta m_i(\bm X))) \partial_\beta g_{\beta}(m_i(\bm X)).$$
For the first term above, we get
\begin{eqnarray*}
   && \left|\frac{1}{N}\sum_{i=1}^N m_i(\bm X) g_\beta(m_i(\bm X)) \sech^2 (\beta m_i(\bm X)) - M(\bm X) g_\beta(M(\bm X))\sech^2(\beta M(\bm X))\right|\\&\le& \frac{1}{N}\sum_{i=1}^N \left|m_i(\bm X) g_\beta(m_i(\bm X)) \sech^2 (\beta m_i(\bm X)) -  M(\bm X) g_\beta(M(\bm X))\sech^2(\beta M(\bm X))\right|\\&\lesssim& \frac{1}{N}\sum_{i=1}^N \left|m_i(\bm X) - M(\bm X)\right| \le \sqrt{\frac{1}{N}\sum_{i=1}^N \left(m_i(\bm X) - M(\bm X)\right)^2} = o_P(1),
\end{eqnarray*}
where the second inequality followed from the boundedness of the derivative of the function $xg(x) \sech^2(x)$ on bounded sets, and the third inequality follows from the Cauchy-Schwarz inequality. Since $g$ is an odd function, $xg(x)$ is an even function, and hence
$$ M(\bm X) g_\beta(M(\bm X))\sech^2(\beta M(\bm X)) = m_+ g_\beta(m_+)\sech^2(\beta m_+).$$
On the other hand, by Assumption \ref{asd1} (2), we also have:
$$\frac{1}{N}\sum_{i=1}^N  (X_i - \tanh(\beta m_i(\bm X))) \partial_\beta g_{\beta}(m_i(\bm X)) = o_P(1).$$
Therefore, we get
\begin{equation}\label{firstterm7762}
    \frac{\hat{\beta}_g - \beta}{\sqrt{N}}\sum_{i=1}^N f_i'(\beta) = -\sqrt{N}(\hat{\beta}_g - \beta)\left(m_+ g_\beta(m_+)\sech^2(\beta m_+) + o_P(1)\right).
\end{equation}

Finally, it follows from the uniform boundedness of $f_i''$ on some bounded neighborhood of $\beta$, together with the consistency of $\hat{\beta}_g$ (Theorem \ref{thm:consiZ}), that:
\begin{equation}\label{secterm7762}
    \frac{(\hat{\beta}_g - \beta)^2}{2N}\sum_{i=1}^N f_i''(\xi_i(\beta)) = o_P(1).
\end{equation}
Combining \eqref{cltfr68}, \eqref{mainjh34}, \eqref{firstterm7762} and \eqref{secterm7762}, we thus have:
$$ -\sqrt{N}(\hat{\beta}_g - \beta)\left(m_+ g(m_+)\sech^2(\beta m_+) + o_P(1)\right) + o_P(1) \xrightarrow{D} \mathcal{N}\left(0,g(m_+)^2\sigma_\beta^2\right).$$
By Slutsky's theorem, we thus have:
$$\sqrt{N}(\hat{\beta}_g - \beta) \xrightarrow{D} \mathcal{N}\left(0,\frac{\sigma_\beta^2}{m_+^2\sech^4(\beta m_+)}\right).$$
Theorem \ref{cltthm8} \eqref{z-est92} now follows on observing that $$\sech^4(\beta m_+) = (1-\tanh^2(\beta m_+))^2 = (1-m_+^2)^2.$$
Since MDPD estimators are special cases of the class of general $Z$-estimators considered, 
this completes the proof of Theorem \ref{cltthm8}.\qed

\section{Simulation Results}
\label{sec:simulatons}

In this section, we demonstrate the performances of the MDPD estimators based on data simulated from (possibly contaminated) Ising models on different graph ensembles. 

\subsection{Ising models on 1-D and 2-D lattices}
To begin with, we generate (using Gibbs sampler) $1000$ samples from the Ising model on a 1-D lattice with $N=2000$ vertices and $\beta=0.5$. We then contaminate the data by setting a certain fraction of entries of each sample to $+1$.  The (average) mean squared error (MSE) and (average) bias of the MDPD estimators $\hat{\beta}_\lambda$ with $\lambda\in \{0,0.1,...,1\}$ are given in Figure \ref{RG}. We see that the MSE of $\hat{\beta}_\lambda$ slightly decreases for strongly contaminated ($40\%$) data as $\lambda$ increases. For data with less contamination or no contamination, MSE of $\hat{\beta}_\lambda$ slightly increases as $\lambda$ increases. Meanwhile, the average bias of the MDPD estimators from contaminated data decreases slightly as $\lambda$ increases. This suggests that the MDPD estimators reduce bias for contaminated data, at the cost of increasing variance. This trade-off is more apparent when the model is sensitive to contamination. 

\begin{figure}[h]
    \centering
    \begin{subfigure}[b]{0.49\textwidth}
        \centering
        \includegraphics[width=8cm]{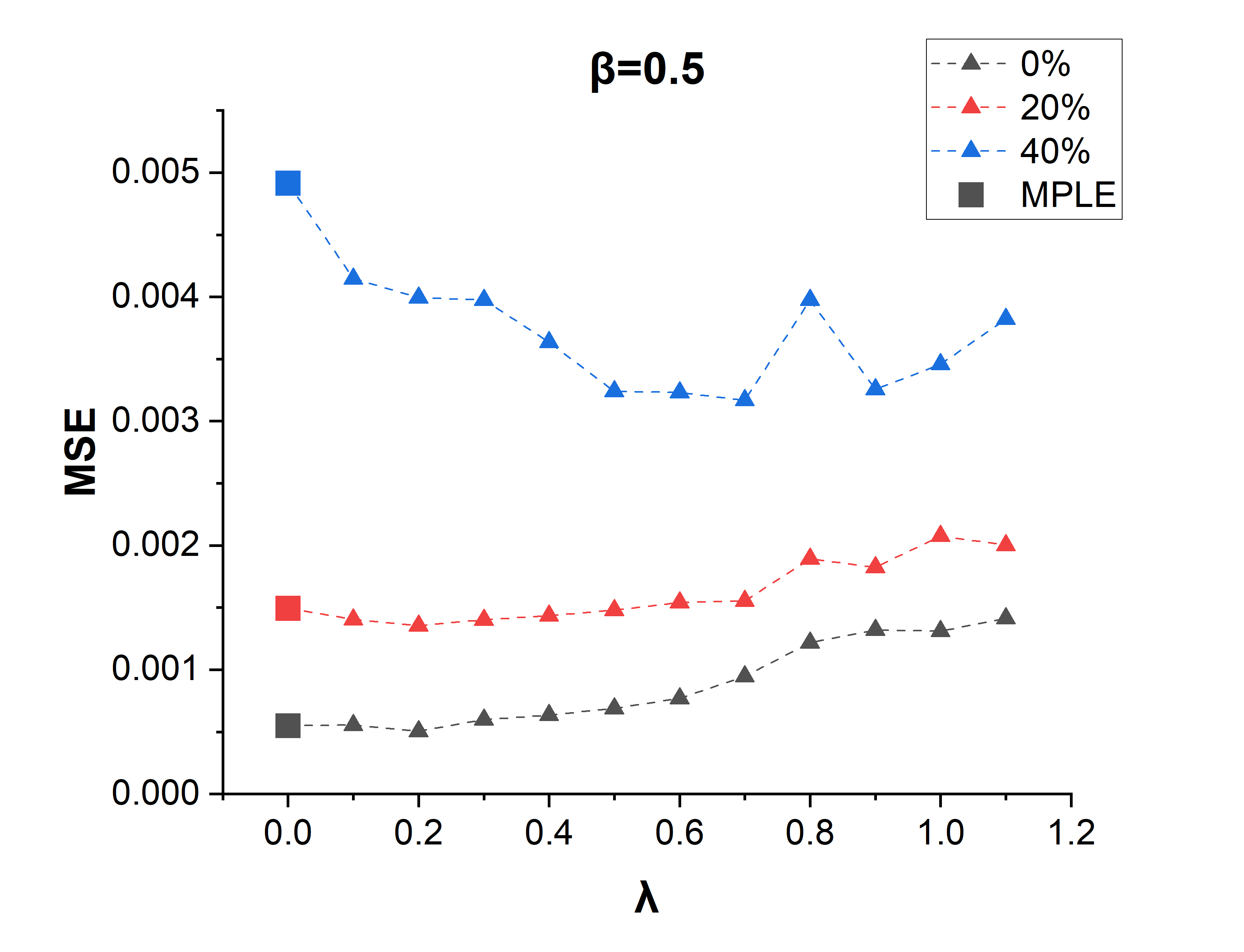}
        \caption[Network2]%
            {{\small MSE}}    
    \end{subfigure}
    \hfill
    \begin{subfigure}[b]{0.49\textwidth}  
        \centering 
        \includegraphics[width=8cm]{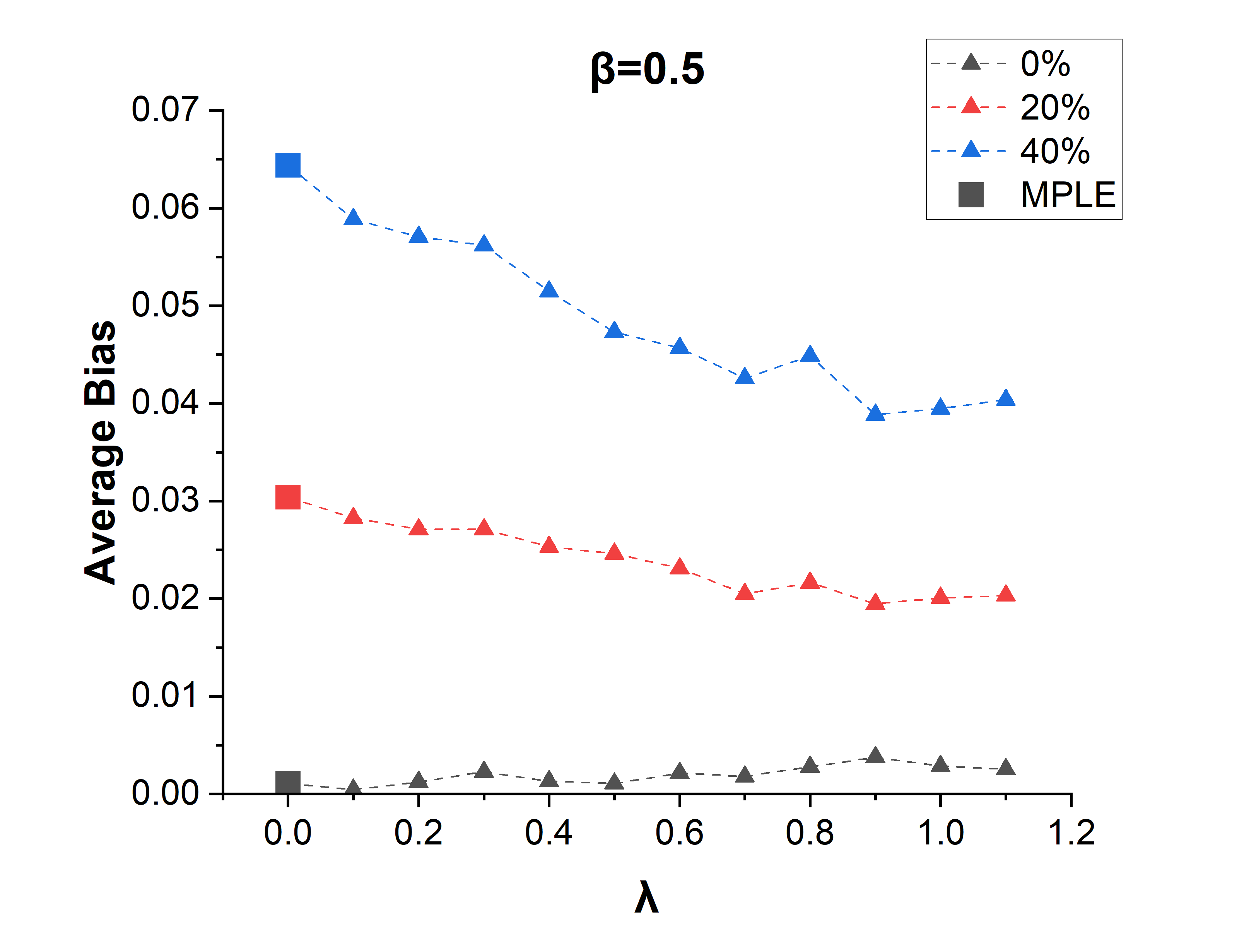}
        \caption[]%
            {{\small Bias}}    
    \end{subfigure}
    \caption{MSE and bias for estimates from 1-D lattice, with $\beta = 0.5$, $N=2000$ and contamination levels set to $0\%$, $20\%$ and $40\%$.}
    \label{RG}
\end{figure}

Similar simulations are performed for Ising models on 2-D lattices with $N=m^2 := 6400$ vertices and $\beta=0.5$. $1000$ samples are generated from this model using a Gibbs sampler, and once again, a certain fraction of entries of each sample are perturbed to $+1$. The (average) MSE and (average) bias of the MDPD estimators $\hat{\beta}_\lambda$ with $\lambda\in \{0,0.1,...,1.1\}$ are given in Figure \ref{GRID}. We observe that for contaminated data, both the MSE and bias of $\hat{\beta}_\lambda$ decrease as $\lambda$ increases, and this trend is more apparent when one goes from 1-D to 2-D lattices.


\begin{figure}[h]
    \centering
    \begin{subfigure}[b]{0.49\textwidth}
        \centering
        \includegraphics[width=8cm]{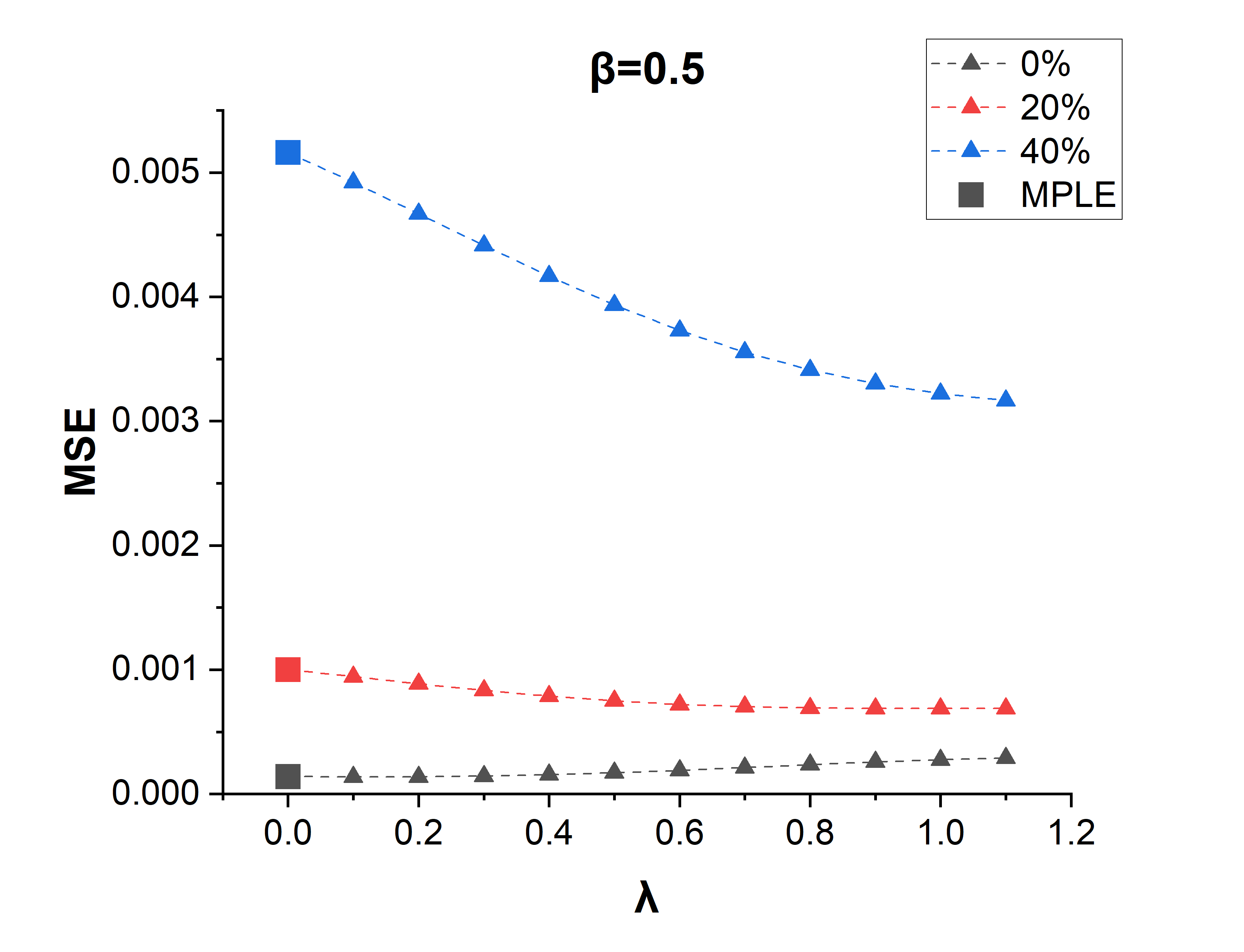}
        \caption[Network2]%
            {{\small MSE}}    
    \end{subfigure}
    \hfill
    \begin{subfigure}[b]{0.49\textwidth}  
        \centering 
        \includegraphics[width=8cm]{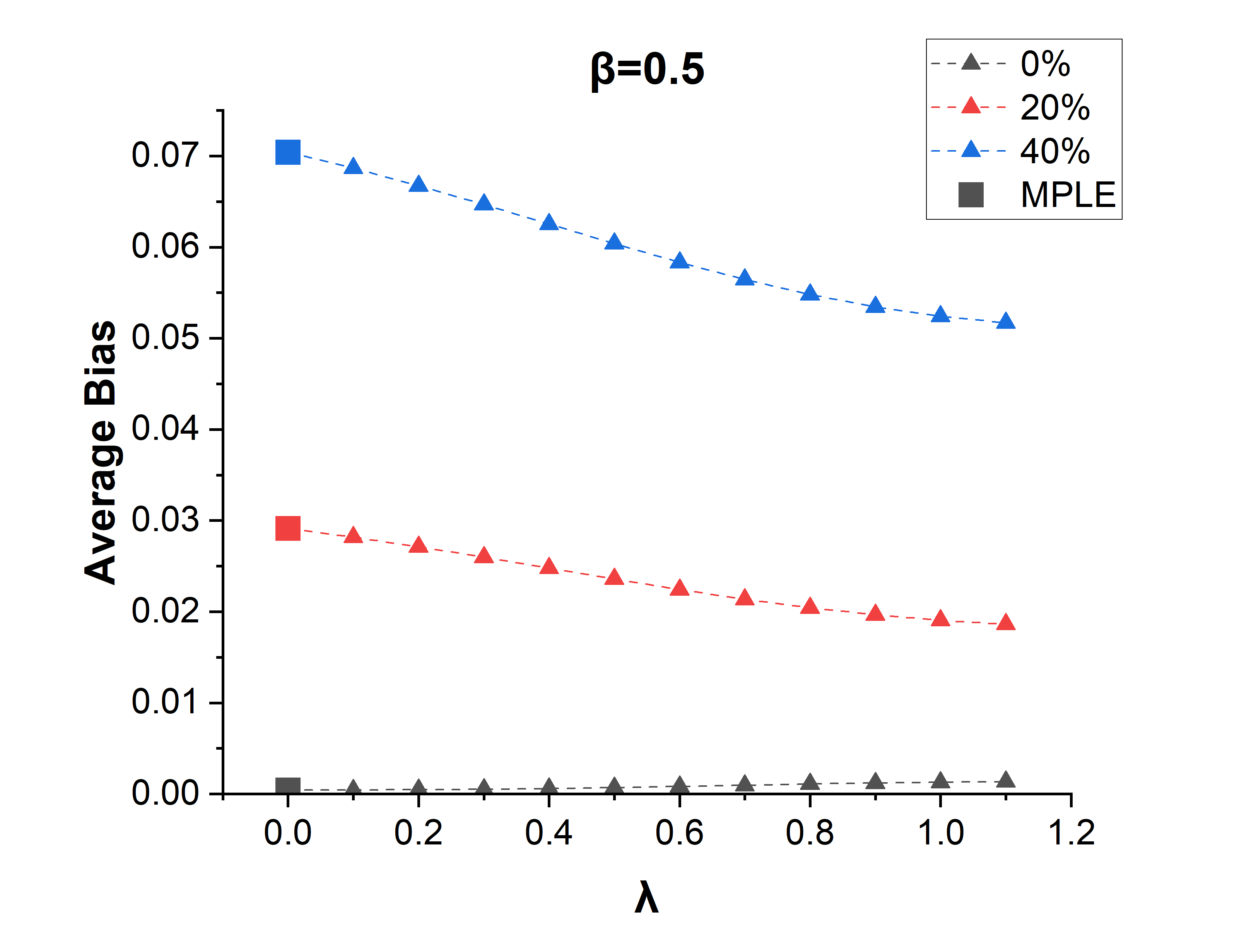}
        \caption[]%
            {{\small Bias}}    
    \end{subfigure}
    \caption{MSE and bias for estimates from 2-D lattice, with $\beta = 0.5$, $N=6400$ and contamination levels set to $0\%$, $20\%$ and $40\%$.}
    \label{GRID}
\end{figure}

\subsection{Ising models on Erd\H{o}s–R\'enyi random graphs}

We now consider the Ising model on sparse Erd\H{o}s–R\'enyi random graphs $G(N,p(N))$ with $N := 2000$. The Hamiltonian of such a model is given by: 
$$
H(\bm x) = \frac{1}{N p(N)}\bm x^\top A(G_N) \bm x,
$$
where $A(G_N)$ is the (random) adjacency matrix of $G(N,p(N))$. 

We start with $p(N):= 5/N$, $\beta:= 0.8$, and vary the parameter $\lambda$ from $0.1$ to $1$ in increments of $0.1$. After generating $1000$ samples from this model, a certain fraction of entries in each sample is flipped (multiplied by $-1$). The (average) MSE and (average) bias of the MDPD estimates are shown in Figure \ref{ERM}. We observe that both the MSE and the bias of $\hat{\beta}_\lambda$ decrease significantly for contaminated data as $\lambda$ increases. It suggests that MDPD estimators perform better then MPLE for contaminated data. For data with no contamination, MSE and bias of $\hat{\beta}_\lambda$ do not vary noticeably.

\begin{figure}[h]
    \centering
    \begin{subfigure}[b]{0.49\textwidth}
        \centering
        \includegraphics[width=8cm]{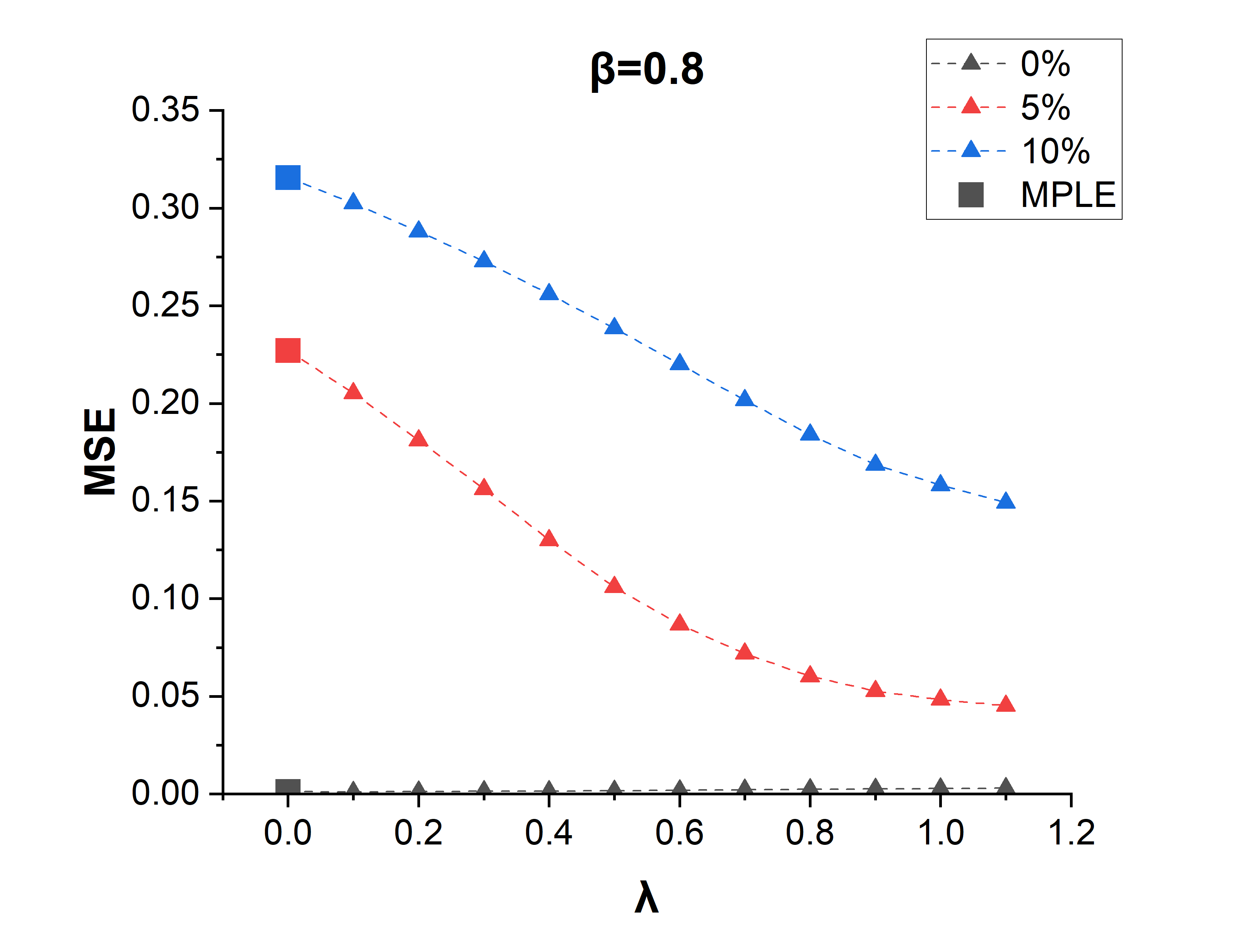}
        \caption[Network2]%
            {{\small MSE}}    
    \end{subfigure}
    \hfill
    \begin{subfigure}[b]{0.49\textwidth}  
        \centering 
        \includegraphics[width=8cm]{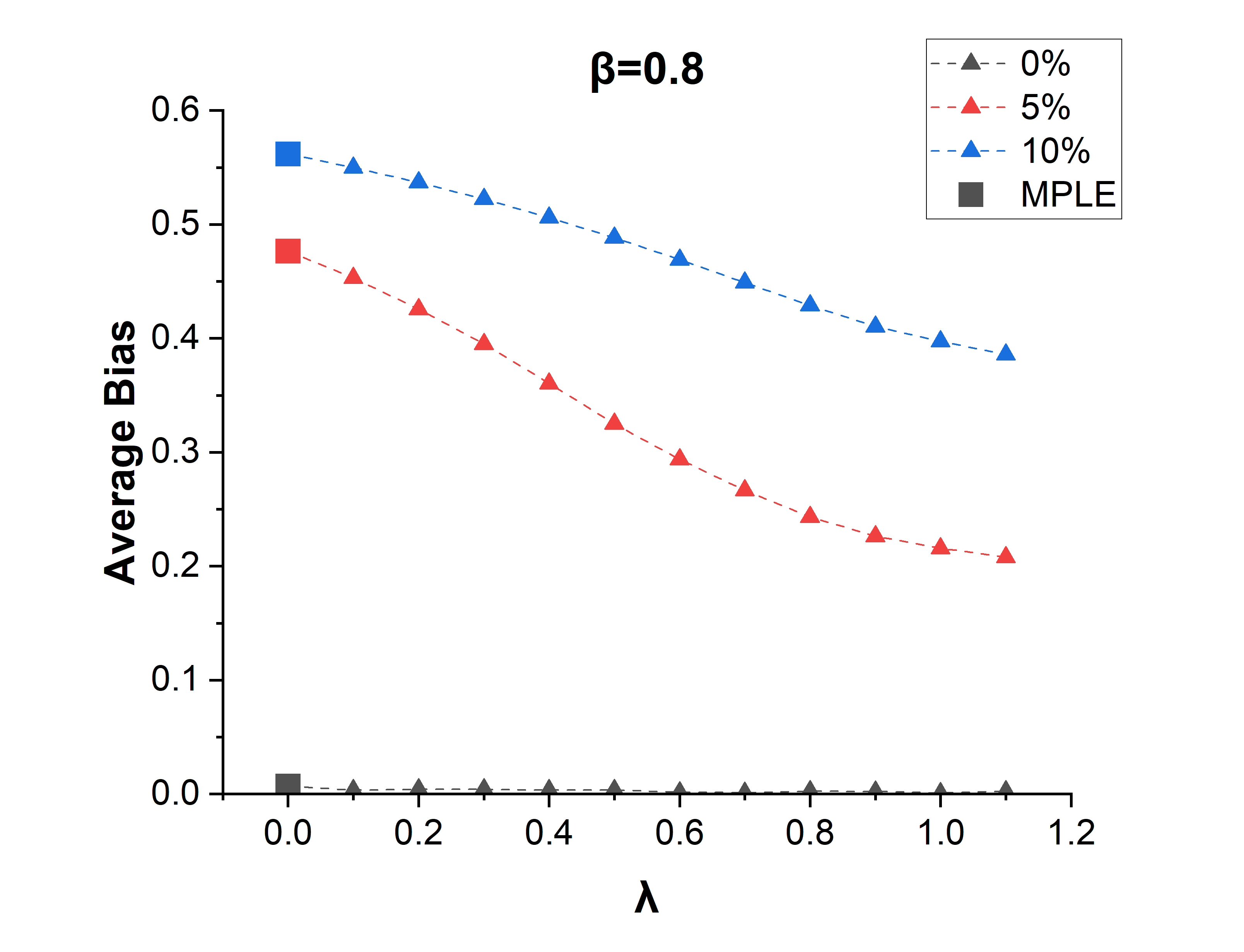}
        \caption[]%
            {{\small Bias}}    
    \end{subfigure}
    \caption{MSE and bias for estimates from Erdős–Rényi random graph with $p=\frac{5}{N}$, $\beta = 0.8$ and $N=2000$.}
    \label{ERM}
\end{figure}

In Figure \ref{ERM1}, we present the MSE and bias of $\hat{\beta}_\lambda$ when the data is simulated from Ising models on $G(N,5/N)$ with different values of $\beta$, and contaminated by $5\%$. Once again, the MSE and bias of $\hat{\beta}_\lambda$ decrease significantly as $\lambda$ increases. Additionally, we note  for every $\lambda$, that the MSE and bias of $\hat{\beta}_\lambda$ are much higher for larger values of $\beta$.

\begin{figure}[h]
    \centering
    \begin{subfigure}[b]{0.49\textwidth}
        \centering
        \includegraphics[width=8cm]{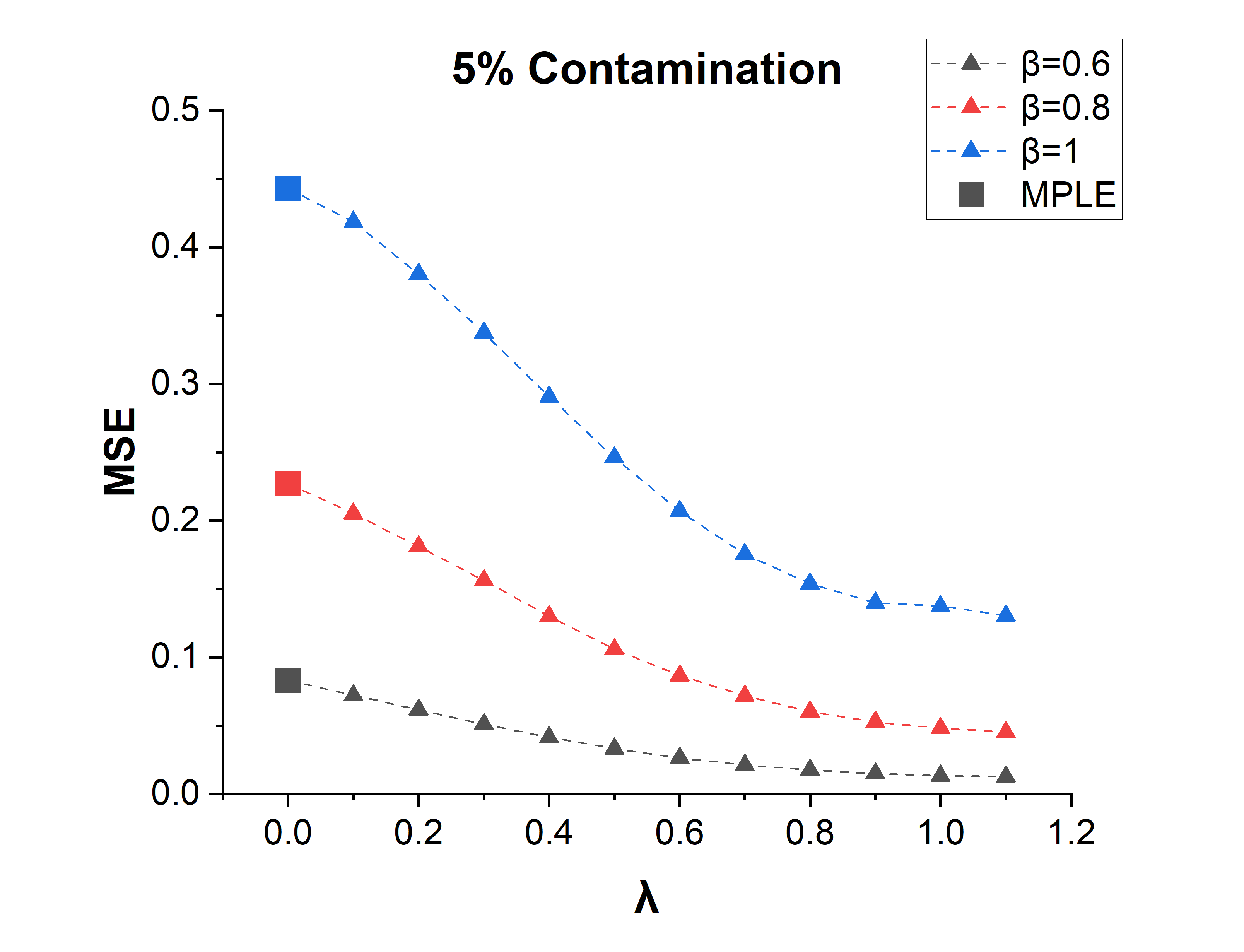}
        \caption[Network2]%
            {{\small MSE}}    
    \end{subfigure}
    \hfill
    \begin{subfigure}[b]{0.49\textwidth}  
        \centering 
        \includegraphics[width=8cm]{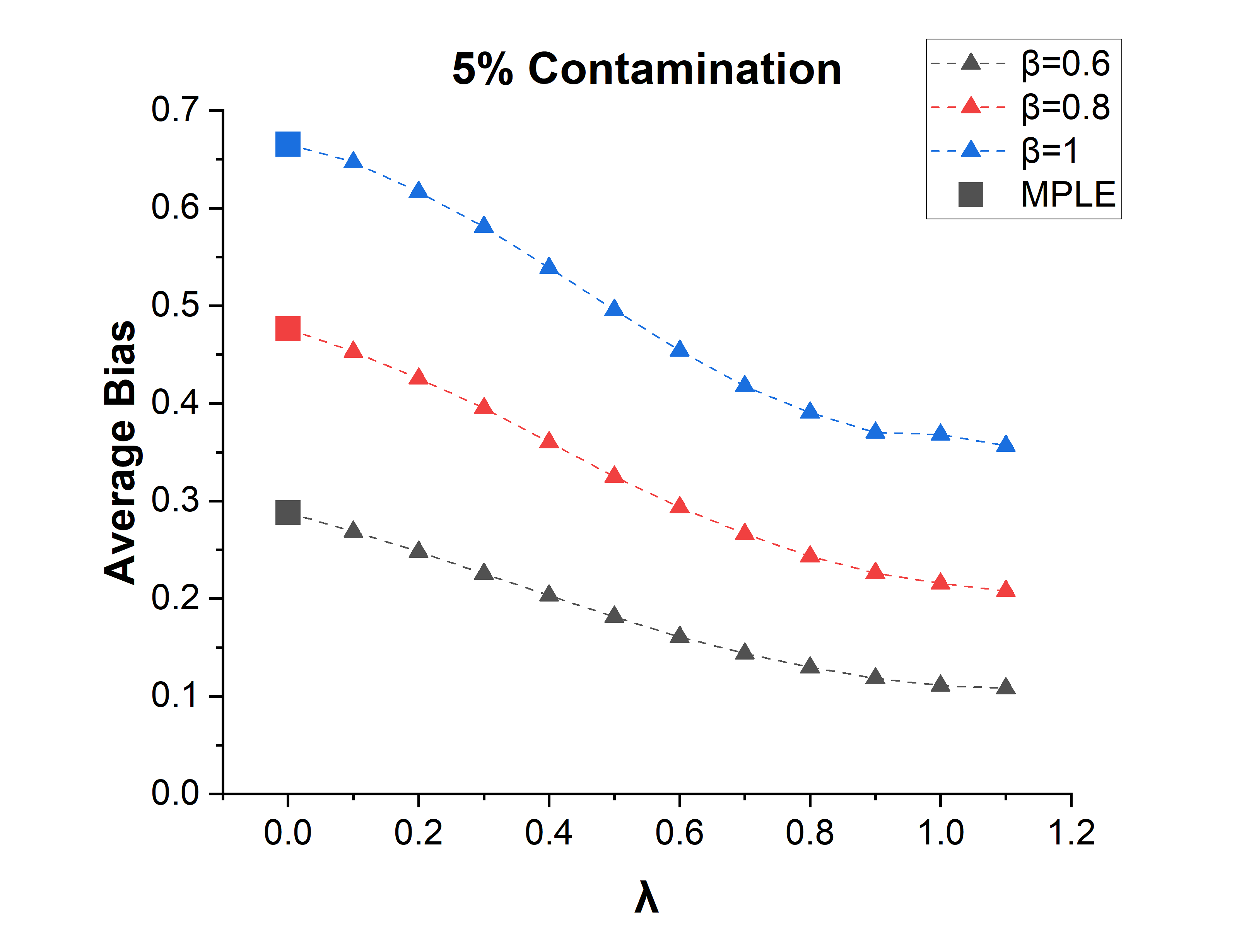}
        \caption[]%
            {{\small Bias}}    
    \end{subfigure}
    \caption{MSE and bias for estimates from Erdős–Rényi random graph with varying $\beta$, $p=\frac{5}{N}$ and $N=2000$.}
    \label{ERM1}
\end{figure}

We also run simulations for Erd\H{o}s–R\'enyi graphs with larger scaling of $p(N)$. In Figure \ref{DERM2}, we take $p :=\frac{\log N}{N}$ and $\beta=1.5>1$ (keeping in mind that the MPLE is $\sqrt{N}$-consistent for $\beta>1$; see Corollary 3.2 in \cite{BM16}). The same decreasing trend for the MSE and bias of the MDPD estimates is observed.

\begin{figure}[h]
    \centering
    \begin{subfigure}[b]{0.49\textwidth}
        \centering
        \includegraphics[width=8cm]{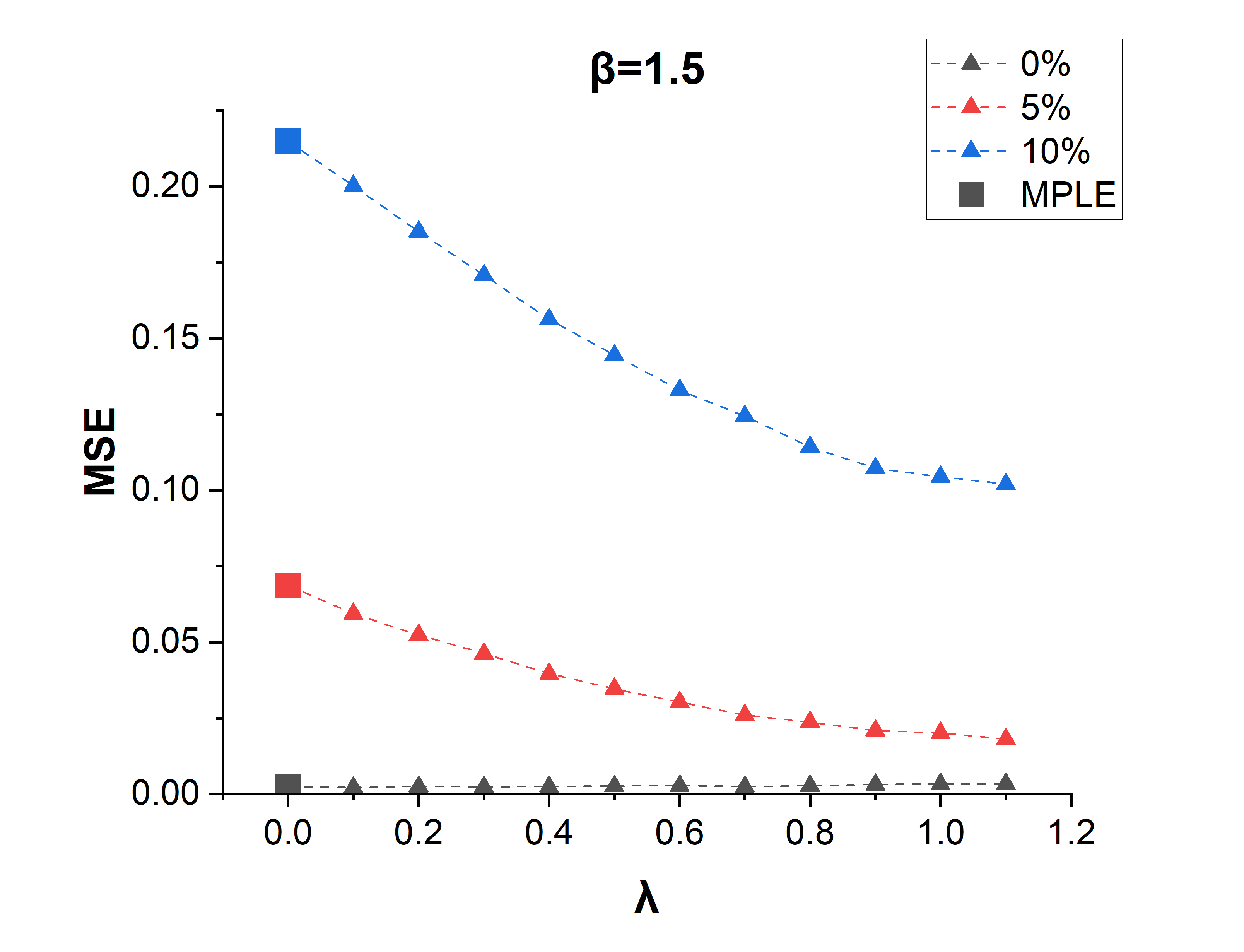}
        \caption[Network2]%
            {{\small MSE}}    
    \end{subfigure}
    \hfill
    \begin{subfigure}[b]{0.49\textwidth}  
        \centering 
        \includegraphics[width=8cm]{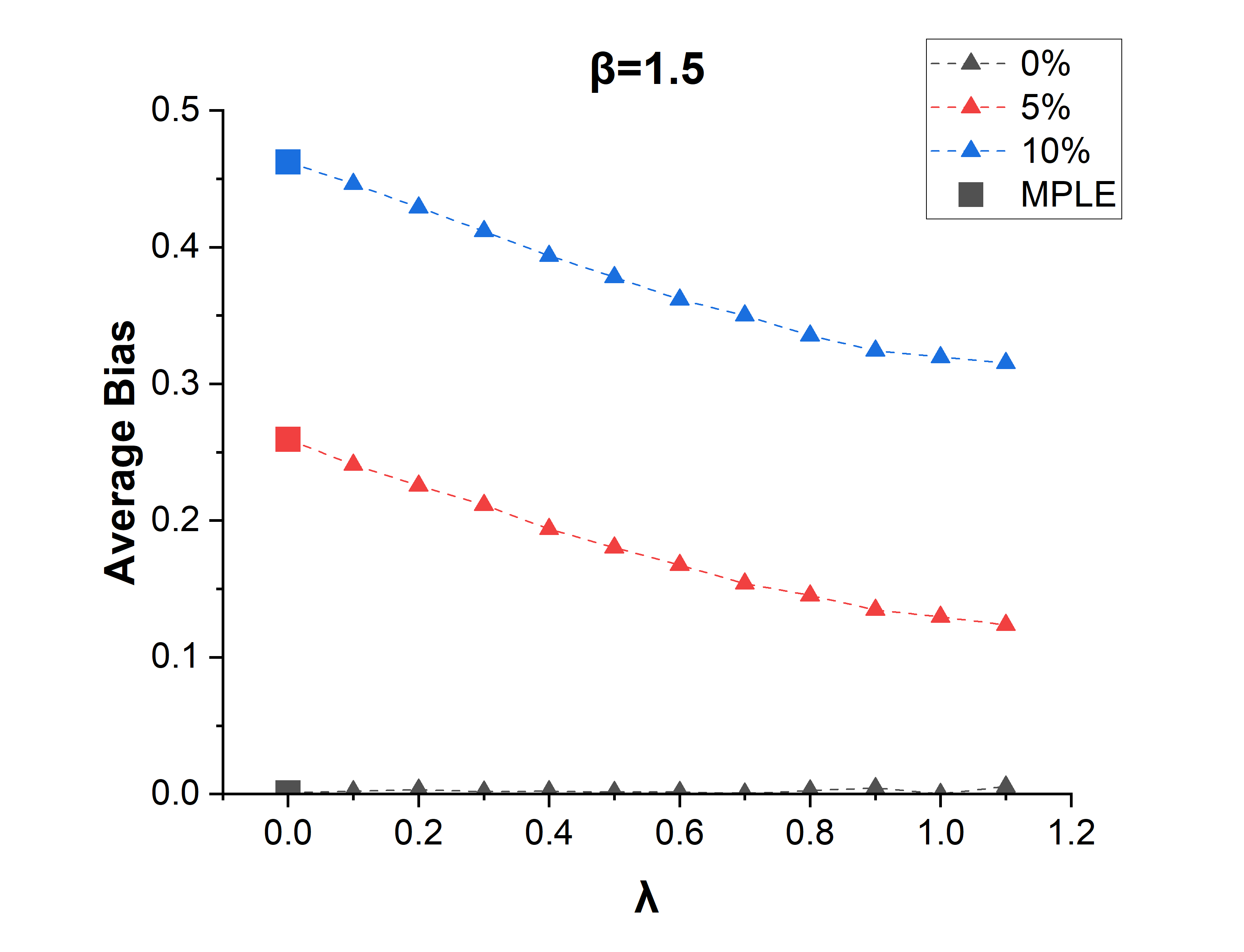}
        \caption[]%
            {{\small Bias}}    
    \end{subfigure}
    \caption{MSE and bias for estimates from Erdős–Rényi graph with $p=\frac{\log N}{N}$,  $\beta = 1.5$ and $N=2000$.}
    \label{DERM2}
\end{figure}

\subsection{Ising models on stochastic block models}
We then conduct simulations from Ising models on stochastic block models (SBM) with $N := 2000$ nodes and two communities of equal size. Within-community connection probabilities are set to $p := \frac{4\log N}{N}$ and the between community connection probability is set to $q:= \frac{\log N}{N}$.
Contamination is introduced to $5\%$ or $10\%$ of the entries of each of the $1000$ samples generated, with the corresponding entries being flipped (multiplied by $-1$). Once again, we observe (in Figure \ref{SBM}) a clear decreasing trend in both the (average) MSE and (average) bias of the MDPD estimates as $\lambda$ increases.



\subsection{The Sherrington–Kirkpatrick model}
 In the
Sherrington–Kirkpatrick model, we have $J_{ij}=N^{-\frac{1}{2}}g_{ij}$, where $(g_{ij})_{1\leq i <j \leq \infty}$ is a fixed realization of a symmetric array of independent standard Gaussian random variables. In this case as well, as mentioned in Example \ref{skmodel_ex}, the MDPD estimators are $\sqrt{N}$-consistent for all $\beta>0$. 

We simulate $1000$ samples from this model with $N:= 2000$ and $\beta := 1$, and contaminate $5\%$ or $10\%$ of each generated sample by flipping the signs of the corresponding entries. As seen in Figure \ref{SKM}, both the (average) MSE and the (average) bias of the MDPD estimates decrease, but this time very slowly, as $\lambda$ increases.

\begin{figure}[h]
    \centering
    \begin{subfigure}[b]{0.49\textwidth}
        \centering
        \includegraphics[width=8cm]{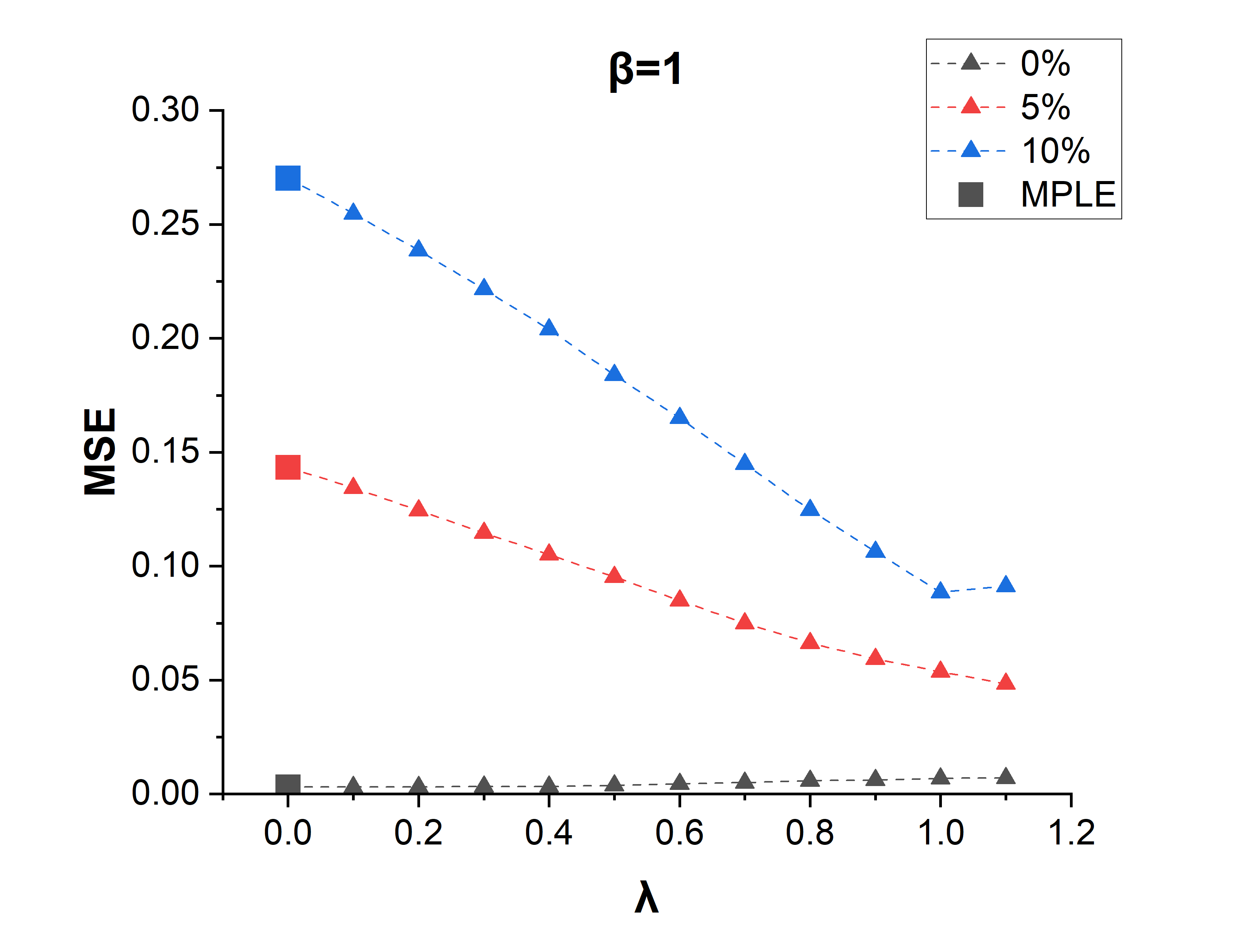}
        \caption[Network2]%
            {{\small MSE}}    
    \end{subfigure}
    \hfill
    \begin{subfigure}[b]{0.49\textwidth}  
        \centering 
        \includegraphics[width=8cm]{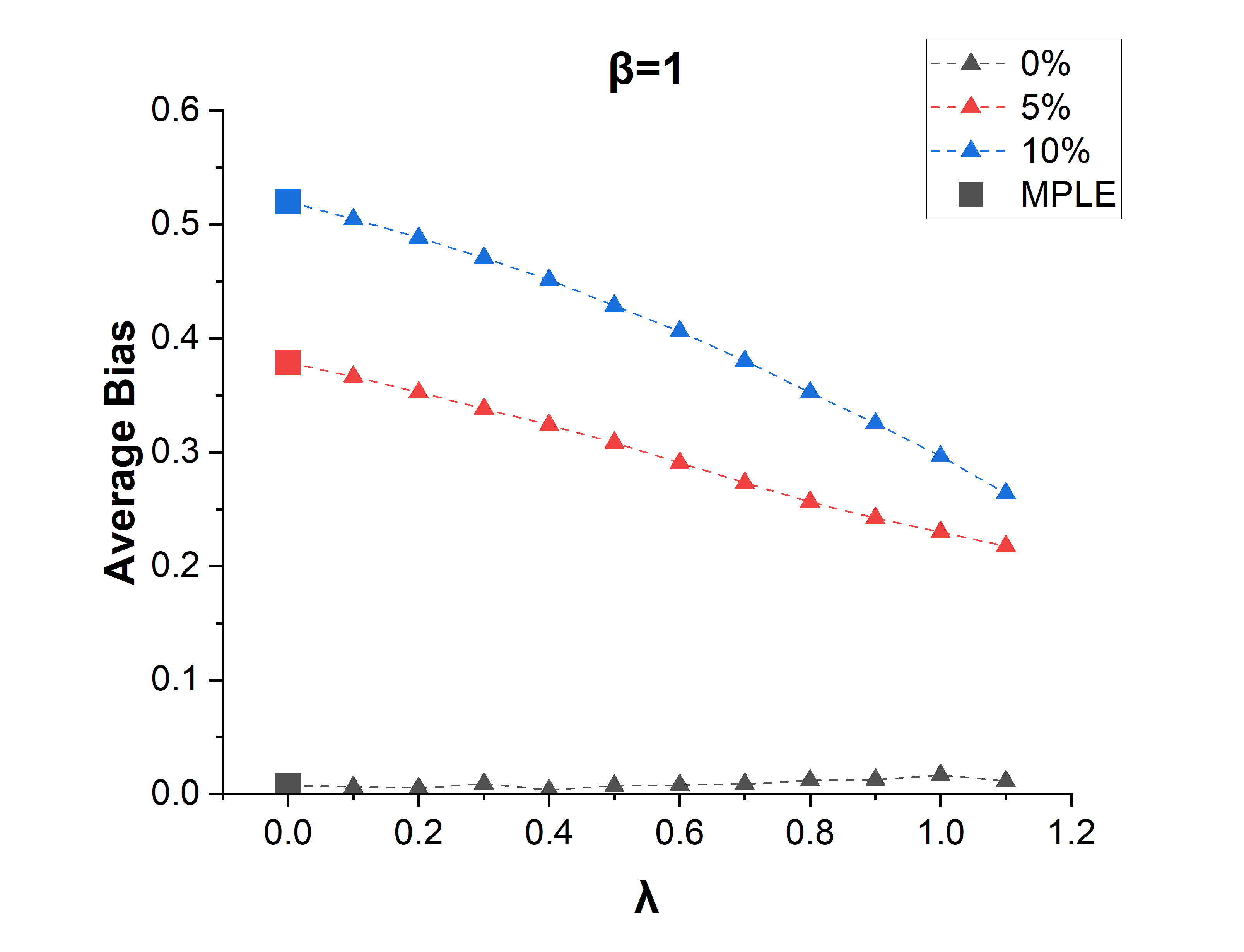}
        \caption[]%
            {{\small Bias}}    
    \end{subfigure}
    \caption{MSE and bias for estimates from the stochastic block model, with $p=\frac{4\log N}{N}$, $q=\frac{\log N}{N}$, $\beta = 1$ and $N=2000$.}
    \label{SBM}
\end{figure}

\begin{figure}[h]
    \centering
    \begin{subfigure}[b]{0.49\textwidth}
        \centering
        \includegraphics[width=8cm]{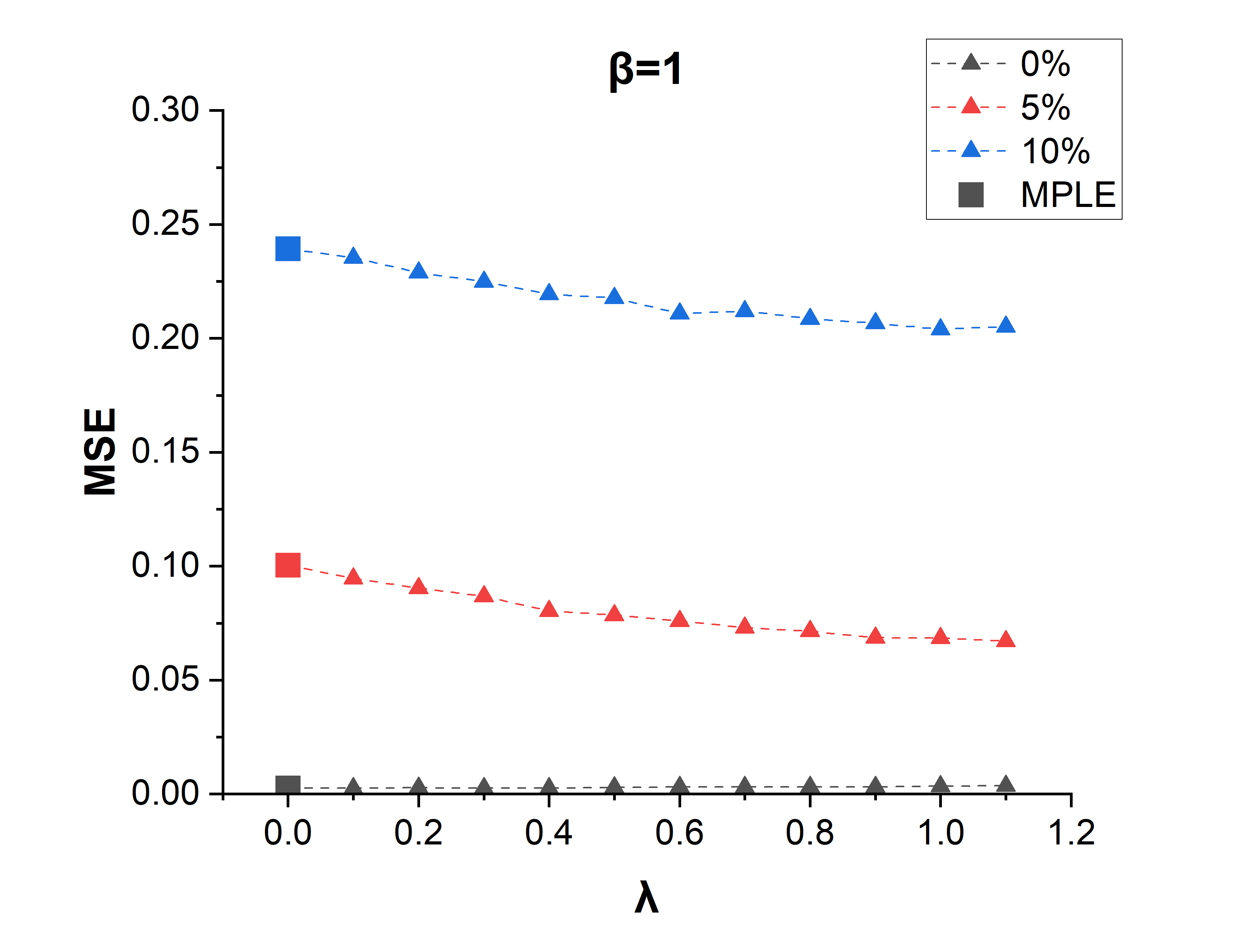}
        \caption[Network2]%
            {{\small MSE}}    
    \end{subfigure}
    \hfill
    \begin{subfigure}[b]{0.49\textwidth}  
        \centering 
        \includegraphics[width=8cm]{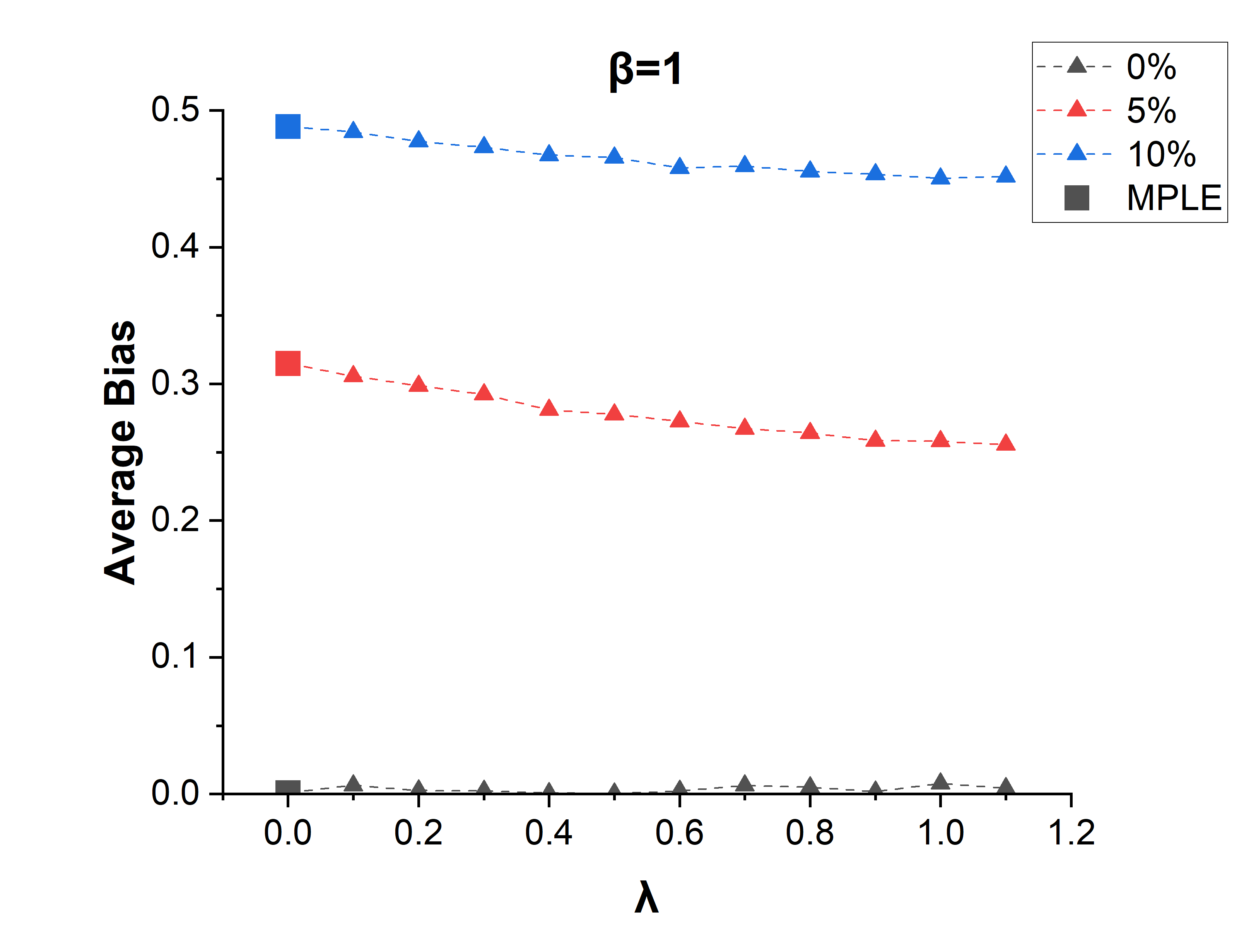}
        \caption[]%
            {{\small Bias}}    
    \end{subfigure}
    \caption{MSE and bias for estimates from the Sherrington Kirkpatrick model with $\beta = 1$, $N=2000$.}
    \label{SKM}
\end{figure}

\subsection{Slightly dense Erd\H{o}s-R\'enyi and stochastic block models}\label{sec:tradeoffn8}

We now consider slightly dense Erd\H{o}s-R\'enyi random graphs and stochastic block models 
which satisfy the conditions of Section \ref{sec:cltasp8}, 
and hence the correspondng MDPD estimators same asymptotic distribution for all $\lambda\geq 0$. 
The following simulation results shows that these MDPD estimators continue to improve the robustness 
under data contamination in spite of being equally efficient as the MPL estimator under the pure model. 
Particularly, Figure \ref{DERM242} and Figure \ref{DERM1} show the results 
for an Ising model on $G(N, N^{-3/4}\log N)$ and $G(N, N^{-1/2})$, respectively, 
while Figure \ref{SBM142} presents the results for a stochastic block model with both within and between group connection probabilities of the order $N^{-1/2}$. 

For Erd\H{o}s-R\'enyi random graphs with $p(N)=\frac{\log N}{N^{3/4}}$, 
there is a clear decreasing trend in the MSE and bias of the MDPD estimators in Figure \ref{DERM242}. 
Additionally, as the scaling of $p(N)$ increases to $\frac{1}{\sqrt{N}}$ in Figure \ref{DERM1}, the trend is still visible. For Ising model on stochastic block model in Figure \ref{SBM142} as well, 
there is a noticeable decreasing trend in the MSE and bias of the MDPD estimators 
under both $5\%$ and $10\%$ contamination.

\begin{figure}[h]
    \centering
    \begin{subfigure}[b]{0.49\textwidth}
        \centering
        \includegraphics[width=8cm]{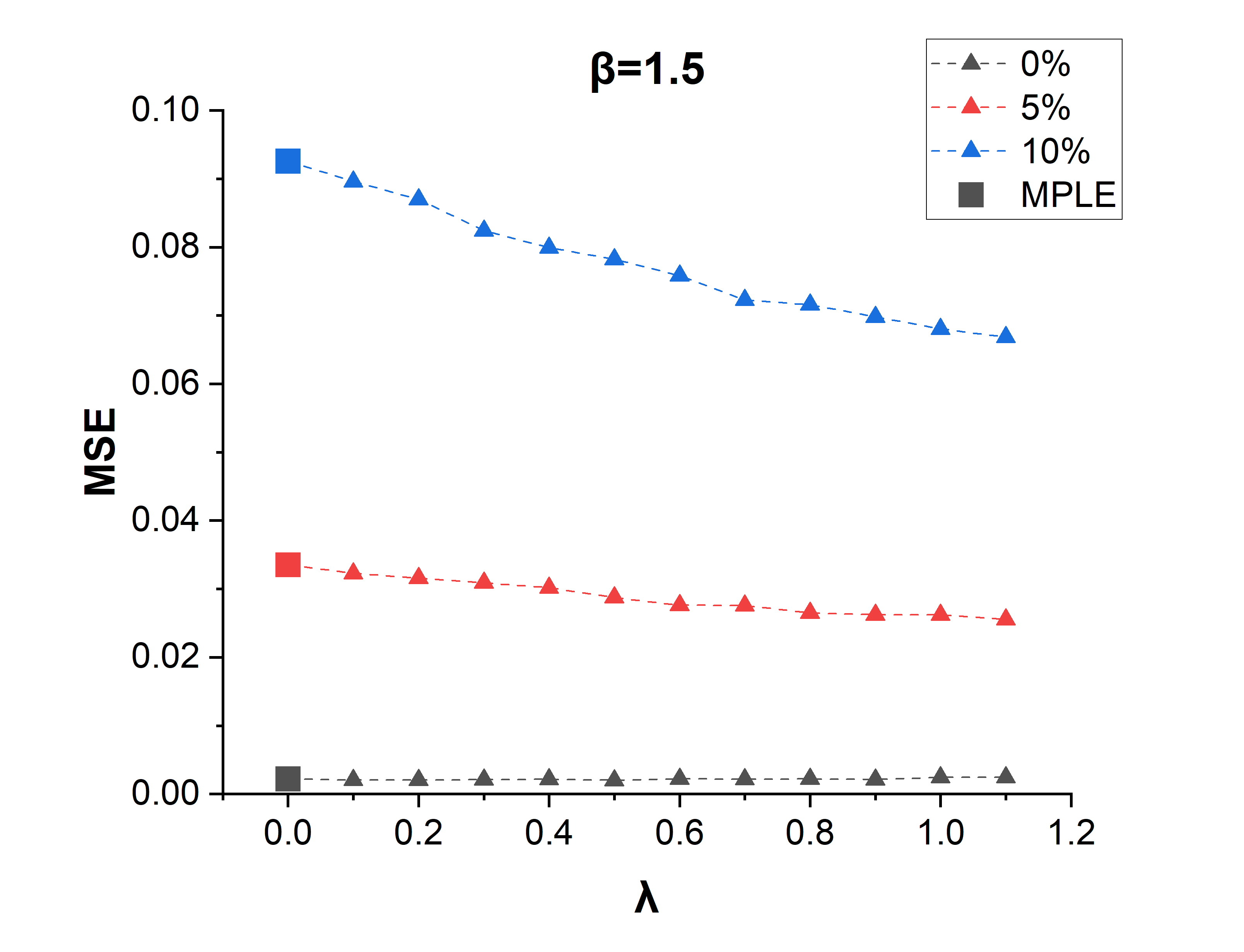}
        \caption[Network2]%
            {{\small MSE}}    
    \end{subfigure}
    \hfill
    \begin{subfigure}[b]{0.49\textwidth}  
        \centering 
        \includegraphics[width=8cm]{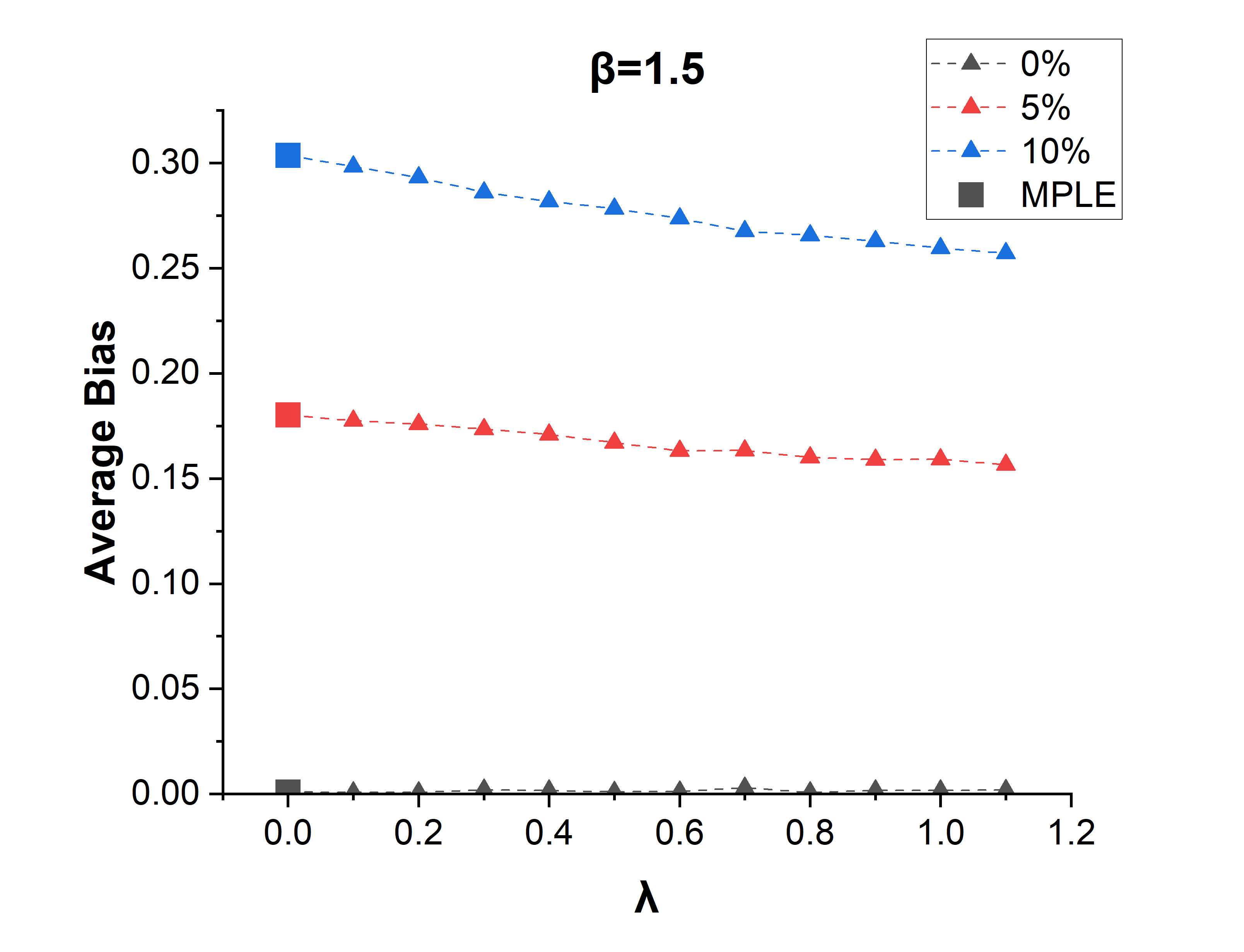}
        \caption[]%
            {{\small Bias}}    
    \end{subfigure}
    \caption{MSE and bias for estimates from Erdős–Rényi graph with $p=\frac{\log N}{N^{3/4}}$,  $\beta = 1.5$ and $N=2000$.}
    \label{DERM242}
\end{figure}

\begin{figure}[h]
    \centering
    \begin{subfigure}[b]{0.49\textwidth}
        \centering
        \includegraphics[width=8cm]{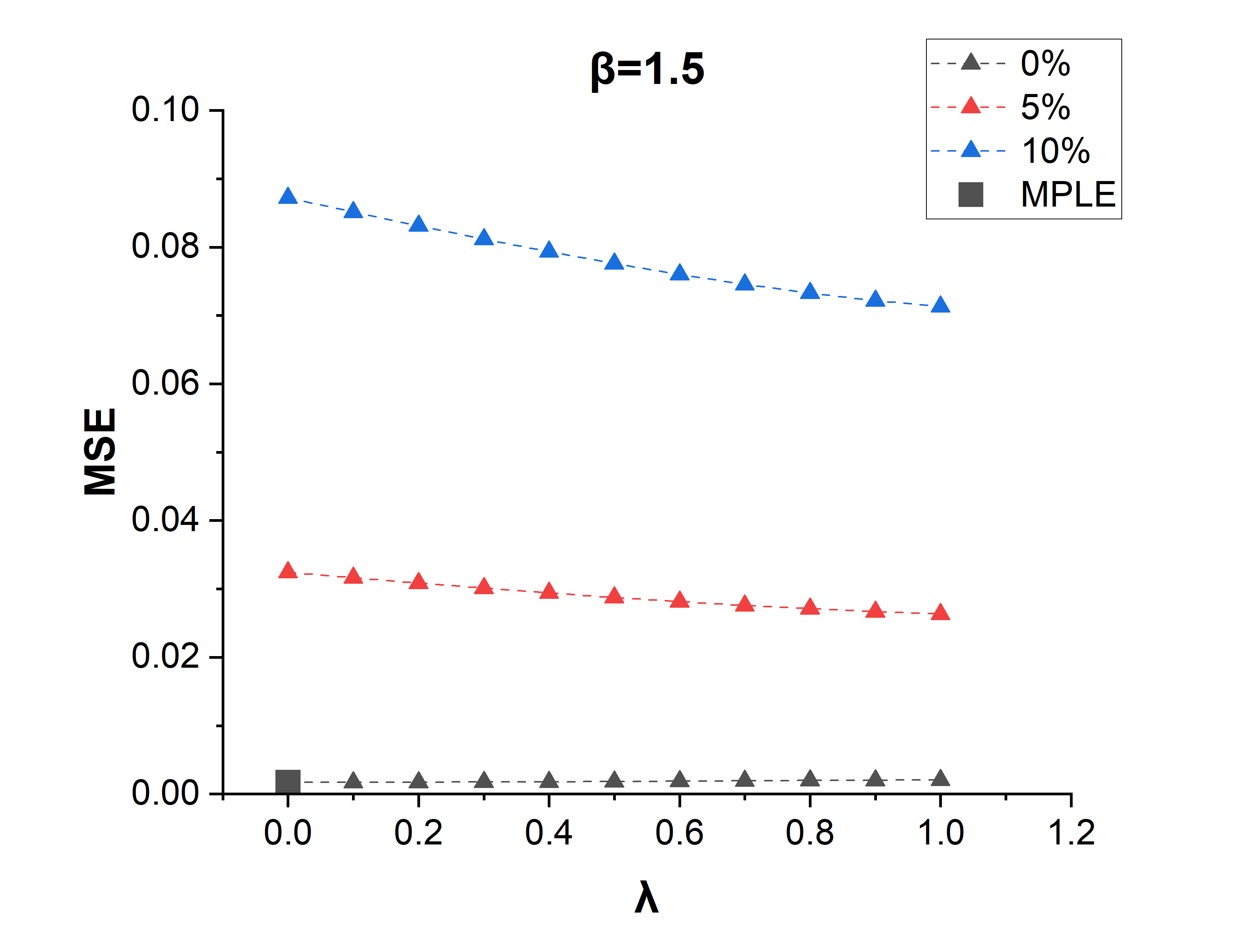}
        \caption[Network2]%
            {{\small MSE}}    
    \end{subfigure}
    \hfill
    \begin{subfigure}[b]{0.49\textwidth}  
        \centering 
        \includegraphics[width=8cm]{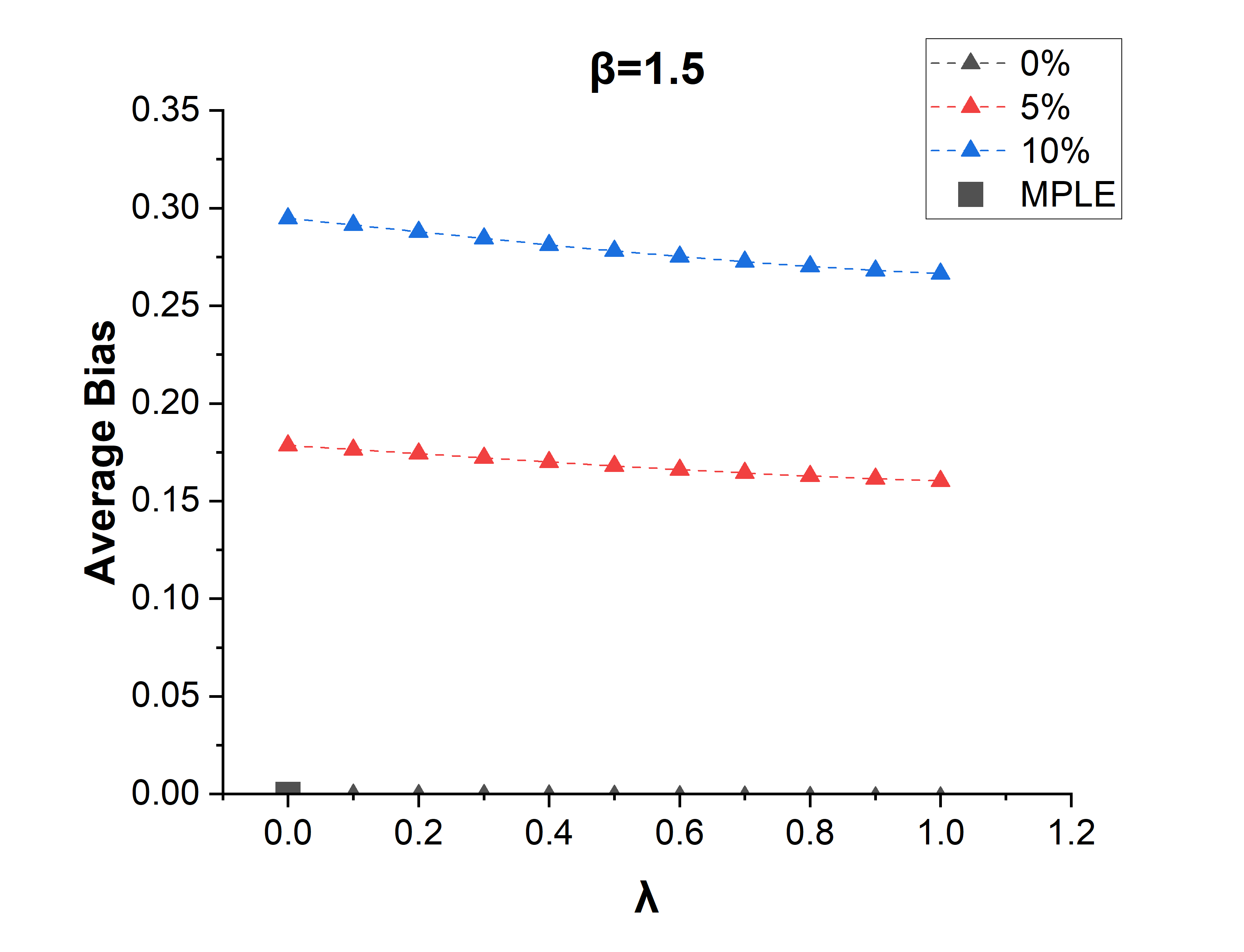}
        \caption[]%
            {{\small Bias}}    
    \end{subfigure}
    \caption{MSE and bias for estimates from Erd\H{o}s–R\'enyi graph with $p=\frac{1}{\sqrt{N}}$, $\beta = 1.5$ and $N=2000$.}
    \label{DERM1}
\end{figure}

\begin{figure}[h]
    \centering
    \begin{subfigure}[b]{0.49\textwidth}
        \centering
        \includegraphics[width=8cm]{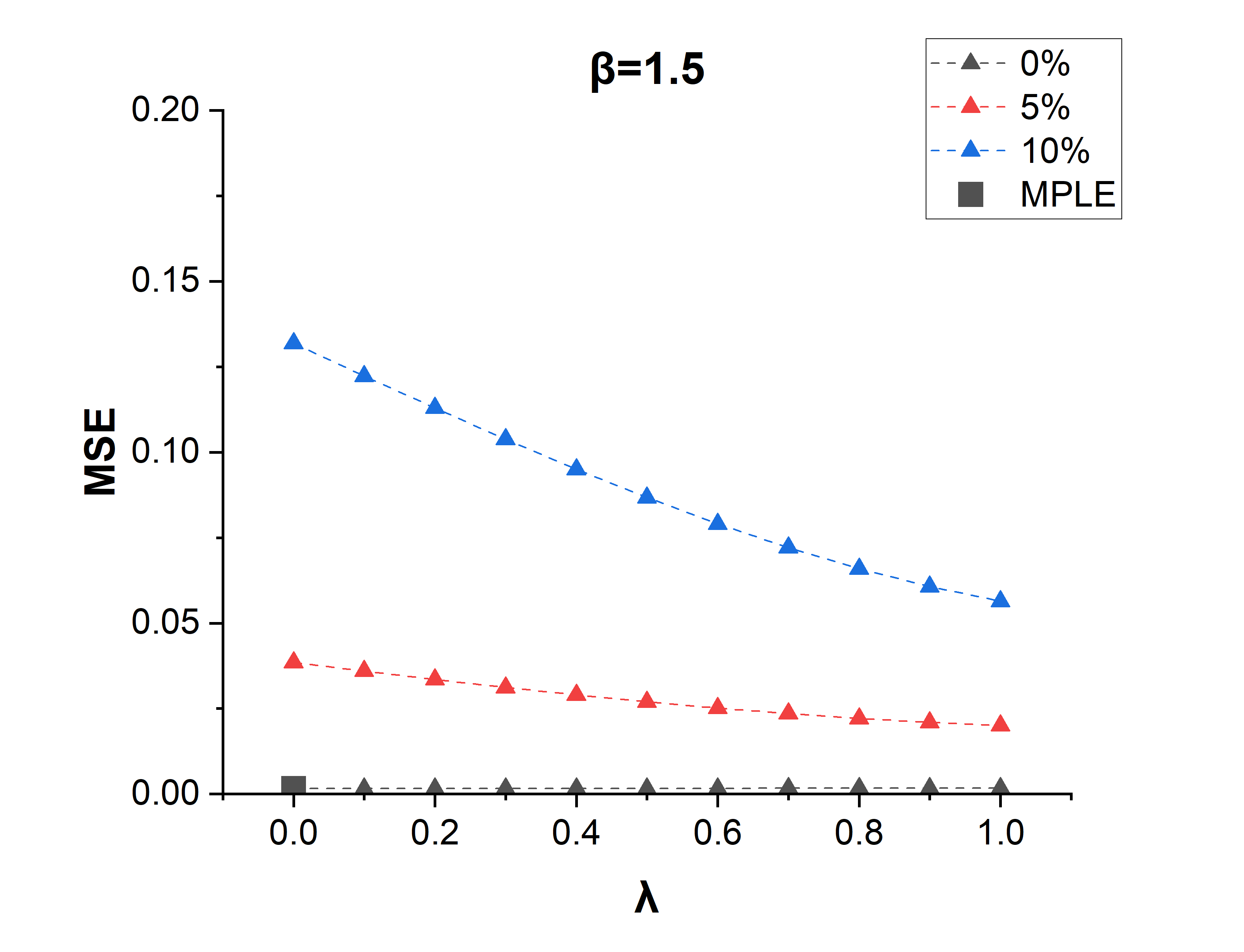}
        \caption[Network2]%
            {{\small MSE}}    
    \end{subfigure}
    \hfill
    \begin{subfigure}[b]{0.49\textwidth}  
        \centering 
        \includegraphics[width=8cm]{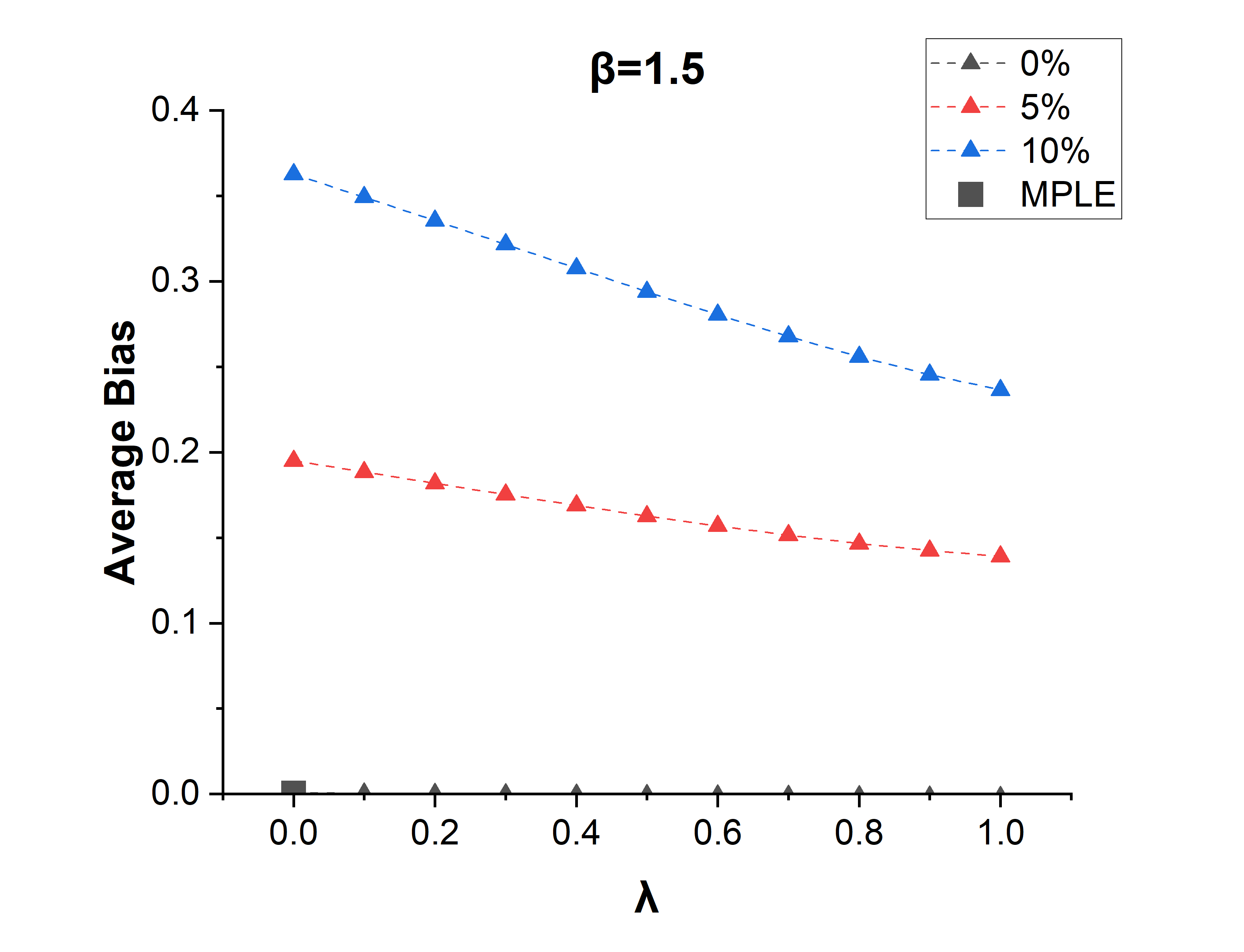}
        \caption[]%
            {{\small Bias}}    
    \end{subfigure}
    \caption{MSE and bias for estimates from the stochastic block model, with $p=\frac{4}{\sqrt{N}}$, $q=\frac{1}{\sqrt{N}}$, $\beta = 1.5$ and $N=2000$.}
    \label{SBM142}
\end{figure}

\section{Real Data Analysis}\label{sec:real data}
In this section, we demonstrate the applicability of the MDPD estimators on some representative datasets from the domains of social networks, neurobiology, and genomics by examining the average leave-one-out 
prediction accuracy obtained based on the MDPD estimators at various $\lambda\geq 0$.
For all these real data examples, the actual (mean) MDPD estimates are also summarized together in Table \ref{tab:beta}.

\subsection{Application in Social Networks Datasets}

In this section, two social network datasets from the Stanford Large Network Dataset Collection (SNAP) \cite{snapnets} are analyzed. The first dataset is a social network of GitHub developers collected from the public API in June 2019 \cite{github}. The vertices of the network are developers who have starred at least 10 repositories, and the edges represent mutual follower relationships between them. The network consists of 37,700 nodes and 289,003 edges. The vertices are assigned binary values depending on whether the GitHub user is a web or machine learning developer, the assignment being derived from the job title of each user.

Another investigated data set is a social network of Deezer users collected from the public API in March 2020 \cite{feather}. The vertices of the network are Deezer users from European countries, and the edges represent mutual follower relationships between them. The network contains 28,281 nodes and 92,752 edges. The vertices represent genders of the users, derived from the name field for each user.

Both datasets are subject to data contamination, particularly label flipping, due to independent failures in correctly labeling the nodal (binary) values. For example, since many job titles may not reflect whether the user is a web or a machine learning developer, the vertices may be flipped with unknown probability in the GitHub network data. However, as Deezer requires its users to select a gender upon account creation, it is expected that the Deezer dataset has fewer independent labeling failures than the GitHub dataset. 

We compare the performance of MDPD estimators and the MPL estimator based on the following \textit{leave-one-out} algorithm. For each vertex $i$, let $\bm J_{-i}$ denote the submatrix of $\bm J$ obtained by deleting the $i^\mathrm{th}$ row and $i^\mathrm{th}$ column of $\bm J$. We then fit an Ising model with the interaction matrix $\bm J_{-i}$ and observation $\bm X_{-i}:=\{X_1,...,X_{i-1},X_{i+1},...X_{N}\}$, and obtain the MDPD estimates $\hat{\beta}_\lambda$ for each $\lambda \in \{0,0.1,\ldots,1\}$.
This is followed by predicting the spin at vertex $i$ using the estimated conditional probability $\p_{\hat{\beta}_\lambda}(X_i|\bm X_{-i})$, repeated across all the vertices $i$, which are plotted in \autoref{fig:network}.

\begin{figure}[!h]
    \centering
    \begin{subfigure}[b]{0.49\textwidth}
        \centering
        \includegraphics[width=1\linewidth]{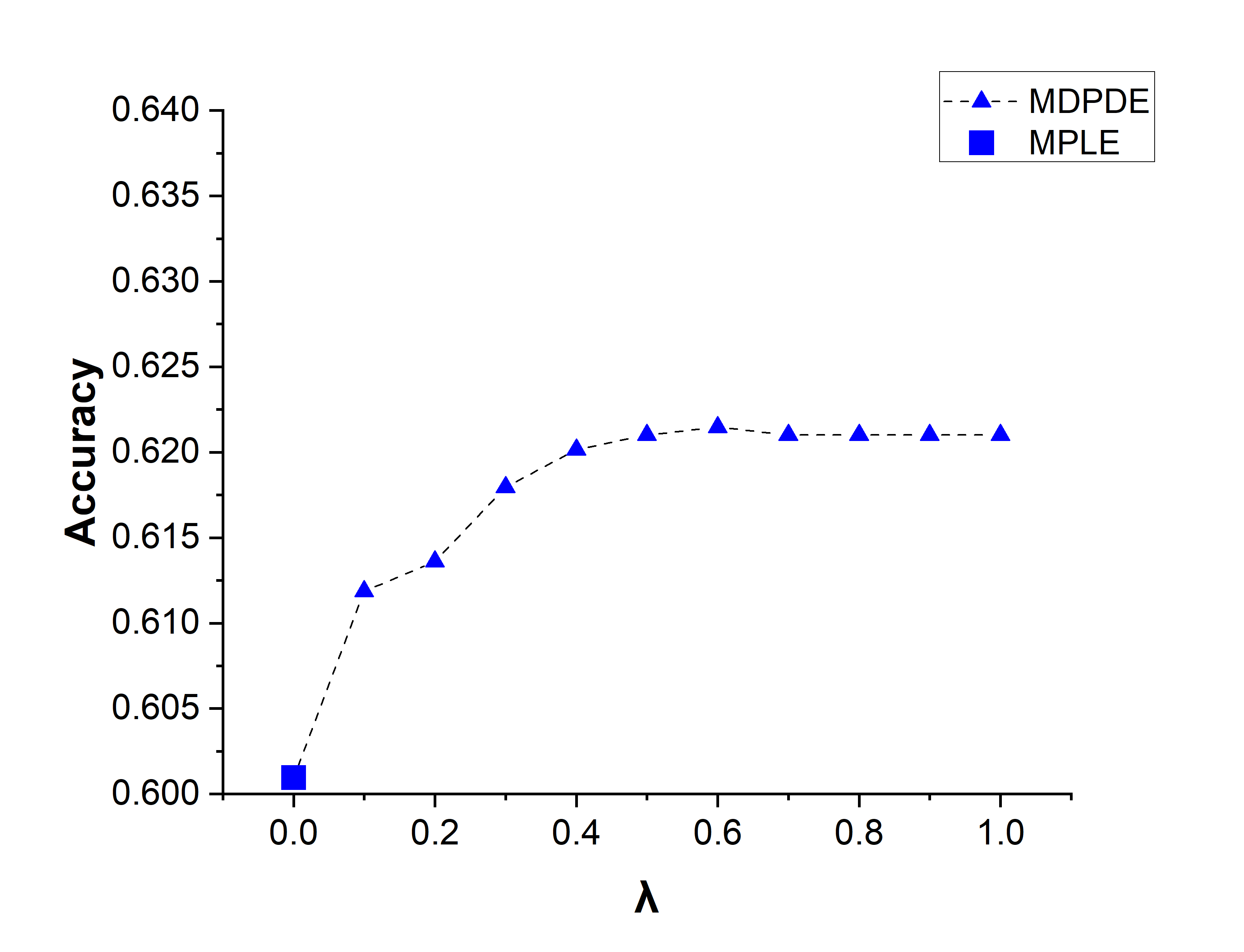}
        \caption{GitHub}
    \end{subfigure}
    \hfill
    \begin{subfigure}[b]{0.49\textwidth}  
        \centering
        \includegraphics[width=1\linewidth]{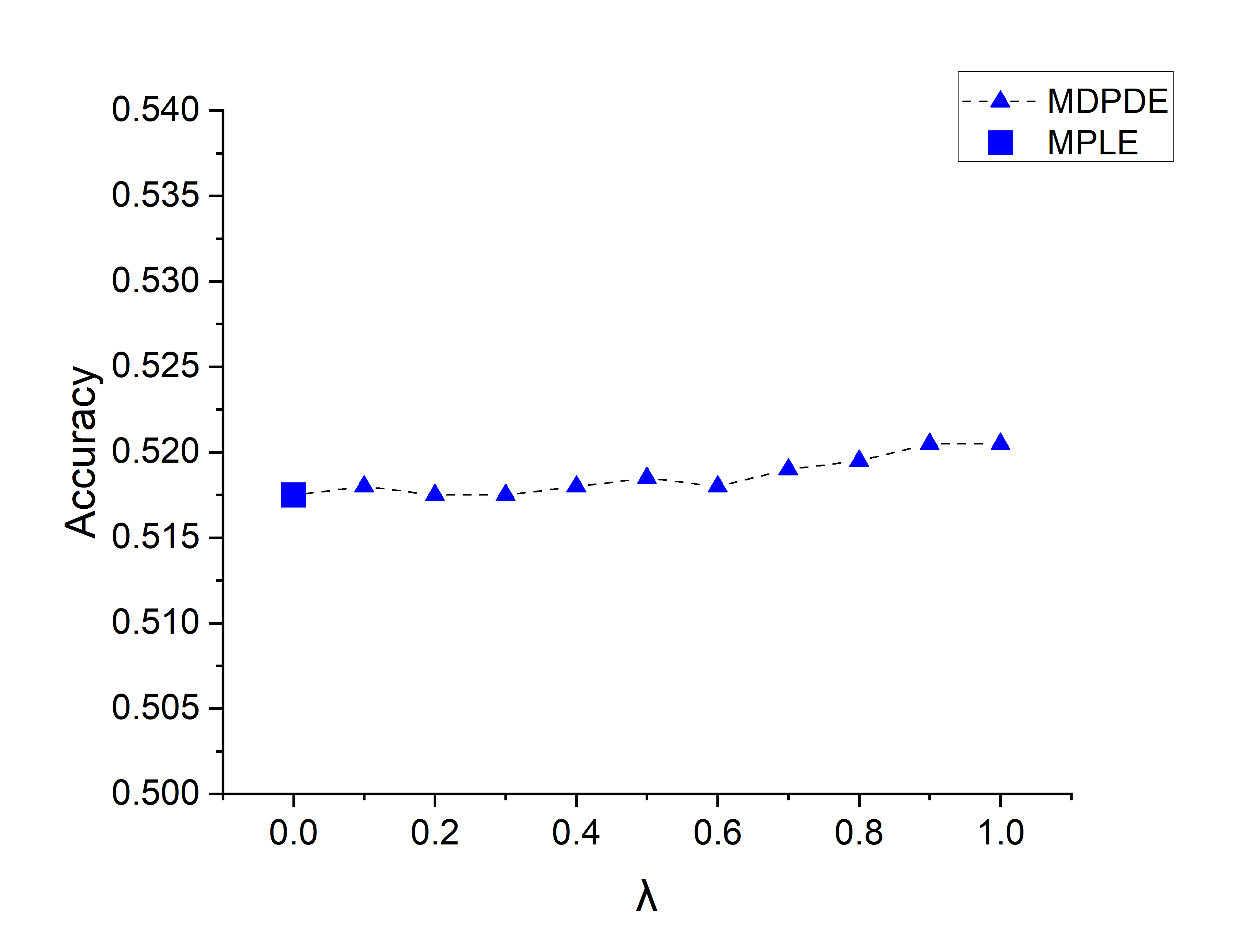}
        \caption{Deezer}
    \end{subfigure}
    \caption{Average prediction accuracy obtained from the social network datasets.}
    \label{fig:network}
\end{figure}

Although there is no significant change in the estimates of $\beta$ for the datasets, \autoref{fig:network} shows significant improvement in the prediction accuracy as $\lambda$ increases particularly for the GitHub dataset. 
In this dataset, the MDPD estimates $\hat{\beta}_{\lambda}$ with $\lambda\geq 0.5$ are performing significantly better in prediction than the estimates with smaller $\lambda$ including the MPL estimate (at $\lambda=0$). 
This, in turn, indicates the presence of contamination in the data, as expected,
and the advantages of the proposed MDPD estimates in such cases. 
For the Deezer dataset, on the other hand, the improvement in the prediction accuracy across any range of $\lambda$ is not very significant. This is expected, because the requirement of Deezer to input genders of new users while registration imposes a natural control on independent failures. However, in such cases as well, the proposed MDPD estimates at any $\lambda>0$ perform no worse, and sometimes a little better, than the MPL estimate. 

\subsection{Applications in Neurobiological Datasets}

We also evaluate the performance of the MDPD estimates on a Neurobiological data acquired through electrophysiological recordings from the Visual Coding Neuropixels dataset of the Allen Brain Observatory \cite{de2020large,allen}. The dataset,  introduced in \cite{rahulsom},  recorded spike trains of 555 neurons in a male mouse aged 116 days (session ID 791319847) via six neuropixel probes, when the mouse was presented with four stimulus categories: natural scenes, static gratings, gabor patches, and flashes. The spike trains are converted into binary data based on the status of the neurons and assigned $+1$ if the neurons are active. Each time the visual stimulus is presented, we record one sample from the spike trains, with 50 samples in total. Since we do not know the true anatomical interaction structure of the neurons, we apply the Peter-Clark (PC) algorithm \cite{glymour2019review, rahulsom}, a popular causal discovery method that infers the structure of a causal directed acyclic graph (DAG) from observational data. The PC algorithm is shown to be consistent based on i.i.d. samples in \cite{pc}, and for time-series data satisfying mild mixing assumptions in \cite{rahulsom}. The algorithm's output is a completed partially directed acyclic graph (CPDAG),
which is then transformed into an undirected graph and used as the underlying network in our analysis.

\begin{figure}[!h]
    \centering
    \begin{subfigure}[b]{0.49\textwidth}
        \centering
        \includegraphics[width=1\linewidth]{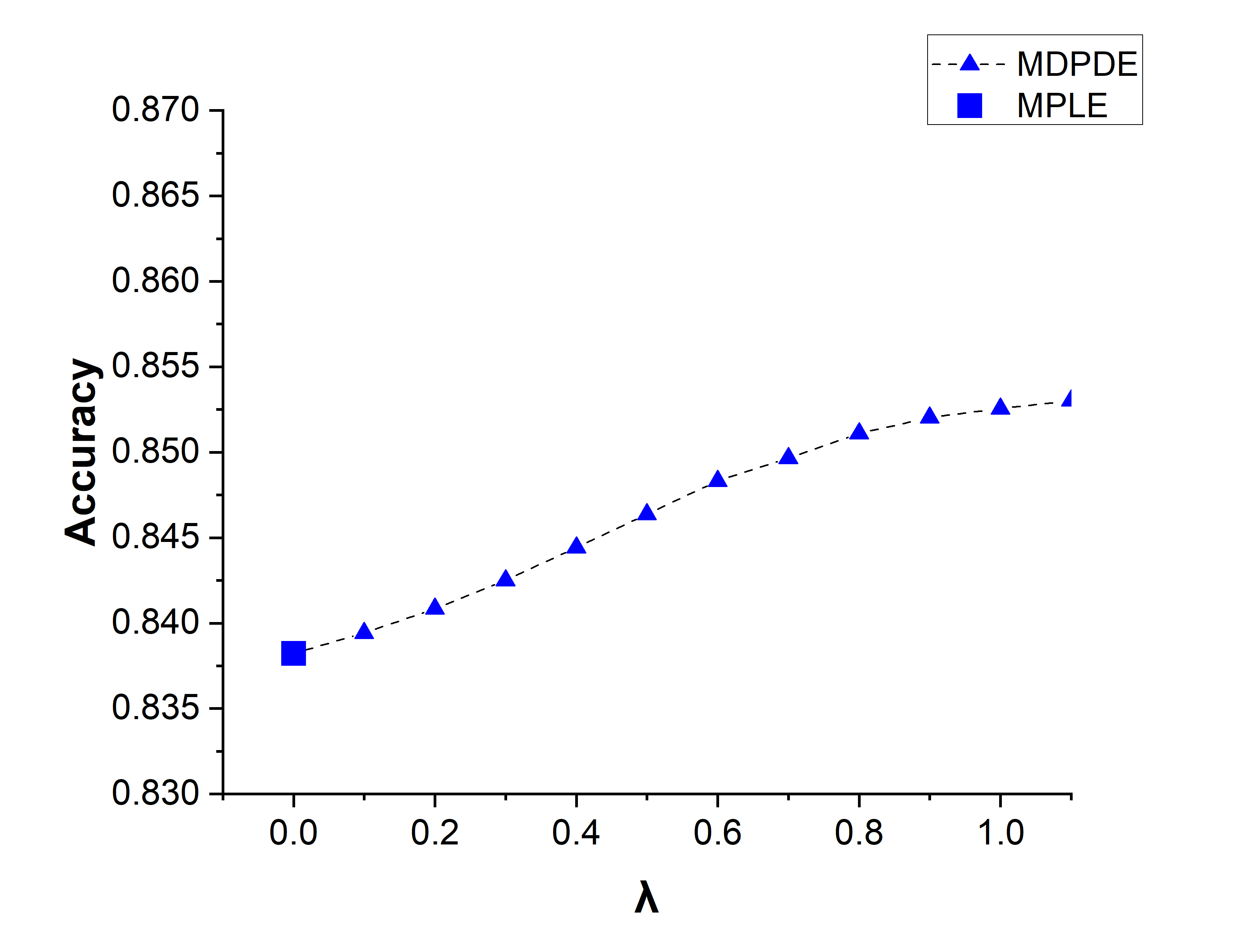}
        \caption{natural scenes}
    \end{subfigure}
    \hfill
    \begin{subfigure}[b]{0.49\textwidth}  
        \centering
        \includegraphics[width=1\linewidth]{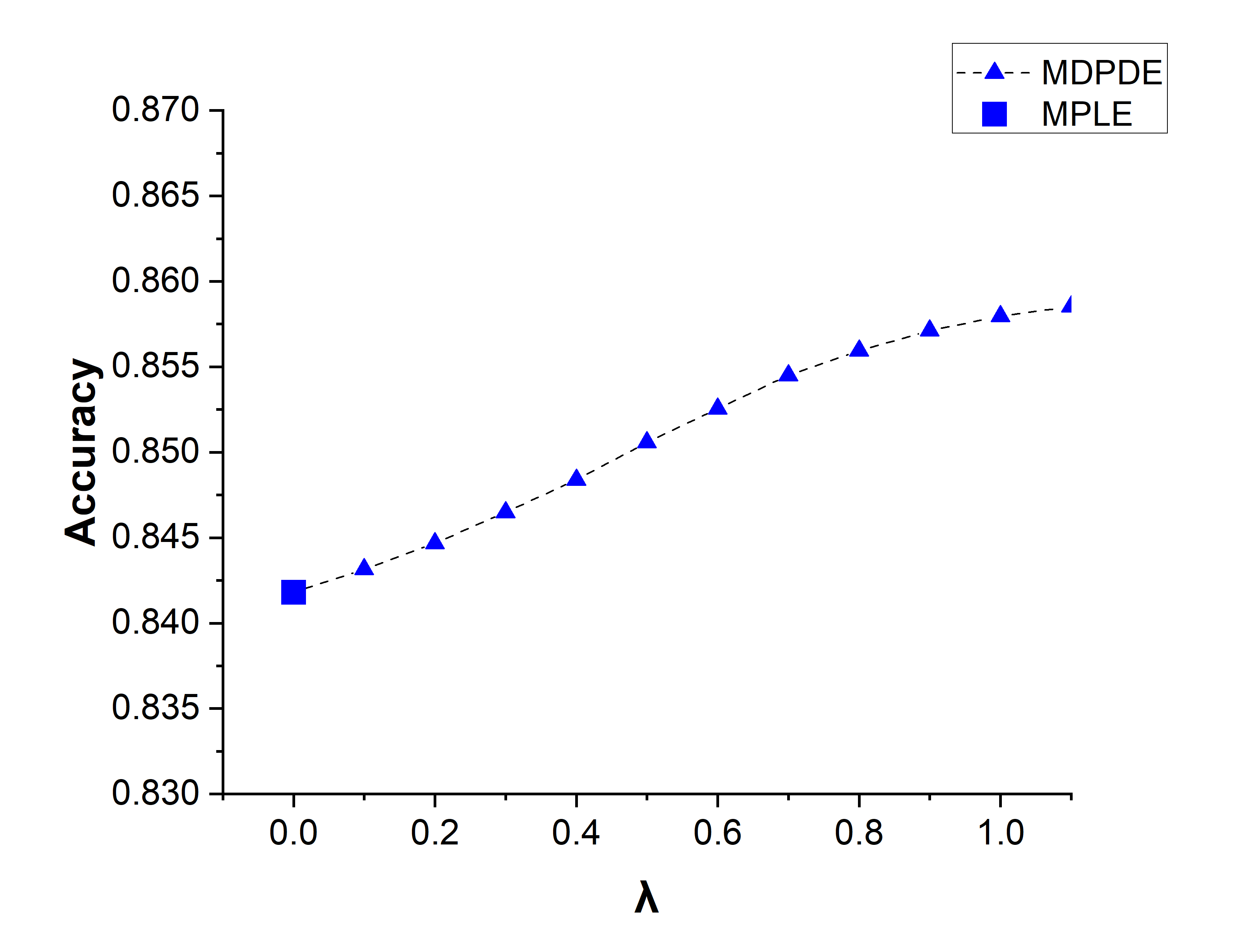}
        \caption{static gratings}
    \end{subfigure}
    \begin{subfigure}[b]{0.49\textwidth}  
        \centering
        \includegraphics[width=1\linewidth]{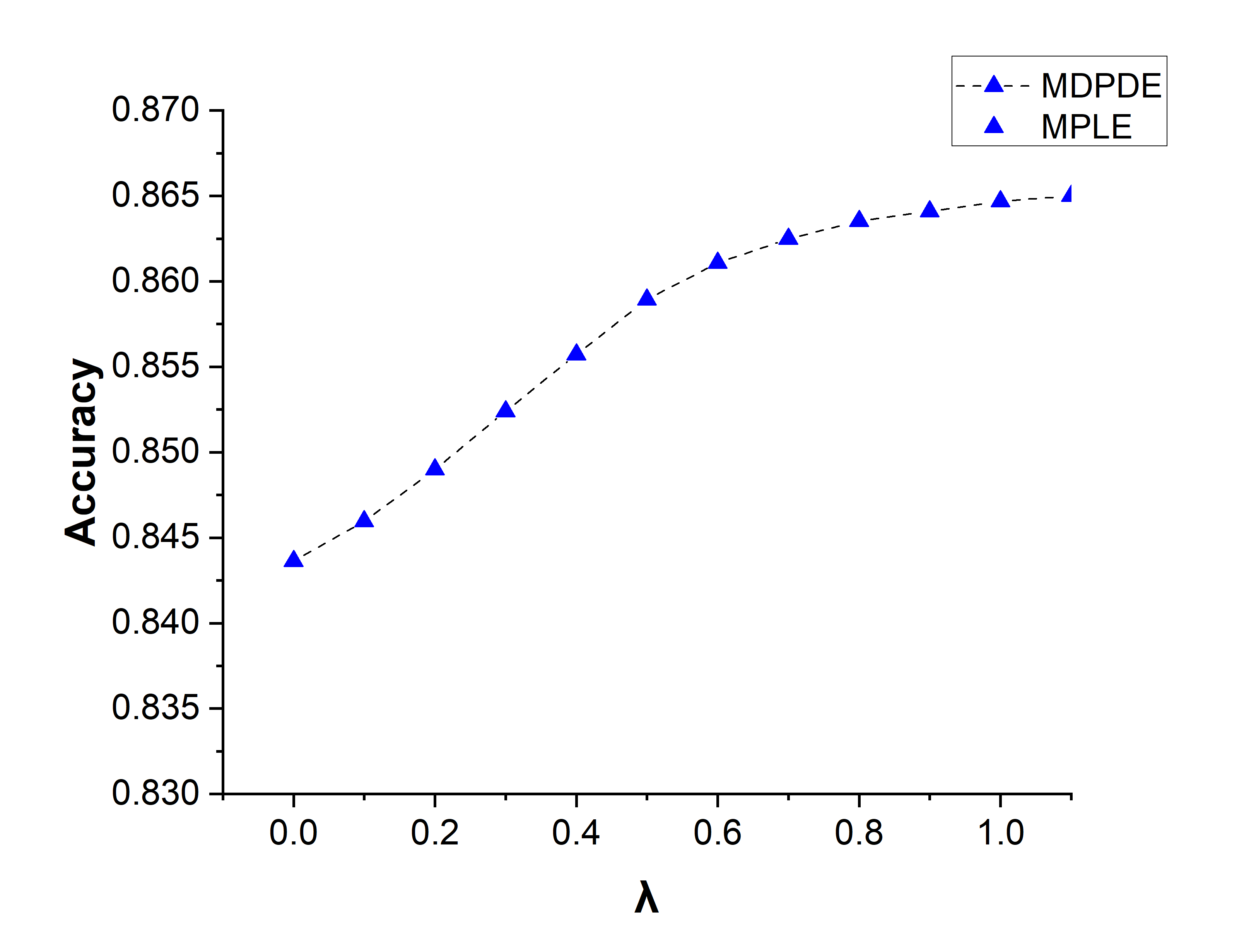}
        \caption{gabors}
    \end{subfigure}
    \hfill
    \begin{subfigure}[b]{0.49\textwidth}  
        \centering
        \includegraphics[width=1\linewidth]{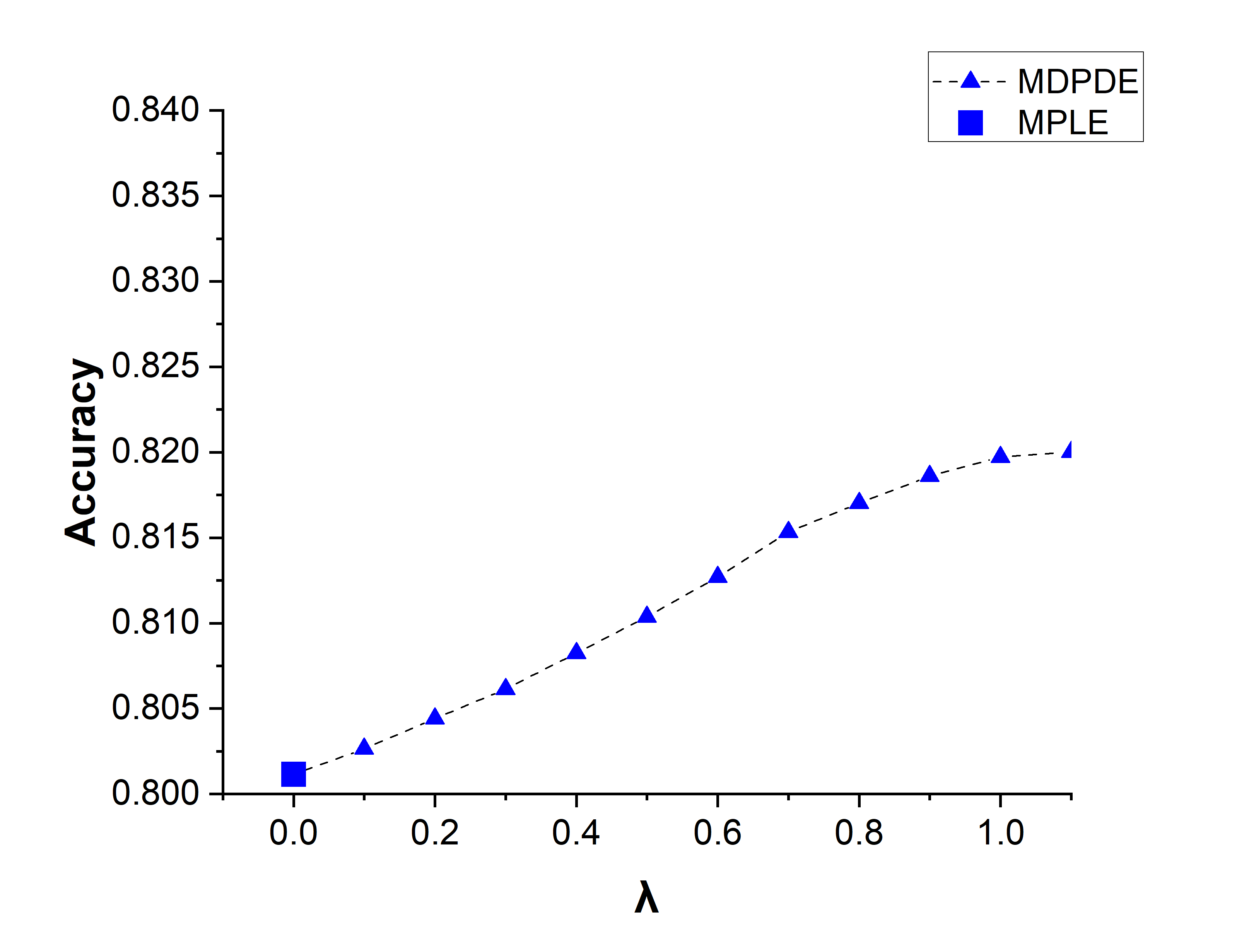}
        \caption{flashes}
    \end{subfigure}
    \caption{Average prediction accuracy obtained from the neuro-biological datasets.}
    \label{fig:neuro}
\end{figure}

Once again, we compare the performances of the MDPD estimators based on the node-wise prediction criterion followed in the previous example (social network datsets). However, in this dataset, there are multiple samples evenly splited into a training set and a test set. So, for each stimulus type, we fit an Ising model and obtain estimates $\hat{\beta}_{\lambda}$ for each sample in the training set, and use the mean of these estimates to predict the spin at each vertex $i$ via the conditional probabilities $\p_{{\hat{\beta}_\lambda}}(X_i|(X_{j})_{j\ne i})$ for samples in the test set. The resulting prediction accuracies are plotted across $\lambda$ in \autoref{fig:neuro}. 
All four plots there demonstrate an increasing trend in prediction accuracy as $\lambda$ increases to $1$, thereby indicating the presence of contamination in the underlying data and showcasing the efficacy of the MDPD estimators in handling such a situation. 


\subsection{Applications in Genomic Datasets}

As our final illustration, we evaluate the performance of the MDPD estimates on gene expression data obtained from the Curated Microarray Database (CuMiDa) (\url{https://sbcb.inf.ufrgs.br/cumida}) (\cite{sbcb}). This is a database consisting of 78 handpicked micro-array data sets for Homo sapiens that were carefully examined from more than 30,000 micro-array experiments from Gene Expression Omnibus using a rigorous filtering criterion. All data sets were individually submitted to background correction, normalization, sample quality analysis, and were manually edited to eliminate erroneous probes. Thus, the data can be regarded as more or less uncontaminated.

These data from the microarray experiment consist of a two-dimensional matrix with genes as rows and samples as columns (originating from different conditions). Each cell in the matrix is a normalized real number indicating how much a gene is expressed in a sample. These expression matrices will usually have thousands of rows and dozens or hundreds of columns. We selected two particular datasets, the $Liver\_GSE14520\_U133A$ dataset, which contains 181 samples from liver cancer patients, and the $Breast\_GSE70947$ dataset, which contains 66 samples from breast cancer patients. Since we do not know the true gene interaction structure, we apply the Peter-Clark (PC) algorithm \cite{glymour2019review, rahulsom} to estimate the underlying network. The values for each node are standardized to binary values $\{-1,+1\}$ by mapping original gene expression values above the mean to $+1$ and values below the mean to $-1$. 


We also follow the same approach as in the neurobiological dataset to compute and report the prediction accuracy in Figure \ref{fig:cumida}. Both these gene expression datasets show similar results for all MDPD estimators with $\lambda\in[0,1]$. This is expected because the datasets have gone through careful pre-processing as mentioned above, which must have minimized the contamination. 

\begin{figure}[!ht]
    \centering
    \begin{subfigure}[b]{0.49\textwidth}
        \centering
        \includegraphics[width=1\linewidth]{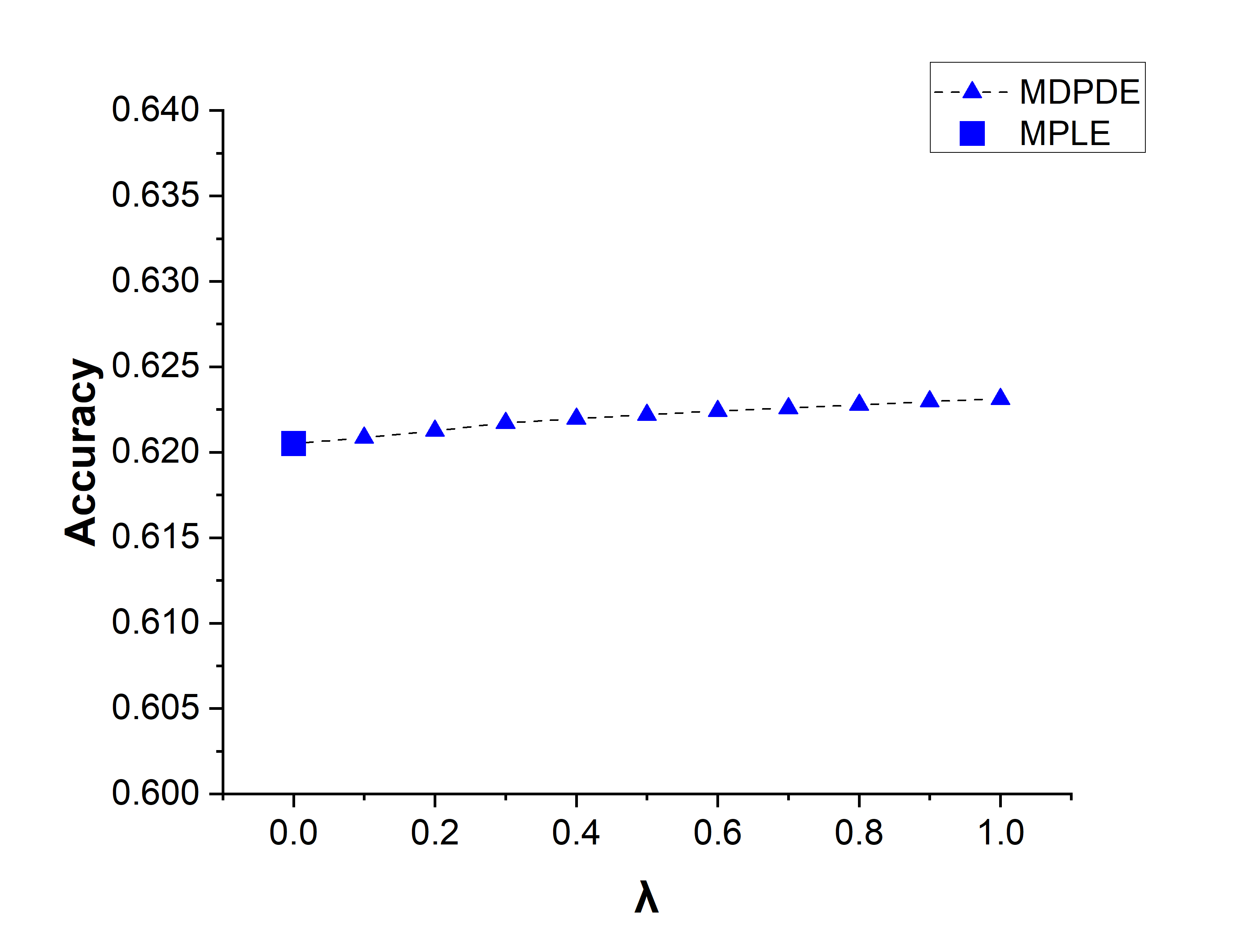}
        \caption{Liver\_GSE14520\_U133A}
    \end{subfigure}
    \hfill
    \begin{subfigure}[b]{0.49\textwidth}  
        \centering
        \includegraphics[width=1\linewidth]{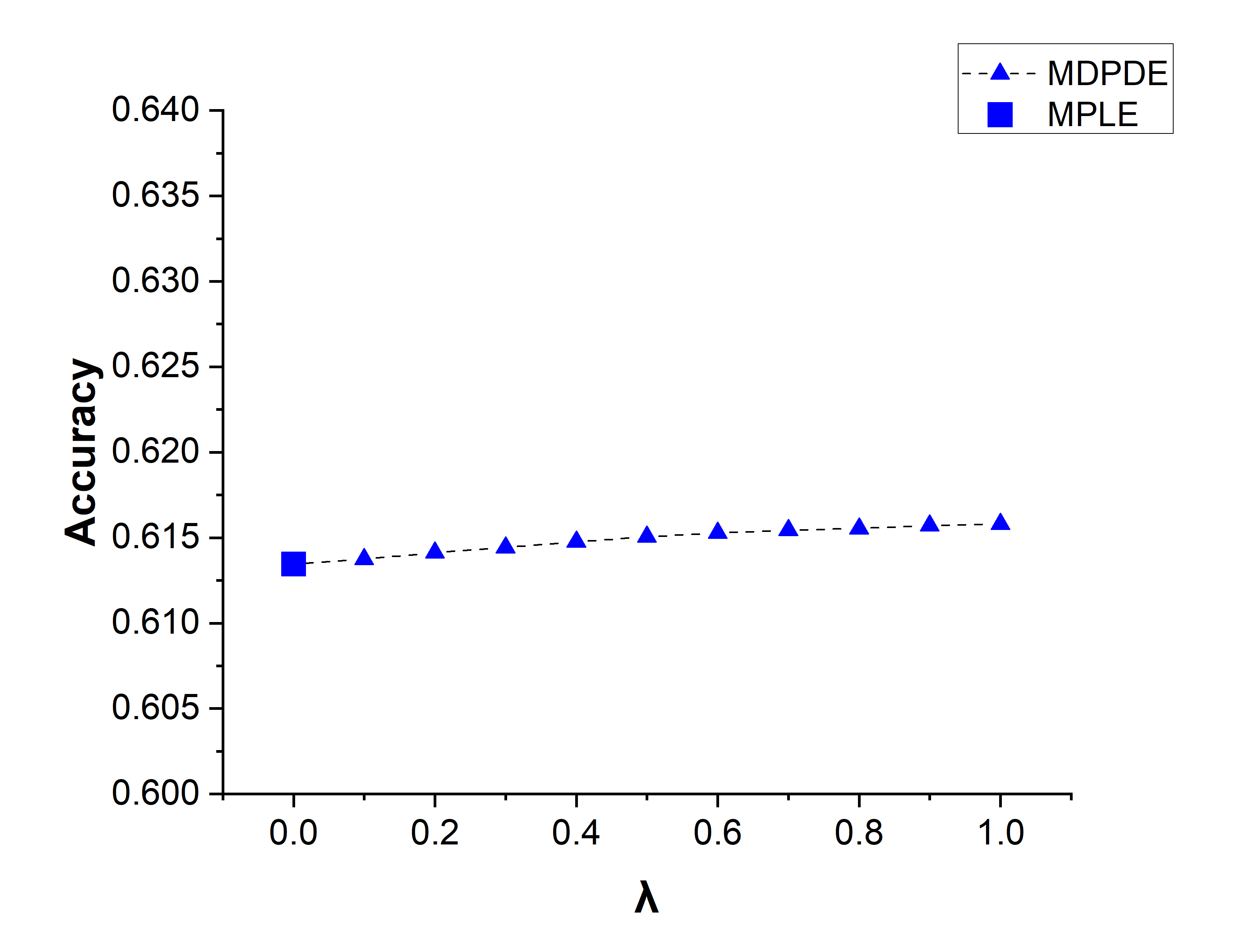}
        \caption{ Breast\_GSE70947}
    \end{subfigure}
    \caption{Average prediction accuracy obtained from the CuMiDa datasets}
    \label{fig:cumida}
\end{figure}

\begin{table}[!ht]
    \tiny
    \centering
    \begin{tabular}{l|l|l|l|l|l|l|l|l}
    \hline
        $\lambda$ & Github & Deezer & natural & static & gabors & flashes  & Breast & Liver \\ 
         & & & scene & gratings & gabors & flashes  & GSE70947 & GSE14520\_U133A \\ \hline
        0 & 0.4051  & 0.0326  & 0.4958 & 0.4955 & 0.6016 & 0.5245  & 0.3359  & 0.2929  \\ \hline
        0.1 & 0.4731  & 0.0334  & 0.5034 & 0.5045 & 0.6264 & 0.5361  & 0.3373  & 0.2940  \\ 
        0.2 & 0.5170  & 0.0341  & 0.5120 & 0.5148 & 0.6597 & 0.5499  & 0.3387  & 0.2950  \\ 
        0.3 & 0.5465  & 0.0347  & 0.5216 & 0.5267 & 0.7036 & 0.5660  & 0.3400  & 0.2961  \\ 
        0.4 & 0.5673  & 0.0353  & 0.5318 & 0.5403 & 0.7582 & 0.5850  & 0.3412  & 0.2970  \\ 
        0.5 & 0.5824  & 0.0358  & 0.5424 & 0.5555 & 0.8160 & 0.6086  & 0.3422  & 0.2979  \\         
        0.6 & 0.5935  & 0.0363  & 0.5529 & 0.5719 & 0.8655 & 0.6387  & 0.3432  & 0.2987  \\ 
        0.7 & 0.6017  & 0.0367  & 0.5625 & 0.5876 & 0.9057 & 0.6667  & 0.3440  & 0.2995  \\ 
        0.8 & 0.6077  & 0.0370  & 0.5709 & 0.6010 & 0.9364 & 0.6935  & 0.3447  & 0.3001  \\ 
        0.9 & 0.6121  & 0.0373  & 0.5777 & 0.6117 & 0.9585 & 0.7139  & 0.3453  & 0.3006  \\ \hline
    \end{tabular}
    \caption{The MDPDEs $\hat{\beta}_{\lambda}$ and the MPLE $\hat{\beta}_{MPL}$ (at $\lambda=0$) 
    for all real datasets}
    \label{tab:beta}
\end{table}

\section{Discussions}

Real-world network data often contain contamination, but classical interaction recovery and structure learning algorithms typically fail to account for such issues, leading to poor performance when applied to contaminated network data.
In this paper, we introduce a density power divergence (DPD) based variant of the computationally efficient maximum pseudolikelihood (MPL) estimator for the interaction strength parameter ($\beta$) in the Ising model, a fundamental framework for dependent binary data on networks. The proposed DPD estimators retain consistency under the pure model, while offering substantially greater robustness under data contamination compared to their MPL counterpart. This robustness is demonstrated by a reduction in mean-squared error in synthetic simulations and by improved prediction accuracy in real data applications. 
The proposed estimators are also shown to be asymptotically fully efficient under the pure data
for a wide range of network models; this has been achieved by proving a first CLT 
for a general class of $Z$-estimators covering these MDPD estimators.

Several interesting directions emerge from this work. A natural extension is to study more general network models such as the Potts model \cite{wu1982potts, gandolfo2010limit, eichelsbacher2015rates, bhowal2025limit}, where each observation can take a finite number of values rather than being restricted to binary states. Another promising avenue is the exploration of tensor versions of the classical Ising and Potts models, which move beyond pairwise interactions to incorporate higher-order dependencies that capture peer-group effects in hyper-network structures \cite{psm1, psm2, psm3, can2023meanfield, mukherjee2024efficient}.
A further direction of potential interest involves network-dependent logistic regression models \cite{cd_ising_I, cd_ising_II, mukherjee_jmlr}, which can be formulated as Ising models with site-specific magnetic fields. Developing robust estimation methods based on DPD estimators for such models would be an important step forward. 


Finally, as noted earlier, the structure learning problem, which seeks to recover the full interaction structure from multiple samples, is also of considerable interest, especially in high-dimensional regimes. In this setting, an important yet challenging direction is to develop penalized versions of robust estimators such as the proposed MDPD estimators, and establish their consistency under the pure model. We hope to pursue some of these important extensions in our future research. 

\bibliographystyle{plain} 
\bibliography{references}

\begin{thebibliography}{10}

\bibitem{spatial}
Sudipto Banerjee, Bradley~P Carlin, and Alan~E Gelfand.
\newblock {\em Hierarchical modeling and analysis for spatial data}.
\newblock Chapman and Hall/CRC, 2003.

\bibitem{basu2017wald}
Ayanendranath Basu, Abhik Ghosh, Abhijit Mandal, Nirian Mart{\i}n, and Leandro
  Pardo.
\newblock A wald-type test statistic for testing linear hypothesis in logistic
  regression models based on minimum density power divergence estimator.
\newblock {\em Electronic Journal of Statistics}, 11:2741--2772, 2017.

\bibitem{basu1998robust}
Ayanendranath Basu, Ian~R Harris, Nils~L Hjort, and MC~Jones.
\newblock Robust and efficient estimation by minimising a density power
  divergence.
\newblock {\em Biometrika}, 85(3):549--559, 1998.

\bibitem{basu2011statistical}
Ayanendranath Basu, Hiroyuki Shioya, and Chanseok Park.
\newblock {\em Statistical inference: the minimum distance approach}.
\newblock CRC press, 2011.

\bibitem{besag_lattice}
Julian Besag.
\newblock Spatial interaction and the statistical analysis of lattice systems.
\newblock {\em Journal of the Royal Statistical Society: Series B
  (Methodological)}, 36(2):192--225, 1974.

\bibitem{besag_nl}
Julian Besag.
\newblock Statistical analysis of non-lattice data.
\newblock {\em Journal of the Royal Statistical Society Series D: The
  Statistician}, 24(3):179--195, 1975.

\bibitem{BM16}
Bhaswar~B Bhattacharya and Sumit Mukherjee.
\newblock Inference in ising models.
\newblock {\em Bernoulli}, 24(1):493--525, 2018.

\bibitem{bhowal2025limit}
Sanchayan Bhowal and Somabha Mukherjee.
\newblock Limit theorems and phase transitions in the tensor curie--weiss potts
  model.
\newblock {\em Information and Inference: A Journal of the IMA}, 14(2):iaaf014,
  2025.
\newblock Open Access.

\bibitem{rahulsom}
Rahul Biswas and Somabha Mukherjee.
\newblock Consistent causal inference from time series with pc algorithm and
  its time-aware extension.
\newblock {\em Statistics and Computing}, 34(1):14, 2024.

\bibitem{Borgs2008}
Christian Borgs, Jennifer~T Chayes, L{\'a}szl{\'o} Lov{\'a}sz, Vera~T S{\'o}s,
  and Katalin Vesztergombi.
\newblock Convergent sequences of dense graphs i: Subgraph frequencies, metric
  properties and testing.
\newblock {\em Advances in Mathematics}, 219(6):1801--1851, 2008.

\bibitem{Borgs2012}
Christian Borgs, Jennifer~T Chayes, L{\'a}szl{\'o} Lov{\'a}sz, Vera~T S{\'o}s,
  and Katalin Vesztergombi.
\newblock Convergent sequences of dense graphs ii. multiway cuts and
  statistical physics.
\newblock {\em Annals of Mathematics}, pages 151--219, 2012.

\bibitem{bresler}
Guy Bresler.
\newblock Efficiently learning ising models on arbitrary graphs.
\newblock In {\em Proceedings of the forty-seventh annual ACM symposium on
  Theory of computing}, pages 771--782, 2015.

\bibitem{bms}
Guy Bresler, Elchanan Mossel, and Allan Sly.
\newblock Reconstruction of markov random fields from samples: some
  observations and algorithms.
\newblock {\em SIAM Journal on Computing}, 42(2):563--578, 2013.

\bibitem{can2023meanfield}
Van~Hao Can and Adrian R{\"o}llin.
\newblock Mean-field spin models -- fluctuation of the magnetization and
  maximum likelihood estimator.
\newblock {\em arXiv preprint arXiv:2312.07313}, 2023.
\newblock 44 pages. Add new examples. Subjects: Probability (math.PR);
  Mathematical Physics (math-ph).

\bibitem{chatterjee}
Sourav Chatterjee.
\newblock Estimation in spin glasses: A first step1.
\newblock {\em The Annals of Statistics}, 35(5):1931--1946, 2007.

\bibitem{chen2011learning}
Yuxin Chen.
\newblock Learning sparse ising models with missing data.
\newblock Project report, Stanford University, Stat 375, 2011.
\newblock
  \url{https://web.stanford.edu/~montanar/TEACHING/Stat375/proj/chen.pdf}.

\bibitem{ising_dna_correlations}
A~Colliva, R~Pellegrini, Alessandro Testori, and Michele Caselle.
\newblock Ising-model description of long-range correlations in dna sequences.
\newblock {\em Physical Review E}, 91(5):052703, 2015.

\bibitem{comets}
Francis Comets and Basilis Gidas.
\newblock Asymptotics of maximum likelihood estimators for the curie-weiss
  model.
\newblock {\em The Annals of Statistics}, pages 557--578, 1991.

\bibitem{cox1974theoretical}
DR~Cox and DV~Hinkley.
\newblock Theoretical statistics chapman and hall, london.
\newblock {\em See Also}, 1974.

\bibitem{cd_ising_I}
Constantinos Daskalakis, Nishanth Dikkala, and Ioannis Panageas.
\newblock Regression from dependent observations.
\newblock In {\em Proceedings of the 51st Annual ACM SIGACT Symposium on Theory
  of Computing}, pages 881--889, 2019.

\bibitem{cd_ising_II}
Constantinos Daskalakis, Nishanth Dikkala, and Ioannis Panageas.
\newblock Logistic regression with peer-group effects via inference in
  higher-order ising models.
\newblock In {\em International Conference on Artificial Intelligence and
  Statistics}, pages 3653--3663. PMLR, 2020.

\bibitem{de2020large}
Saskia~EJ de~Vries, Jerome~A Lecoq, Michael~A Buice, Peter~A Groblewski,
  Gabriel~K Ocker, Michael Oliver, David Feng, Nicholas Cain, Peter
  Ledochowitsch, Daniel Millman, et~al.
\newblock A large-scale standardized physiological survey reveals functional
  organization of the mouse visual cortex.
\newblock {\em Nature neuroscience}, 23(1):138--151, 2020.

\bibitem{pivotal_clt}
Nabarun Deb.
\newblock Pivotal clts for pseudolikelihood via conditional centering in
  dependent random fields.
\newblock arXiv preprint arXiv:2510.04972, 2025.
\newblock \url{https://arxiv.org/abs/2510.04972}.

\bibitem{meanfield}
Nabarun Deb and Sumit Mukherjee.
\newblock Fluctuations in mean-field ising models.
\newblock {\em The Annals of Applied Probability}, 33(3):1961--2003, 2023.

\bibitem{dobrushin_condition_robust_ising}
Ilias Diakonikolas, Daniel~M Kane, Alistair Stewart, and Yuxin Sun.
\newblock Outlier-robust learning of ising models under dobrushin’s
  condition.
\newblock In {\em Conference on Learning Theory}, pages 1645--1682. PMLR, 2021.

\bibitem{eichelsbacher2015rates}
Peter Eichelsbacher and Bastian Martschink.
\newblock On rates of convergence in the curie--weiss--potts model with an
  external field.
\newblock {\em Annales de l'Institut Henri Poincar{\'e} Probabilit{\'e}s et
  Statistiques}, 51(1):252--282, 2015.

\bibitem{sbcb}
Bruno~C{\'e}sar Feltes, Eduardo~Bassani Chandelier, Bruno~Iochins Grisci, and
  M{\'a}rcio Dorn.
\newblock Cumida: An extensively curated microarray database for benchmarking
  and testing of machine learning approaches in cancer research.
\newblock {\em Journal of Computational Biology}, 26(4):376--386, 2019.

\bibitem{fisher1922mathematical}
Ronald~A Fisher.
\newblock On the mathematical foundations of theoretical statistics.
\newblock {\em Philosophical transactions of the Royal Society of London.
  Series A, containing papers of a mathematical or physical character},
  222(594-604):309--368, 1922.

\bibitem{gandolfo2010limit}
Daniel Gandolfo, Jean Ruiz, and Marc Wouts.
\newblock Limit theorems and coexistence probabilities for the curie--weiss
  potts model with an external field.
\newblock {\em Stochastic Processes and their Applications}, 120(1):84--104,
  2010.

\bibitem{geman_graffinge}
Stuart Geman and Christine Graffigne.
\newblock Markov random field image models and their applications to computer
  vision.
\newblock In {\em Proceedings of the international congress of mathematicians},
  volume~1, page~2. Berkeley, CA, 1986.

\bibitem{joint}
Promit Ghosal and Sumit Mukherjee.
\newblock Joint estimation of parameters in ising model.
\newblock {\em The Annals of Statistics}, 48(2), 2020.

\bibitem{ghosh2015influence}
Abhik Ghosh.
\newblock Influence function analysis of the restricted minimum divergence
  estimators: A general form.
\newblock {\em Electronic Journal of Statistics}, 9:1017--1040, 2015.

\bibitem{ghosh2023robustness}
Abhik Ghosh.
\newblock Robustness concerns in high-dimensional data analyses and potential
  solutions.
\newblock In {\em Big Data Analytics in Chemoinformatics and Bioinformatics},
  pages 37--60. Elsevier, 2023.

\bibitem{ghosh2013robust}
Abhik Ghosh and Ayanendranath Basu.
\newblock Robust estimation for independent non-homogeneous observations using
  density power divergence with applications to linear regression.
\newblock {\em Electronic Journal of Statistics}, 7:2420--2456, 2013.

\bibitem{ghosh2016robust}
Abhik Ghosh and Ayanendranath Basu.
\newblock Robust estimation in generalized linear models: the density power
  divergence approach.
\newblock {\em Test}, 25(2):269--290, 2016.

\bibitem{ghosh2017robust}
Abhik Ghosh and Ayanendranath Basu.
\newblock Robust and efficient parameter estimation based on censored data with
  stochastic covariates.
\newblock {\em Statistics}, 51(4):801--823, 2017.

\bibitem{ghosh2024robust}
Abhik Ghosh, Mar{\'\i}a Jaenada, and Leandro Pardo.
\newblock Robust adaptive variable selection in ultra-high dimensional linear
  regression models.
\newblock {\em Journal of Statistical Computation and Simulation},
  94(3):571--603, 2024.

\bibitem{ghosh2020ultrahigh}
Abhik Ghosh and Subhabrata Majumdar.
\newblock Ultrahigh-dimensional robust and efficient sparse regression using
  non-concave penalized density power divergence.
\newblock {\em IEEE Transactions on Information Theory}, 66(12):7812--7827,
  2020.

\bibitem{glymour2019review}
Clark Glymour, Kun Zhang, and Peter Spirtes.
\newblock Review of causal discovery methods based on graphical models.
\newblock {\em Frontiers in genetics}, 10:524, 2019.

\bibitem{ising_models_independent_failures}
Surbhi Goel, Daniel~M Kane, and Adam~R Klivans.
\newblock Learning ising models with independent failures.
\newblock In {\em Conference on Learning Theory}, pages 1449--1469. PMLR, 2019.

\bibitem{disease}
Peter~J Green and Sylvia Richardson.
\newblock Hidden markov models and disease mapping.
\newblock {\em Journal of the American statistical association},
  97(460):1055--1070, 2002.

\bibitem{hampel1}
Frank~R. Hampel.
\newblock The influence curve and its role in robust estimation.
\newblock {\em Journal of the American Statistical Association},
  69(346):383--393, 1974.

\bibitem{neural}
John~J Hopfield.
\newblock Neural networks and physical systems with emergent collective
  computational abilities.
\newblock {\em Proceedings of the national academy of sciences},
  79(8):2554--2558, 1982.

\bibitem{election_spin}
Rub{\'e}n Ibarrondo, Mikel Sanz, and Rom{\'a}n Or{\'u}s.
\newblock Forecasting election polls with spin systems.
\newblock {\em SN Computer Science}, 3(1):44, 2022.

\bibitem{ising}
Ernst Ising.
\newblock Beitrag zur theorie des ferromagnetismus.
\newblock {\em Zeitschrift f{\"u}r Physik}, 31(1):253--258, 1925.

\bibitem{katiyar2020robust}
Ashish Katiyar, Vatsal Shah, and Constantine Caramanis.
\newblock Robust estimation of tree structured ising models.
\newblock {\em stat}, 1050:10, 2020.

\bibitem{snapnets}
Jure Leskovec and Andrej Krevl.
\newblock {SNAP Datasets}: {Stanford} large network dataset collection.
\newblock \url{http://snap.stanford.edu/data}, June 2014.

\bibitem{robust_learning_ising}
Erik~M Lindgren, Vatsal Shah, Yanyao Shen, Alexandros~G Dimakis, and Adam
  Klivans.
\newblock On robust learning of ising models.
\newblock In {\em NeurIPS Workshop on Relational Representation Learning},
  2019.

\bibitem{maronna2019robust}
Ricardo~A Maronna, R~Douglas Martin, Victor~J Yohai, and Mat{\'\i}as
  Salibi{\'a}n-Barrera.
\newblock {\em Robust statistics: theory and methods (with R)}.
\newblock John Wiley \& Sons, 2019.

\bibitem{innovations}
Andrea Montanari and Amin Saberi.
\newblock The spread of innovations in social networks.
\newblock {\em Proceedings of the National Academy of Sciences},
  107(47):20196--20201, 2010.

\bibitem{mukherjee_jmlr}
Somabha Mukherjee, Ziang Niu, Sagnik Halder, Bhaswar~B. Bhattacharya, and
  George Michailidis.
\newblock Logistic regression under network dependence.
\newblock {\em Journal of Machine Learning Research}, 25(411):1--62, 2024.

\bibitem{psm1}
Somabha Mukherjee, Jaesung Son, and Bhaswar~B Bhattacharya.
\newblock Fluctuations of the magnetization in the p-spin curie--weiss model.
\newblock {\em Communications in Mathematical Physics}, 387(2):681--728, 2021.

\bibitem{psm2}
Somabha Mukherjee, Jaesung Son, and Bhaswar~B Bhattacharya.
\newblock Estimation in tensor ising models.
\newblock {\em Information and Inference: A Journal of the IMA},
  11(4):1457--1500, 2022.

\bibitem{psm3}
Somabha Mukherjee, Jaesung Son, and Bhaswar~B Bhattacharya.
\newblock Phase transitions of the maximum likelihood estimators in the p-spin
  curie-weiss model.
\newblock {\em Bernoulli}, 31(2):1502--1526, 2025.

\bibitem{mukherjee2024efficient}
Somabha Mukherjee, Jaesung Son, Swarnadip Ghosh, and Sourav Mukherjee.
\newblock Efficient estimation in tensor curie-weiss and erdős-rényi ising
  models.
\newblock {\em Electronic Journal of Statistics}, 18(1):2405--2449, 2024.

\bibitem{hubers_contamination_ising}
Adarsh Prasad, Vishwak Srinivasan, Sivaraman Balakrishnan, and Pradeep
  Ravikumar.
\newblock On learning ising models under huber's contamination model.
\newblock {\em Advances in neural information processing systems},
  33:16327--16338, 2020.

\bibitem{allen}
Allen Institute~MindScope Program.
\newblock Allen brain observatory -- neuropixels visual coding [dataset].
\newblock \url{http://brain-map.org/explore/circuits}, 2019.

\bibitem{highdim_ising}
Pradeep Ravikumar, Martin~J Wainwright, and John~D Lafferty.
\newblock High-dimensional ising model selection using l1-regularized logistic
  regression.
\newblock {\em The Annals of Statistics}, 38(3):1287, 2010.

\bibitem{hampel2}
Peter~J Rousseeuw, Frank~R Hampel, Elvezio~M Ronchetti, and Werner~A Stahel.
\newblock {\em Robust statistics: the approach based on influence functions}.
\newblock John Wiley \& Sons Incorporated, 2005.

\bibitem{github}
Benedek Rozemberczki, Carl Allen, and Rik Sarkar.
\newblock Multi-scale attributed node embedding, 2019.

\bibitem{feather}
Benedek Rozemberczki and Rik Sarkar.
\newblock {Characteristic Functions on Graphs: Birds of a Feather, from
  Statistical Descriptors to Parametric Models}.
\newblock In {\em Proceedings of the 29th ACM International Conference on
  Information and Knowledge Management (CIKM '20)}, page 1325–1334. ACM,
  2020.

\bibitem{pc}
Peter Spirtes, Clark~N Glymour, and Richard Scheines.
\newblock {\em Causation, prediction, and search}.
\newblock MIT press, 2000.

\bibitem{vmlc16}
Marc Vuffray, Sidhant Misra, Andrey Lokhov, and Michael Chertkov.
\newblock Interaction screening: Efficient and sample-optimal learning of ising
  models.
\newblock {\em Advances in neural information processing systems}, 29, 2016.

\bibitem{vml20}
Marc Vuffray, Sidhant Misra, and Andrey~Y Lokhov.
\newblock Efficient learning of discrete graphical models.
\newblock {\em Journal of Statistical Mechanics: Theory and Experiment},
  2021(12):124017, 2022.

\bibitem{semi_supervised_ising}
Daiqing Wu and Molei Liu.
\newblock Robust and efficient semi-supervised learning for ising model.
\newblock {\em Biometrics}, 81(2):ujaf060, 2025.

\bibitem{wu1982potts}
F.~Y. Wu.
\newblock The potts model.
\newblock {\em Reviews of Modern Physics}, 54(1):235--268, 1982.

\end{thebibliography}

\appendix

\section{Verifying Assumptions \ref{asd0}--\ref{asd2}}\label{sec:struct}
Let us now show that Assumptions \ref{asd0}--\ref{asd2} hold for Ising models on interaction structures satisfying Conditions \ref{bounded_sum}--\ref{spectral_gap}. 
By Lemma 2.1(b) in \cite{meanfield}, under Conditions \ref{bounded_sum}, \ref{regular} and \ref{spectral_gap}, we have:
$$\p\left(\frac{1}{N}\sum_{i=1}^N(m_i(\bm X)  - M(\bm X))^2 > \varepsilon\right) \le \frac{\e e^{\frac{\delta}{2}\sum_{i=1}^N (m_i(\bm X) - M(\bm X))^2}}{e^{\frac{\delta}{2}N\varepsilon}} \lesssim e^{O\left(\|\bm J_N\|_F^2+\sum_{i=1}^N(R_i-1)^2\right) - \frac{\delta}{2}N\varepsilon}$$ for some $\delta>0$, where $R_i:= \sum_{j=1}^N J_{i,j}$. It now follows from Conditions \ref{regular} and \ref{mean_field} that $\|\bm J_N\|_F^2+\sum_{i=1}^N(R_i-1)^2 = o_P(N)$, and hence, $$\frac{1}{N}\sum_{i=1}^N(m_i(\bm X)  - M(\bm X))^2 \xrightarrow{P} 0.$$
The fact that $\bar{X} - M(\bm X) = o_P(1)$ follows immediately from the exponential concentration given by Corollary 1.1  in \cite{meanfield}, thereby verifying  Assumption \ref{asd0}. 

Next, define $\bm d = \bm d^{(N)} := (d_i)_{i=1}^N$, where $$d_i := c_i(\alpha_im_i(\bm X)+(1-\alpha_i) M(\bm X))\cdot (m_i(\bm X)-M(\bm X)).$$ Note that $N^{-1} \sum_{i=1}^N d_i^2 \xrightarrow{P} 0$, and
$N^{-1} \bm d^\top \bm J \bm d \le \lambda_1(\bm J_N) N^{-1} \|\bm d\|^2 \xrightarrow{P} 0.$ Hence, by Theorem 5.3 in \cite{pivotal_clt}, we have:
$$\frac{1}{\sqrt{N}}\sum_{i=1}^N  c_i(\alpha_im_i(\bm X)+(1-\alpha_i) M(\bm X))\cdot (m_i(\bm X)-M(\bm X)) (X_i-\tanh(\beta m_i(\bm X))) = o_P(1).$$ Moreover, if we take $d_i = c_i(m_i(\bm X))/\sqrt{N}$, then once again by Theorem 5.3 in \cite{pivotal_clt}, we have:
$$\frac{1}{N}\sum_{i=1}^N c_i(m_i(\bm X))(X_i-\tanh(\beta m_i(\bm X))) = \frac{1}{\sqrt{N}}\sum_{i=1}^N d_i (X_i-\tanh(\beta m_i(\bm X))) = o_P(1),$$
thereby completing the verification of Assumption \ref{asd1}.

For verifying Assumption \ref{asd2}, we once again have from Theorem 5.3 in \cite{pivotal_clt} (on taking their constants $c_i$ to be $1$):
$$\frac{1}{\sqrt{N}}\sum_{i=1}^N (X_i -\tanh(\beta m_i(\bm X))) \xrightarrow{D} \mathcal{N}(0,(1-m_+^2)(1-\beta(1-m_+^2))).$$ This shows that Assumption \ref{asd2} holds with $\sigma_\beta^2 = (1-m_+^2)(1-\beta(1-m_+^2))$.

\section{Proof of Lemma \ref{sx}}\label{prooflem1}
Without loss of generality, we can work under the assumption that $\|\bm J\|=1$. Define a function $F:\{-1,1\}^N \times \{-1,1\}^N \rightarrow \R$ as
\begin{equation*}
    F(\tau, \tau^\prime)=\frac{1}{2}\sum_{i=1}^{N}(f_i(\tau)m_i(\tau)+f_i(\tau^{\prime})m_i(\tau^\prime))(\tau_i - \tau_i^\prime),
\end{equation*}
where 
$
    f_i(\bm X) := f(\beta,m_i(\bm X)).
$
We notice that the function $F$ is anti-symmetric, i.e. $F(\tau,\tau^\prime)=-F(\tau^\prime, \tau)$. 

Suppose now, that $\bm X$ is simulated from the distribution \eqref{modeldef}. Choose a coordinate $I$ uniformly at random from the set $[N] := \{1,\ldots,N\}$, and replace the $I^{\mathrm{th}}$ coordinate of $\bm X$ by a sample drawn from the conditional distribution of $X_{I}$ given $(X_{j})_{j\neq I}$ . Call the resulting vector $\bm X^\prime$. Then $(\bm X, \bm X^\prime)$ is an exchangeable pair of random variables. We notice that
$$
F(\bm X,\bm X^{\prime}) = f_{I}(\bm X) m_I(\bm X)(X_I-X_I^\prime).
$$
Note that:
\begin{equation*}
\begin{aligned}
    g(\bm X) := & \E[F(\bm X,\bm X^\prime)|\bm X]\\
    = & \frac{1}{N}\sum_{i=1}^N f_i(\bm X) m_i(\bm X)(X_i-\E[X_i|(X_j)_{j\neq i}])\\
    = & \frac{1}{N}\sum_{i=1}^N f_i(\bm X) m_i(\bm X)(X_i-\tanh(\beta m_i(\bm X)))\\
    = & S_{\bm X,\lambda}(\beta).
\end{aligned}
\end{equation*}
Now, $\E[g(\bm X)^2]=\E[g(\bm X)F(\bm X,\bm X^\prime)]$. Since $(\bm X,\bm X^\prime)$ is an exchangeable pair,
\begin{equation*}
    \E[g(\bm X^\prime)F(\bm X^\prime,\bm X)]=\E[g(\bm X)F(\bm X,\bm X^\prime)].
\end{equation*}
It now follows from the anti-symmetry of $F$, that
\begin{equation}\label{exp_square}
    \E[g(\bm X)^2]=\frac{1}{2}\E[(g(\bm X)-g(\bm X^\prime))F(\bm X,\bm X^\prime)].
\end{equation}

For any $1\leq j \leq N$ and $\tau \in \{-1,1\}^N$, let $\tau^{(j)}:=(\tau_1,...,\tau_{j-1}, -\tau_j,\tau_{j+1},...,\tau_N)$ and 
\begin{equation*}
\begin{aligned}
    p_j(\tau) = & \p(X_j^\prime = -\tau_j|\bm X=\tau,I=j)\\
    = & \frac{e^{-\beta\tau_j m_j(\tau)}}{e^{\beta m_j(\tau)}+e^{-\beta m_j(\tau)}}
\end{aligned}
\end{equation*}
Then
\begin{equation*}
    \begin{aligned}
        & \E[(g(\bm X)-g(\bm X^\prime))F(\bm X,\bm X^\prime)|\bm X]\\
        = & \frac{1}{N}\sum_{i=1}^N (g(\bm X)-g(\bm X^{(j)}))F(\bm X,\bm X^{(j)})p_j(\bm X)\\
        = & \frac{2}{N}\sum_{i=1}^N (g(\bm X)-g(\bm X^{(j)})) f_j(\bm X)m_j(\bm X)X_j p_j(\bm X).
    \end{aligned}
\end{equation*}
In order to simplify notations, let us define:
\begin{equation*}
    a_i(\tau) := \tau_i-\tanh(\beta m_i(\tau)),
\end{equation*}
\begin{equation*}
    b_{ij}(\tau):= \tanh(\beta m_i(\tau))-\tanh(\beta m_i(\tau^{(j)})),
\end{equation*}
\begin{equation*}
    c_{ij}(\tau):= f_i(\tau)-f_i(\tau^{(j)}).
\end{equation*}
Then,
\begin{equation*}
\begin{aligned}
    & g(\bm X)-g(\bm X^{(j)})\\
    = & \frac{1}{N}\sum_{i=1}^N f_i(\bm X)m_i(\bm X)(X_i-\tanh(\beta m_i(\bm X)))-\frac{1}{N}\sum_{i=1}^N f_i(\bm X^{(j)})m_i(\bm X^{(j)})(X^{(j)}_i-\tanh(\beta m_i(\bm X^{(j)})))\\
    = & \frac{1}{N}\sum_{i=1}^N(m_i(\bm X)-m_i(\bm X^{(j)}))f_i(\bm X)a_i(\bm X)+\frac{1}{N}\sum_{i=1}^N m_i(\bm X^{(j)})f_i(\bm X)(a_i(\bm X)-a_i(\bm X^{(j)}))\\
    & +\frac{1}{N}\sum_{i=1}^N m_i(\bm X^{(j)})a_i(\bm X)(f_i(\bm X)-f_i(\bm X^{(j)}))\\ &+\frac{1}{N}\sum_{i=1}^N m_i(\bm X^{(j)})(a_i(\bm X^{(j)})-a_i(\bm X))(f_i(\bm X)-f_i(\bm X^{(j)}))\\
    = & \frac{2 X_j}{N}\sum_{i=1}^N J_{ij}f_i(\bm X)a_i(\bm X) + \frac{2f_j(\bm X^{(j)})\cdot m_j(\bm X) X_j}{N}-\frac{1}{N}\sum_{i=1}^N f_i(\bm X)m_i(\bm X^{(j)})b_{ij}(\bm X) \\
    & +\frac{1}{N}\sum_{i=1}^N m_i(X^{(j)})a_i(\bm X)c_{ij}(\bm X)+ \frac{1}{N}\sum_{i=1}^N m_i(\bm X^{(j)})b_{ij}(\bm X)c_{ij}(\bm X)\\
\end{aligned}
\end{equation*}
Let $T_{1j}$, $T_{2j}$, $T_{3j}$, $T_{4j}$ and $T_{5j}$ be the five terms in the last expression, in that order. We can rewrite \ref{exp_square} as:
\begin{equation}\label{fivedec}
    \E[g(\bm X)^2] = \frac{1}{N}\sum_{j=1}^N (T_{1j}+T_{2j}+T_{3j}+T_{4j}+T_{5j})f_j(\bm X)m_j(\bm X)X_j p_j(\bm X).
\end{equation}

To begin with, note that by Lemma \ref{f_bound}, $0 < f_i(\bm X)\leq 2^{\lambda}$. This, together with the fact that $|a_i| \le 2$, implies that $\sum_{i=1}^N f_i(\bm X)^2 a_i^2 \leq 4^{\lambda+1}N$ and
\begin{equation}\label{fmp}
    \sum_{j=1}^N f_j(\bm X)^2 m_j(\bm X)^2 p_j(\bm X)^2\leq 4^{\lambda} \sum_{j=1}^N m_j(\bm X)^2 = 4^{\lambda} \|\bm J \bm X\|^2 \leq 4^{\lambda}N.
\end{equation}
Then, for the first term of the sum in the RHS of \eqref{fivedec}, we have:
\begin{eqnarray}\label{combeq21}
    \bigg|\frac{1}{N}\sum_{j=1}^N T_{1j}f_j(\bm X)m_j(\bm X)X_j p_j(\bm X)\bigg| &=&  \frac{2}{N^2}\bigg|\sum_{i,j=1}^N f_i(\bm X)a_i(\bm X)J_{ij} f_j(\bm X)m_j(\bm X) p_j(\bm X)\bigg|\nonumber\\
    & \leq& \frac{2}{N^2} \sqrt{4^{\lambda+1}N}\sqrt{4^{\lambda}N} = \frac{2^{2\lambda+2}}{N}.
\end{eqnarray}
Next, for the second term in the RHS of \eqref{fivedec}, we have:
\begin{eqnarray}\label{combeq22}
    \left|\frac{1}{N}\sum_{j=1}^N T_{2j}f_j(\bm X)X_jm_j(\bm X) p_j(\bm X)\right| & =& \left|\frac{2}{N^2}\sum_{j=1}^N  f_j(\bm X^{(j)})f_j(\bm X) m_j(\bm X)^2 p_j(\bm X)\right|\nonumber\\
    & \leq&  \frac{2\cdot 4^\lambda}{N^2}\sum_{j=1}^N m_j(\bm X)^2 \leq \frac{2^{2\lambda+1}}{N}.
\end{eqnarray}

Let us now bound the third term in the RHS of \eqref{fivedec}. If we let $\bm J_2$ be the matrix $(J_{ij}^2)_{1\leq i,j\leq N}$, then we have:
$$
    \|\bm J_2\|\leq \max_{1\leq i\leq N}\sum_{j=1}^N J_{ij}^2 \leq \max_{1\leq i\leq N}  \|\bm e_i^\top \bm J\|^2\leq \max_{1\leq i\leq N} \|\bm e_i\|^2\|\bm J\|^2 = 1
$$
where $\bm e_i$ denotes the unit vector with $1$ in the $i^{\mathrm{th}}$ coordinate and zeros elsewhere.
Now, letting $h(x) = \tanh(\beta x)$, it is not difficult to verify that $\|h^{\prime\prime}\|_{\infty}\leq 2\beta^2$. Thus,
$$
    |h(m_i(\bm X))-h(m_i(\bm X^{(j)}))-(m_i(\bm X)-m_i(\bm X^{(j)}))h^{\prime}(m_i(\bm X))|\leq \beta^2 (m_i(\bm X)-m_i(\bm X^{(j)}))^2
$$
Letting $d_i(\bm X) := h^{\prime}(m_i(\bm X))$, and noting that $m_i(\bm X)-m_i(\bm X^{(j)})=2J_{ij}X_j$, we can rewrite the above inequality as:
\begin{equation*}
    |b_{ij}(\bm X)-2J_{ij}X_j d_{i}(\bm X)|\leq 4\beta^2 J_{ij}^2.
\end{equation*}
Also, we note that $|d_i(\bm X)|\leq \beta$. With all the bounds above, we see that for any $\bm u,\bm v\in \mathbb{R}^N$,
\begin{equation}\label{xy}
\begin{aligned}
    \bigg|\sum_{i,j}u_i v_j b_{ij}(\bm X) \bigg| & \leq \bigg|\sum_{i,j}u_i v_j(2J_{ij}X_j d_i(\bm X))) \bigg|+ \bigg|\sum_{i,j}u_i v_j (b_{ij}(\bm X)-2J_{ij}X_j d_i(\bm X)) \bigg|\\
    & \leq 2\bigg(\sum_{i}(u_i d_i(\bm X))^2 \bigg)^{\frac{1}{2}}\bigg(\sum_{j}(v_j X_j)^2 \bigg)^{\frac{1}{2}} + 4\sum_{i,j}|u_i v_j|\beta^2 J_{ij}^2\\
    & \leq (2\beta + 4\beta^2)\|\bm u\|\|\bm v\|.
\end{aligned}
\end{equation}
Since $\tanh$ is $1$-Lipschitz,  $|b_{ij}(\bm X)| \leq 2\beta|J_{ij} |$. Hence, we have:
\begin{equation}\label{Jxy}
    \bigg|\sum_{i,j}u_i v_j J_{ij} b_{ij}(\bm X) \bigg|\leq \sum_{ij}|u_i v_j|2\beta J_{ij}^2 \leq 2\beta \|\bm u\|\|\bm v\|.
\end{equation}
With \eqref{xy} and \eqref{Jxy}, we can bound the $T_{3j}$ term as follows:
\begin{eqnarray}\label{combeq23}
    &&\bigg|\frac{1}{N}\sum_{j=1}^N T_{3j}f_j(\bm X)m_j(\bm X)X_j p_j(\bm X)\bigg|\nonumber\\ &=&  \frac{1}{N^2}\bigg|\sum_{i,j=1}^N f_i(\bm X)m_i(\bm X^{(j)})b_{ij}(\bm X) f_j(\bm X)m_j(\bm X) X_j p_j(\bm X)\bigg|\nonumber\\
    &=&  \frac{1}{N^2}\bigg|\sum_{i,j}^N f_i(\bm X)(m_i(\bm X)-2J_{ij}X_j)b_{ij}(\bm X) f_j(\bm X)m_j(\bm X) X_j p_j(\bm X)\bigg|\nonumber\\
    &\leq& \textcolor{black}{\frac{4^\lambda}{N^2}(2\beta +4\beta^2 + 4\beta)\max \left\{1,\sqrt{\sum_{i=1}^N m_i(\bm X)^2}\right\} \sqrt{ \sum_{j=1}^N m_j(\bm X)^2 p_j^2}}\nonumber\\
    &\leq& \frac{2^{2\lambda+1}(2\beta^2+3\beta)}{N}.
\end{eqnarray}
Next, we bound the fourth term in the RHS of \eqref{fivedec}. Denote $f_i(\bm X) = f(\beta, m_i(\bm X))$ by $q(m_i(\bm X))$. Then, in view of Lemma \ref{f_bound} we have:
$$
    |q(m_i(\bm X))-q(m_i(\bm X^{(j)}))-(m_i(\bm X)-m_i(\bm X^{(j)}))q^{\prime}(m_i(\bm X))|\leq D_{\lambda} \beta^2 (m_i(\bm X)-m_i(\bm X^{(j)}))^2.
$$
for some constant $D_\lambda$ depending only on $\lambda$. Letting $t_i(\bm X) := q^{\prime}(m_i(\bm X))$, and noting that $m_i(\bm X)-m_i(\bm X^{(j)})=2J_{ij}X_j$, we can rewrite the above inequality as:
\begin{equation*}
    |c_{ij}(\bm X)-2J_{ij}X_j t_{i}(\bm X)|\leq 4 D_{\lambda} \beta^2 J_{ij}^2.
\end{equation*}
Also we note that $|t_i(\bm X)|\leq C_\lambda  \beta$ for some constant $C_\lambda$ depending only on $\lambda$. With all the bounds above and the bounds on the operator norms of $\bm J$ and $\bm J_2$, we see that for any $\bm u, \bm v\in \R^N$,
\begin{equation}\label{cxy}
\begin{aligned}
    \bigg|\sum_{i,j}u_i v_j c_{ij}(\bm X) \bigg| & \leq \bigg|\sum_{i,j}u_i v_j(2J_{ij}X_j t_i(\bm X))) \bigg|+ \bigg|\sum_{i,j}u_i v_j (c_{ij}(\bm X)-2J_{ij}X_j t_i(\bm X)) \bigg|\\
    & \leq 2\bigg(\sum_{i}(u_i t_{i}(\bm X))^2 \bigg)^{\frac{1}{2}}\bigg(\sum_{j}(v_j X_j)^2 \bigg)^{\frac{1}{2}} + 4\sum_{i,j}|u_i v_j|D_{\lambda}  \beta^2 J_{ij}^2\\
    & \leq 2(C_\lambda \beta+ 2 D_\lambda\beta^2))\|\bm u\|\|\bm v\|.
\end{aligned}
\end{equation}
Also, from Lemma \ref{f_bound}, we have $|c_{ij}(\bm X)| \leq  2C_\lambda\beta|J_{ij} |$. Thus, we have:
\begin{equation}\label{Jcxy}
    \bigg|\sum_{i,j}u_i v_j J_{ij} c_{ij}(\bm X) \bigg|\leq 2C_\lambda \beta\sum_{i,j}|u_i v_j| J_{ij}^2 \leq 2C_\lambda\beta \|\bm u\|\|\bm v\|.
\end{equation}

With \ref{fmp}, \ref{cxy} and \ref{Jcxy}, and noticing that $\sum a_i(\bm X)^2 m_i(\bm X)^2 \leq 4N$, we get:
\begin{equation}\label{combeq24}
\begin{aligned}
    &\bigg|\frac{1}{N}\sum_{j=1}^N T_{4j}f_j(\bm X)m_j(\bm X)X_j p_j(\bm X)\bigg| \\
    =\ & \frac{1}{N^2}\bigg|\sum_{i,j=1}^N m_i(\bm X^{(j)})a_i(\bm X)c_{ij}(\bm X)f_j(\bm X)m_j(\bm X) X_j p_j(\bm X)\bigg|\\
    =\ &  \frac{1}{N^2}\bigg|\sum_{i,j=1}^N a_i(\bm X)(m_i(\bm X)-2J_{ij}X_j)c_{ij}(\bm X) f_j(\bm X)m_j(\bm X) X_j p_j(\bm X)\bigg|\\
    \leq \ & \frac{K_{1,\lambda,\beta}}{N^2}\sqrt{\sum_{i=1}^N a_i(\bm X)^2m_i(\bm X)^2} \sqrt{ \sum_{j=1}^N f_j(\bm X)^2m_j(\bm X)^2 p_j(\bm X)^2}\\
    \leq \ & \frac{K_{\lambda,\beta}}{N}
\end{aligned}
\end{equation}
for some constants $K_{1,\lambda,\beta}$ and $K_{\lambda,\beta}$ depending only on $\lambda$ and $\beta$.

Finally we bound the $T_{5j}$ term. 
Since for any $1 \leq i \leq N$, $\sum_{j=1}^N J_{ij}^2\leq 1$, we have:
\begin{equation}\label{J31}
    \sum_{j=1}^N |J_{ij}|^3 \leq \sum_{j=1}^N J_{ij}^2\leq 1.
\end{equation}
We define the matrix $\bm J_3:=(J_{ij}^3)_{1\leq i,j\leq N}$. Note that:
\begin{equation}\label{J32}
    \|\bm J_3\|\leq \max_{1\leq i\leq N}\sum_{j=1}^N |J_{ij}|^3 \leq 1.
\end{equation}

Also, we have $|b_{ij}(\bm X)|\leq 2\beta |J_{ij}|$ and $|c_{ij}(\bm X)| \leq L_{\lambda}\beta|J_{ij} |$ for some constant $L_\lambda$ depending only on $\lambda$. Together with the bound on the operator norm of $\bm J_3$, we get:
\begin{equation}\label{bigeq51}
\begin{aligned}
    &\ \bigg|\sum_{i,j=1}^N J_{ij}X_j b_{ij}(\bm X)c_{ij}(\bm X)f_j(\bm X)m_j(\bm X) X_j p_j(\bm X)\bigg|\\
    \leq &\ 2L_\lambda \beta^2\sum_{i,j=1}^N\bigg|J_{ij}^3 f_j(\bm X)m_j(\bm X) p_j(\bm X) \bigg|\\
    \leq &\ 2L_\lambda \beta^2 \|\bm J_3\| \sqrt{N}\sqrt{\sum_{j=1}^N f_j^2(\bm X) m_j^2(\bm X) p_j^2(\bm X)}\\ \leq & 2^{\lambda+1} L_\lambda \beta^2 N.
\end{aligned}
\end{equation}

Similarly, using the bound on the operator norm of $\bm J_2$, we have:
\begin{equation}\label{bigeq52}
\begin{aligned}
    &\ \bigg|\sum_{i,j=1}^N m_i(\bm X) b_{ij}(\bm X)c_{ij}(\bm X)f_j(\bm X)m_j(\bm X) X_j p_j(\bm X)\bigg|\\
    \leq &\ 2L_\lambda \beta^2\sum_{i,j=1}^N \bigg|m_i(\bm X)f_j(\bm X)m_j(\bm X) X_j p_j(\bm X)\bigg| J_{ij}^2\\
    \leq &\ 2^{\lambda+1} L_\lambda \beta^2 N.
\end{aligned}
\end{equation}
Combining \eqref{bigeq51} and \eqref{bigeq52}, we have:
\begin{equation}\label{combeq25}
\begin{aligned}
    &\bigg|\frac{1}{N}\sum_{j=1}^N T_{5j}f_j(\bm X)m_j(\bm X)X_j p_j(\bm X)\bigg| \\
    = \ &  \frac{1}{N^2}\bigg|\sum_{i,j=1}^N m_i(\bm X^{(j)})b_{ij}(\bm X)c_{ij}(\bm X)f_j(\bm X)m_j(\bm X) X_j p_j(\bm X)\bigg|\\
    = \ &  \frac{1}{N^2}\bigg|\sum_{i,j=1}^N (m_i(\bm X)-2J_{ij}X_j)b_{ij}(\bm X)c_{ij}(\bm X) f_j(\bm X)m_j(\bm X) X_j p_j(\bm X)\bigg|\\
    \leq \ & \frac{2^{\lambda+3} L_\lambda \beta^2}{N}.
\end{aligned}
\end{equation}

Combining \eqref{combeq21}, \eqref{combeq22}, \eqref{combeq23}, \eqref{combeq24} and \eqref{combeq25}, we have:
\begin{eqnarray}
    & &\E[g(X)^2] \nonumber\\
    &=&  \frac{1}{N}\sum_{j=1}^N (T_{1j}+T_{2j}+T_{3j}+T_{4j}+T_{5j})f_j(X)m_j(X)X_j p_j(X)\nonumber\\
    &= & O_{\lambda,\beta}\left(\frac{1}{N}\right).
\end{eqnarray}
This completes the proof of Lemma \ref{sx}. \qed

\section{Proof of Lemma \ref{sbeta1}}\label{prooflem2}
Note that both $S_{\bm X,\lambda}(\beta)$ and $S_{\bm X,\lambda}'(\beta)_1$ are of the following form:
$$
S_{\bm X,\lambda}^{\prime}(\beta)_1=\frac{1}{N}\sum_{i=1}^N \bigg[\phi(\beta,m_i(\bm X)) m_i(\bm X)(X_i - \tanh(\beta m_i(\bm X)))\bigg].
$$
for $\phi(\beta,t) := f(\beta,t)$ and $\partial f(\beta, t)/\partial \beta$, respectively. The only property of the function $\phi$ that was required for the proof of Lemma \ref{sx} to go through, was the boundedness of $\phi, \partial \phi/\partial x$ and $\partial^2\phi/\partial x^2$. Hence, to prove Lemma \ref{sbeta1}, we need the same three properties for the new function $$\phi(\beta,t) := \frac{\partial f(\beta, t)}{\partial \beta} = \frac{t}{\beta} \frac{\partial f(\beta, t)}{\partial \beta}\quad \text{(see \eqref{symmbt})},$$ which are shown in Lemma \ref{f_bound2}. This proves Lemma \ref{sbeta1}. \qed

\section{Additional Technical Lemmas}
\begin{lem}\label{f_bound}
    For any $\beta>0$ and $\lambda>0$, we have
    \begin{enumerate}[label=(\roman*)]
        \item $0<f(\beta,x) \leq 2^\lambda$, 
        \item $|\partial f(\beta,x)/\partial x|\leq C_\lambda \beta$, and for any $h \in \R$,
        \item $\left|\frac{\partial^2 f(\beta,x)}{\partial x^2}\right|\leq D_\lambda \beta^2$.
    \end{enumerate}
    where $C_{\lambda}$ and $D_\lambda$ are constants dependent only on $\lambda$. 
\end{lem}
\begin{proof}
To ease notations, let us fix $\beta>0$, and denote $f(\beta,x)$  by $q(x)$. That $q(x)>0$, is trivial. We notice that
\begin{equation}\label{qxexprc}
    q(x)=  \frac{\cosh((\lambda-1)\beta x)}{\cosh^{\lambda+1}(\beta x)} \\
    =  \frac{\cosh(\lambda \beta x)}{\cosh^{\lambda}(\beta x)}- \frac{\sinh(\lambda \beta x)}{\cosh^{\lambda}(\beta x)}\cdot \tanh(\beta x).
\end{equation}
Now, we have the following: 
\begin{equation*}
    \frac{\cosh(\lambda \beta x)}{\cosh^{\lambda}(\beta x)}  = \frac{e^{\lambda \beta x}+e^{-\lambda \beta x}}{(e^{ \beta x}+e^{- \beta x})^\lambda}\cdot 2^{\lambda-1}\\
     \leq \frac{e^{\lambda \beta x}+e^{-\lambda \beta x}}{e^{\lambda \beta x}+e^{-\lambda \beta x}}\cdot 2^{\lambda}\\
     = 2^{\lambda},
\end{equation*}
Since $q(x)$ is non-negative, and the second term in \eqref{qxexprc} is an even function, part (i) of Lemma \ref{f_bound} follows. An elementary calculation gives:
\begin{equation}\label{coshbound}
    q^\prime(x) =  \frac{\beta\lambda \sinh((\lambda-2)\beta x)-\beta\sinh(\lambda \beta x)}{\cosh^{\lambda+2}(\beta x)}
\end{equation}
\textcolor{black}{The numerator of \eqref{coshbound} can be upper-bounded by:
\[
\left| \lambda \sinh((\lambda - 2)\beta x) - \sinh(\lambda \beta x) \right|
\leq \lambda e^{|(\lambda - 2)\beta x|} + e^{\lambda \beta |x|} \le \lambda e^{(\lambda +2)\beta |x|} + e^{\lambda \beta |x|}.
\]
On the other hand, the denominator of \eqref{coshbound} can be lower-bounded by:
\[
\cosh^{\lambda + 2}(\beta x) \geq \left( \frac{e^{|\beta x|}}{2} \right)^{\lambda + 2} = 2^{-(\lambda + 2)} e^{(\lambda + 2)\beta |x|}.
\]
Combining these two bounds, we get:
\[
|q'(x)| \le \beta 2^{\lambda+2}\left(\lambda + e^{-2\beta|x|}\right) \le \beta 2^{\lambda+2}\left(\lambda + 1\right),
\]}
which proves part (ii).
\textcolor{black}{Another straightforward but tedious calculation shows: 
\begin{equation*}
\begin{aligned}
    q^{\prime\prime}(x) = & \beta^2\bigg(\frac{\lambda(\lambda-3)\cosh(\lambda \beta x)}{\cosh^{\lambda+2}(\beta x)}\\
    & + \frac{\sinh(\lambda \beta x)}{\cosh^{\lambda}(\beta x)}\cdot \tanh(\beta x)(2+3\lambda-3\lambda^2+(\lambda+1)(\lambda+2)\tanh^2(\beta x)) \bigg)\\
\end{aligned}
\end{equation*}}
and hence, we have:
\begin{equation*}
    |q^{\prime\prime}(x)|\leq \beta^2 2^{\lambda}(|\lambda(\lambda-3)|+|2+3\lambda-3\lambda^2|+(\lambda+1)(\lambda+2)).
\end{equation*}
Part (iii) now follows from Taylor's theorem, which completes the proof of Lemma \ref{f_bound}.
\end{proof}

\begin{lem}\label{f_bound2}
    Define $p(x):=x q^\prime(x)$ where $q(x) := f(\beta,x)$. Then, for any $\beta>0$ and $\lambda>0$, there exist constants $U_\lambda, V_\lambda$ and $W_\lambda$, such that for all $x$,
    $$|p(x)| \leq U_\lambda, ~|p^{\prime}(x)|  \leq V_\lambda, ~\text{and}~|p^{\prime\prime}(x)| \leq W_\lambda.$$
\end{lem}
\begin{proof}
From \eqref{coshbound}, we have
\begin{equation*}
\begin{aligned}
    p(x) = xq^\prime(x) = & \frac{\lambda \sinh((\lambda-2)\beta x)-\sinh(\lambda \beta x)}{\cosh^{\lambda+2}(\beta x)}\cdot\beta x\\
    = & \bigg(\frac{2\lambda\sinh(\lambda \beta x)}{\cosh^{\lambda}(\beta x)}\left(1-\frac{\tanh(\beta x)}{\tanh(\lambda \beta x)}\right)-\frac{(\lambda+1)\sinh(\lambda \beta x)}{\cosh^{\lambda}(\beta x)}\sech^2(\beta x)\bigg) \beta x\\
    = & \frac{\sinh(\lambda\beta x)}{\cosh^{\lambda}(\beta x)}\bigg(2\lambda\left(1-\frac{\tanh(\beta x)}{\tanh(\lambda \beta x)}\right) -(\lambda+1)\sech^2(\beta x)\bigg) \beta x.
\end{aligned}
\end{equation*}
Let us define:
$$
p_1(z) := \frac{\sinh(\lambda  z)}{\cosh^{\lambda}(z)},\quad p_2(z):=2\lambda z\left(1-\frac{\tanh z}{\tanh(\lambda z)}\right),\quad \text{and}\quad p_3(z):= (\lambda+1) z \sech^2(z),
$$
Clearly, we have:
\begin{equation}\label{sinh}
\begin{aligned}
    |p_1(z)|=\bigg|\frac{\sinh(\lambda z)}{\cosh^{\lambda}(z)}\bigg| & = \bigg|\frac{e^{\lambda z}-e^{-\lambda z}}{(e^{ z}+e^{- z})^\lambda}\cdot 2^{\lambda-1}\bigg|\\
    & \leq \frac{e^{\lambda z}+e^{-\lambda z}}{e^{\lambda z}+e^{-\lambda z}}\cdot 2^{\lambda}\\
    & = 2^{\lambda}.
\end{aligned}
\end{equation}
We may also notice that
\begin{equation*}
    |p_1^\prime(z)|=\lambda \bigg| \frac{\cosh(\lambda z)}{\cosh^{\lambda}(z)}-\frac{\sinh(\lambda z)}{\cosh^{\lambda}(z)}\tanh(z)  \bigg| = \lambda|q(z/\beta)|\leq \textcolor{black}{\lambda} 2^{\lambda}.
\end{equation*}
It follows from the bound of $|q'(x)|$ in the proof of Lemma \ref{f_bound}, that: $$|p_1^{\prime\prime}(z)|=\lambda\left|\frac{1}{\beta}q^\prime(z/\beta)\right|\leq \textcolor{black}{\lambda (\lambda+1) 2^{\lambda+2}}~.$$

Next, note that:
\begin{equation*}
\begin{aligned}
    p_2(z) & = 2\lambda z \left(1-\frac{e^{2z}-1}{e^{2z}+1}\cdot \frac{e^{2\lambda z}+1}{e^{2\lambda z}-1}\right)\\
    & = \frac{4\lambda z\cdot (e^{2\lambda z}-e^{2z})}{(e^{2z}+1)(e^{2\lambda z}-1)}\\
    & = 4\lambda \bigg(\frac{z}{e^{2z}+1}-\frac{z}{e^{2\lambda z}-1}\cdot\tanh(z)\bigg).
\end{aligned}
\end{equation*}
First, suppose that $z\ge 0$. Since $e^{z}\geq 1+z$ we have:
$$
     \frac{z}{e^{2z}+1}\leq \frac{z}{2z+2}\leq \frac{1}{2}.
$$
It also follows from straightforward algebra, that:
$$
   \textcolor{black}{\left|\left(\frac{z}{e^{2z}+1}\right)^\prime\right|  = \left|\frac{1}{2} +\frac{1-e^{4z}-4ze^{2z}}{2(e^{2z}+1)^2}\right| \le A}
   $$
    and
    $$
    \textcolor{black}{\left|\left(\frac{z}{e^{2z}+1}\right)^{\prime\prime}\right| = \left|\frac{4 e^{2z} \left[ (1 + z)(e^{2z} + 1) - 2z e^{2z} \right]}{(e^{2z} + 1)^3}\right| \le B}$$
for some finite, positive constants $A$ and $B$. 
Hence,
$$    
     \left|\frac{z}{e^{2\lambda z}-1}\right|\leq \frac{1}{2\lambda}, \quad \left|\left(\frac{z}{e^{2\lambda z}-1}\right)^{\prime}\right|\leq \frac{A}{\lambda}\quad \text{and}\quad \left|\left(\frac{z}{e^{2\lambda z}-1}\right)^{\prime\prime}\right|\leq \frac{B}{\lambda}~.
$$
Since $\sech^2(z) \in (0,1)$ and $\tanh(z) \in (-1,1)$, we can also conclude that $p_2$ and its first two derivatives are bounded by some constant (depending on $\lambda$).

We may also notice that
\begin{equation*}
\begin{aligned}
    |p_3(z)|=|(\lambda+1)\sech^2(z)\cdot z| & = \frac{4(\lambda+1)|z|}{(e^{z}+e^{-z})^2}\\
    & \leq \frac{4(\lambda+1)|z|}{e^{2|z|}}\\
    & \leq \frac{4(\lambda+1)|z|}{2|z|+1}\\
    & \leq 2(\lambda+1),
\end{aligned}
\end{equation*}
and
\begin{equation*}
    |p_3^{\prime}(z)|=|(\lambda+1)(\sech^2(z)-2z\cdot\sech^2(z)\tanh(z))|\leq 5(\lambda+1),\quad  |p_3^{\prime\prime}(z)|\leq 16(\lambda+1).
\end{equation*}
This completes the proof of Lemma \ref{f_bound2}.
\end{proof}

\begin{lem}\label{prokh}
Suppose that $M, u > 0$ are fixed, and $u_N$ is a sequence of non-negative real numbers, satisfying:
$$\sqrt{N}|e^{-u_N}- e^{-u}| \le M$$ for some $N\ge M^2/(e^{-u}-e^{-2u})^2$. Then, $$|u_n-u| \le \frac{Me^{2u}}{\sqrt{N}}~.$$
\end{lem}
\begin{proof}
    First, note that $$e^{-u_N} \ge e^{-u} - \frac{M}{\sqrt{N}} \ge e^{-2u}~.$$
    Hence, we have for some sequence $\xi_N \ge e^{-2u}$,
    $$\frac{M}{\sqrt{N}} \ge |e^{-u_N}- e^{-u}| = |u_N-u|e^{-\xi_N} \ge |u_N-u| e^{-2u}~,$$
    which completes the proof of Lemma \ref{prokh}.
\end{proof}

\end{document}